\documentclass[runningheads]{llncs}

\usepackage{style}

\usepackage{mathrsfs}
\usepackage{amsmath}
\usepackage{amssymb}

\usepackage{color}
\usepackage{xcolor}

\usepackage{graphicx}

\usepackage{subcaption}
\usepackage[draft,backgroundcolor=white!50!orange]{todonotes}

\usepackage{hyperref}
\usepackage{url}
\newcommand{\doi}[1]{\textsc{doi}: \href{http://dx.doi.org/#1}{\nolinkurl{#1}}}

\usepackage{longtable, array}
\usepackage{arydshln}

\usepackage[firstpage]{draftwatermark}
\SetWatermarkText{\hspace*{125mm}\raisebox{132mm}{\includegraphics[scale=0.12]{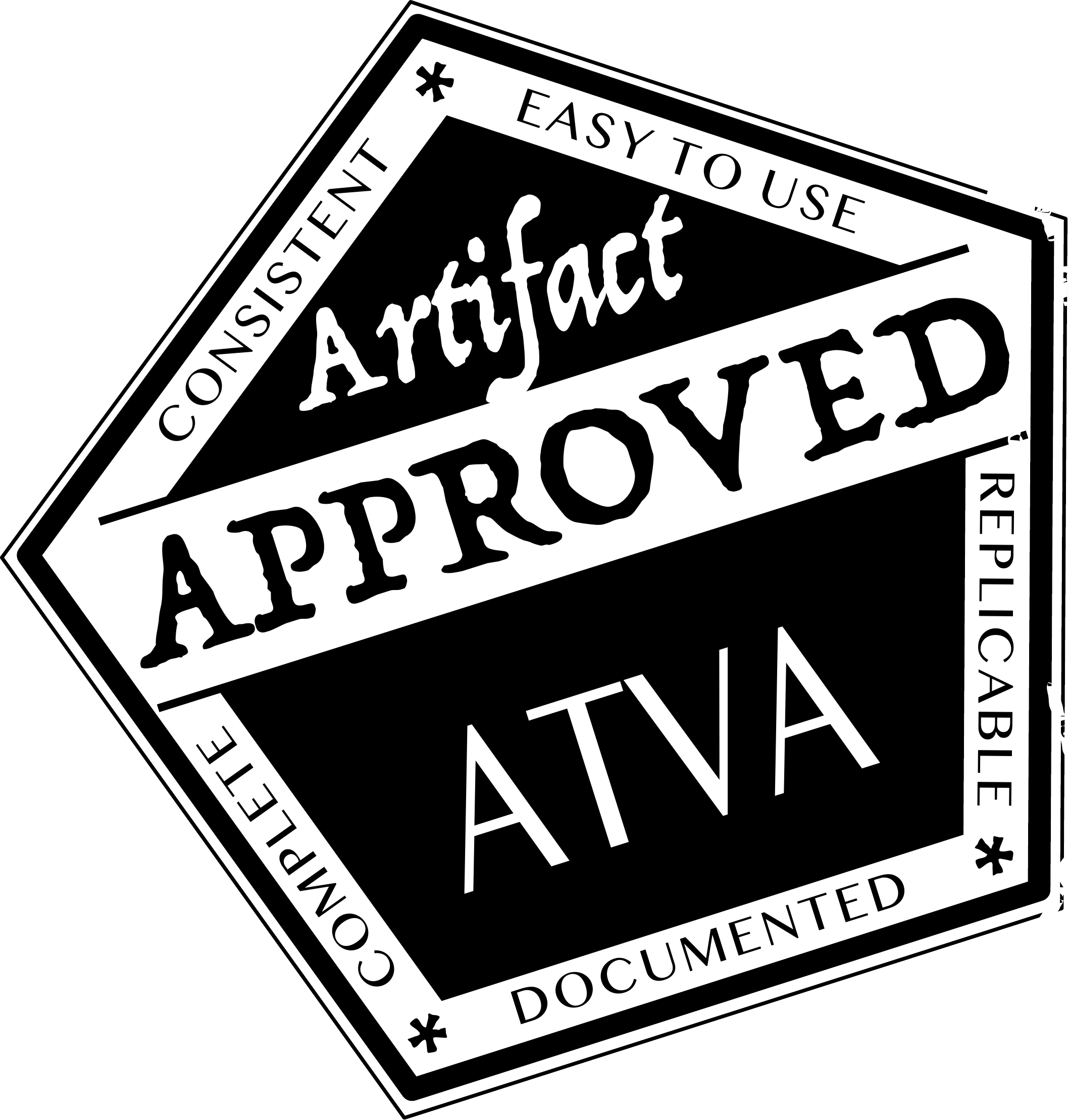}}}
\SetWatermarkAngle{0}

\begin{document}
\title{Model Checking Data Flows in Concurrent Network Updates (Full Version)\thanks{This work was supported by the German Research Foundation (DFG) Grant Petri Games (392735815) and the Collaborative Research Center “Foundations of Perspicuous Software Systems” (TRR 248, 389792660), and by the European Research Council (ERC) Grant OSARES (683300).}
}
\titlerunning{Model Checking Data Flows in Concurrent Network Updates (Full Version)}
%
\author{
Bernd Finkbeiner\inst{1} \and
Manuel Gieseking\inst{2} \and\\
Jesko Hecking-Harbusch\inst{1} \and
Ernst-R\"udiger Olderog\inst{2}
}

\authorrunning{
B.\ Finkbeiner et al.
} 
%
\institute{
Saarland University, Saarbr\"ucken, Germany\\
\and
University of Oldenburg, Oldenburg, Germany\\
}

\maketitle

\begin{abstract}
  We present a model checking approach for the verification of
  data flow correctness in networks during concurrent updates of the network configuration.
  This verification problem is of great importance
  for software-defined networking (SDN), where errors 
  can lead to packet loss, black holes, and security
  violations.  Our approach is based on a specification of temporal
  properties of individual data flows, such as the requirement that the flow
  is free of cycles.  We check whether these properties are
  simultaneously satisfied for all active data flows while the network configuration is updated. 
  To represent the behavior of the concurrent network controllers and the
  resulting evolutions of the configurations, we introduce an
  extension of Petri nets with a transit relation, which characterizes
  the data flow caused by each transition of the Petri net. For safe Petri
  nets with transits, we reduce the verification of temporal flow properties
  to a circuit model checking problem that can be solved with
  effective verification techniques like IC3, interpolation, and
  bounded model checking.  We report on encouraging experiments with a
  prototype implementation based on the hardware model checker ABC.
\end{abstract}

\section{Introduction}
\nocite{*}
\label{sec:intro}
Software-defined networking (SDN)~\cite{DBLP:journals/ccr/McKeownABPPRST08,DBLP:journals/cacm/CasadoFG14} is a networking technology that
separates the packet forwarding process, called the 
\emph{data plane}, from the routing process, called the \emph{control plane}.
Updates to the routing configuration can be initiated by a central
controller and are then implemented in a distributed manner in the 
network. The separation of data plane and control plane makes the
management of a software-defined network dramatically more efficient than a traditional
network.
The model checking of network configurations and concurrent updates between them is a serious challenge.
The distributed update process can cause issues like forwarding loops,
black holes, and incoherent routing which, from the perspective of the end-user, result in performance degradation, broken
connections, and security violations.
Correctness of concurrent network updates has previously
been addressed with restrictions like \emph{consistent
updates}~\cite{DBLP:conf/sigcomm/ReitblattFRSW12}:
every packet is guaranteed during its entire journey to either encounter the initial routing
configuration or the final routing
configuration, but never a mixture in the sense that some switches
still apply the old routing configuration and others already apply the new
routing configuration. 
Under these restrictions, updates to network configurations can be synthesized~\cite{DBLP:conf/cav/El-HassanyTVV17,DBLP:conf/cav/McClurgHC17}.
Ensuring consistent updates is slow and expensive: switches must
store multiple routing tables and messages must be tagged with
version numbers.

In this paper, we propose the \emph{verification} of network configurations and concurrent updates between them. 
We specify desired properties of the data flows in the network, such as the absence of
loops, and then automatically check, for a given initial routing configuration
and a
concurrent update, whether the specified
properties are simultaneously satisfied for all active data flows
while the routing configuration is updated. This allows us to check a
specific concurrent update and to thus only impose a sequential
order where this is strictly needed to avoid an erroneous configuration
during the update process.

Our approach is based on temporal logic and model checking. The
control plane of the network can naturally be specified as a Petri
net. Petri nets are convenient to differentiate between
sequential and parallel update steps. The data plane, however, is
more difficult to specify. The standard flow relation of a Petri net
does not describe which ingoing token of a transition
transits to which outgoing token. In theory, such a
connection could be made with \emph{colored} Petri nets~\cite{jensen92}, by using a
uniquely colored token for each data flow in the network. Since there is no
bound on the number of packets, this would 
require infinitely many tokens and colors to track the infinitely many data flows.
To avoid this problem, we develop an extension of Petri nets called \emph{Petri nets with
transits}, which augment standard Petri nets with a \emph{transit
relation}. This relation specifies the precise data flow between ingoing
and outgoing tokens of a transition. In Petri nets with transits, a single token can carry an unbounded number of data flows. 

We introduce a linear-time temporal logic called \emph{Flow-LTL} to
specify the correct data flows in Petri nets with transits.
The logic expresses requirements on several separate timelines:
global conditions, such as fairness, are expressed in terms of the
global timeline of the system run. 
Requirements on individual 
data flows, such as that the data flow does not enter a loop, on the other hand,
are expressed in terms of the timeline of that specific data flow. The next operator,
for example, refers to the next step taken by the particular data flow, independently
of the behavior of other, simultaneously active, data flows.

Concurrent updates of software-defined networks can be modeled as safe Petri nets
with transits. We show that the model checking problem of the infinite state space of Petri
nets with transits against a Flow-LTL formula can be reduced to the LTL model
checking problem for Petri nets with a finite number of tokens; and that this model checking
problem can in turn be reduced to checking a hardware circuit against
an LTL formula. This ultimately results in a
standard verification problem, for which highly efficient tools such
as ABC~\cite{abcAsTheyWantIt} exist.

\section{Motivating Example}
\label{sec:motivation}
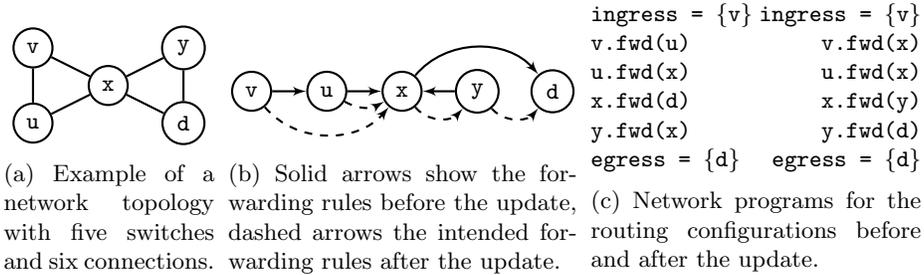
\begin{figure}[t]
\centering
	\begin{subfigure}[t]{0.23\textwidth}
		\centering
		\begin{tikzpicture}
			\tikzset{vertex/.style = {shape=circle, draw, minimum size=1.5em}}
			\tikzset{edge/.style = {- = latex'}}
			\node[vertex] (u) at  (0,0)  {\texttt{u}};
			\node[vertex] (v) at  (0,1)  {\texttt{v}};
			\node[vertex] (x) at  (1,0.5){\texttt{x}};
			\node[vertex] (y) at  (2,0)  {\texttt{d}};
			\node[vertex] (d) at  (2,1)  {\texttt{y}};
			\draw[edge] (v) to (u);
			\draw[edge] (v) to (x);
			\draw[edge] (u) to (x);
			\draw[edge] (y) to (d);
			\draw[edge] (y) to (x);
			\draw[edge] (d) to (x);
		\end{tikzpicture}
		\subcaption{Example of a network topology with five switches and six connections.}
		\label{fig:motivation_topology}
	\end{subfigure}%
~
	\begin{subfigure}[t]{0.38\textwidth}
		\centering
		\begin{tikzpicture}
			\tikzset{vertex/.style = {shape=circle, draw, minimum size=1.5em}}
			\tikzset{edge/.style = {->,> = latex'}}
			\node[vertex] (v) at  (0,0)  {\texttt{v}};
			\node[vertex] (u) at  (1,0)  {\texttt{u}};
			\node[vertex] (x) at  (2,0)  {\texttt{x}};
			\node[vertex] (y) at  (3,0)  {\texttt{y}};
			\node[vertex] (d) at  (4,0)  {\texttt{d}};
			\draw[edge] (v) to (u);
			\draw[edge] (u) to (x);
			\draw[edge] (y) to (x);
			\draw[edge, bend left=50] (x) to (d);
			\draw[edge, bend right=50, dashed] (v) to (x);
			\draw[edge, bend right=30, dashed] (u) to (x);
			\draw[edge, bend right=50, dashed] (y) to (d);
			\draw[edge, bend right=50, dashed] (x) to (y);
		\end{tikzpicture}
		\subcaption{Solid arrows show the forwarding rules before the update, dashed arrows the intended forwarding rules after the update.}
		\label{fig:motivation_update}
	\end{subfigure}%
~
	\begin{subfigure}[]{0.36\textwidth}
			\texttt{ingress = \{v\}} \texttt{ingress = \{v\}}
			\texttt{v.fwd(u)} \qquad\qquad \texttt{  v.fwd(x)}
			\texttt{u.fwd(x)} \qquad\qquad \texttt{  u.fwd(x)}
			\texttt{x.fwd(d)} \qquad\qquad \texttt{  x.fwd(y)}
			\texttt{y.fwd(x)} \qquad\qquad \texttt{  y.fwd(d)}
			\texttt{egress = \{d\}} \quad \texttt{egress = \{d\}}
		\subcaption{Network programs for the routing configurations before and after the update.}
		\label{fig:motivation_program}
	\end{subfigure}%
	\caption{Example (due to~\cite{DBLP:conf/networking/ForsterMW16}) of an update to a software-defined network.
	}
	\label{fig:motivation}
\end{figure}

We motivate our approach with a typical network update problem taken from the literature~\cite{DBLP:conf/networking/ForsterMW16}.
Consider the simple network topology shown in Fig.~\ref{fig:motivation_topology}. From the global point of view, our goal is to update
the network from the routing configuration shown with solid lines in Fig.~\ref{fig:motivation_update} to the routing configuration
shown with dashed lines. Such routing configurations are typically given as static NetCore~\cite{DBLP:conf/icfp/FosterHFMRSW11,DBLP:conf/nsdi/MonsantoRFRW13} programs like the ones
shown in Fig.~\ref{fig:motivation_program}. The \texttt{ingress} and \texttt{egress} sections define where packets enter and leave the network, respectively. 
Expressions of the form \texttt{v.fwd(u)} define that switch~\texttt{v} forwards packets to switch~\texttt{u}.

It is not straightforward to see how the update from \refFig{motivation_update} can be implemented in a distributed manner. 
If switch~\texttt{x} is updated to forward to switch~\texttt{y} \emph{before} \texttt{y} is updated to forward to switch~\texttt{d}, then any data
flow that reaches \texttt{x} is sent into a loop between \texttt{x} and
\texttt{y}. A correct update process
must thus ensure sequentiality between switch updates
\texttt{upd(y.fwd(d))} and \texttt{upd(x.fwd(y))}, in this order.
The only other switch with changing routing is switch \texttt{v}. This update
can occur in any order. A correct concurrent update would thus work as follows:
\[
	\texttt{(upd(y.fwd(d)) >> upd(x.fwd(y))) || upd(v.fwd(x))},
\]
where \texttt{>>} and \texttt{||} denote sequential and parallel composition, respectively.

\begin{figure}[t]
\centering
	\begin{tikzpicture}[node distance=1.25cm, on grid]%
		\node [envplace, tokens=1]						(s0)  [label=right:$\mathit{sw_\texttt{v}}$] {};
		\node [transition, above=of s0] 					(t0)  [label=above:$\mathit{fwd}_{\texttt{v}\rightarrow \texttt{u}}$] {};
		\node [envplace, left=of t0, tokens=1]				(s1)  [label=left:$\mathit{sw_\texttt{u}}$] {};
		\node [envplace, right=of t0, tokens=1]				(s2)  [label=above:\texttt{~~~~~v.fwd(u)}] {};
		
		\draw[thick,shorten >=1pt,<->,shorten <=1pt] (s0) -- (t0);
		\draw[thick,shorten >=1pt,<->,shorten <=1pt] (s1) -- (t0);
		\draw[thick,shorten >=1pt,<->,shorten <=1pt] (s2) -- (t0);
		
		\draw[thick,->,dotted, cdc_Blue, shorten >=1pt, shorten <=1pt] ([xshift=-0.2cm]s0.north) -- ([xshift=-0.2cm]t0.south);
		\draw[thick,->,dotted, cdc_Blue, shorten >=1pt, shorten <=1pt] ([yshift=-0.2cm]t0.west) -- ([yshift=-0.2cm]s1.east);
		\draw[thick,<->, dashed, cdc_Green, shorten >=1pt, shorten <=1pt] ([yshift=0.2cm]t0.west) -- ([yshift=0.2cm]s1.east);

		\node [transition, below=of s0] 								(t1)  [label=below:$\mathit{fwd}_{\texttt{v}\rightarrow \texttt{x}}$] 	{};
		\node [envplace, left=of t1, tokens=1]						(s3)  [label=left:$\mathit{sw_\texttt{x}}$]    {};
		\node [envplace, right=of t1]							(s4)  [label=below:\texttt{~~~~~v.fwd(x)}]    {};
		
		\draw[thick,shorten >=1pt,<->,shorten <=1pt] (s0) -- (t1);
		\draw[thick,shorten >=1pt,<->,shorten <=1pt] (s3) -- (t1);
		\draw[thick,shorten >=1pt,<->,shorten <=1pt] (s4) -- (t1);
			
		\node [transition, left of=s0]							(t00)  [label=above:$\mathit{ingress}_\texttt{v}$] 	{};
		\draw[thick,shorten >=1pt,<->,shorten <=1pt] (s0) -- (t00);
		
		\draw[thick,->,dotted, cdc_Blue, shorten >=1pt, shorten <=1pt] ([yshift=-0.2cm]t00.east) -- ([yshift=-0.2cm]s0.west);
		\draw[thick,<->,dashed, cdc_Green, shorten >=1pt, shorten <=1pt] ([yshift=0.2cm]t00.east) -- ([yshift=0.2cm]s0.west);
			
		\draw[thick,->,dotted, cdc_Blue, shorten >=1pt, shorten <=1pt] ([xshift=-0.2cm]s0.south) -- ([xshift=-0.2cm]t1.north);
		\draw[thick,->,dotted, cdc_Blue, shorten >=1pt, shorten <=1pt] ([yshift=-0.2cm]t1.west) -- ([yshift=-0.2cm]s3.east);
		\draw[thick,<->, dashed, cdc_Green, shorten >=1pt, shorten <=1pt] ([yshift=0.2cm]t1.west) -- ([yshift=0.2cm]s3.east);
		
		\node [envplace, tokens=1, right=6cm of s1]						(u0)  [label=above:$\mathit{update\_start}$] {};
		\node [envplace, left of=u0]								(u1)  [label=left:] {};
		\node [envplace, right of=u0]								(u2)  [label=left:] {};
		\node [envplace, below of=u1, below of=u1]								(u4)  [label=left:] {};
		\node [envplace, below of=u2, right of=u2]								(u5)  [label=left:] {};
		\node [envplace, below of=u2, below of=u2]								(u6)  [label=left:] {};
		
		\node [transition] 	at ($(u1)!0.5!(u4)$)							(t2)  [label=left:$\mathit{update}_\texttt{v}~$] 	{}
			edge [pre]  (s2)
			edge [pre]  (u1)
			edge [post]	(s4)
			edge [post]  (u4);
		
		\node [transition, below of=u0]		(U0)  [label={[align=left]below:$\mathit{concurrent}\_$\\~~~$\mathit{update}$}] {};
		\draw[thick,shorten >=1pt,->,shorten <=1pt] (u0) -- (U0);
		\draw[thick,shorten >=1pt,->,shorten <=1pt] (U0) -- (u1);
		\draw[thick,shorten >=1pt,->,shorten <=1pt] (U0) -- (u2);
		
		\node [transition, right of=u2]		(U1)  [label=right:$\mathit{update}_\texttt{y}$] {};
		\draw[thick,shorten >=1pt,->,shorten <=1pt] (u2) -- (U1);
		\draw[thick,shorten >=1pt,->,shorten <=1pt] (U1) -- (u5);
		
		\node [transition, below of=u5]		(U2)  [label=right:$\mathit{update}_\texttt{x}$] {};
		\draw[thick,shorten >=1pt,->,shorten <=1pt] (u5) -- (U2);
		\draw[thick,shorten >=1pt,->,shorten <=1pt] (U2) -- (u6);
		
		\node [envplace, left of=U2, below=0.75cm of U1, opacity=0]								(u66)  [label={[label distance=-8mm]left:$\ldots$}] {};
		\node [envplace, right of=U2, below=0.75cm of U1, opacity=0]								(u7)  [label={[label distance=-8mm]right:$\ldots$}] {};
		\draw[thick,shorten >=1pt,->,shorten <=1pt] (u66) -- (U1);
		\draw[thick,shorten >=1pt,->,shorten <=1pt] (U1) -- (u7);
		
		\node [envplace, right of=U2, above=0.75cm of U2, opacity=0]								(u8)  [label={[label distance=-8mm]right:$\ldots$}] {};
		\node [envplace, left of=U2, above=0.75cm of U2, opacity=0]								(u9)  [label={[label distance=-8mm]left:$\ldots$}] {};
		\draw[thick,shorten >=1pt,->,shorten <=1pt] (u9) -- (U2);
		\draw[thick,shorten >=1pt,->,shorten <=1pt] (U2) -- (u8);

	\end{tikzpicture}
	\caption{Example Petri net with transits encoding the data plane on the left and the control plane on the right. The standard flow relation, describing the flow of tokens, is depicted by solid black arrows, the transit relation by colored arrows. Colors that occur on both ingoing and outgoing arrows of a transition define that the transition extends the data flow. If an outgoing arrow has a color that does not appear on an ingoing arrow, a new data flow is initiated.
}
	\label{fig:SDN_singleupdate}
\end{figure}
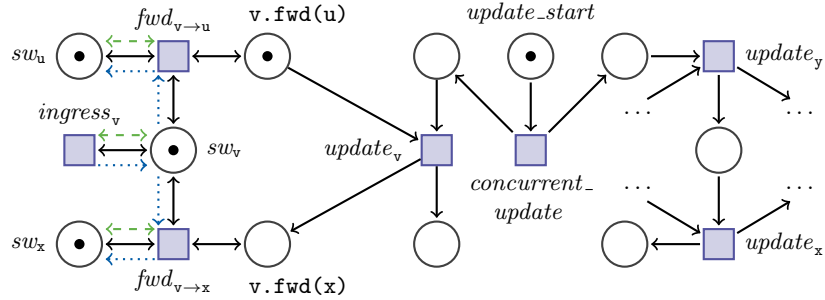

Figure~\ref{fig:SDN_singleupdate} shows a Petri net model for the network topology and the concurrent update from the initial to the final routing configuration from \refFig{motivation}. 
The right-hand side models the control plane, where, beginning in $\mathit{update\_start}$, the
update of \texttt{v} and, concurrently, the sequential update to \texttt{y} and then to
\texttt{x} is initiated. 
Each marking of the net represents a control state of the network. 
Changes to the control state are thus modeled by the standard flow relation.
Leaving out the control plane allows us to verify configurations of network topologies. 

On the left-hand side, we model the data plane by extending the Petri net with a transit relation. 
This new type of Petri nets will be defined formally in the next section.
We only depict the update to the data flow in and from switch~\texttt{v}.
Places $sw_\texttt{u}$, $sw_\texttt{v}$, and $sw_\texttt{x}$ represent the switches \texttt{u}, \texttt{v},  and \texttt{x}, respectively.
The data plane is modeled by the transit relation which indicates the extension of the data flows during each transition at the switches. 

The standard flow relation is depicted by solid black arrows and the transit relation by colored arrows. If an outgoing arrow has a color that does \emph{not} appear
on an ingoing arrow, then a new data flow is initiated. In our example, data flows are initiated by transition $\mathit{ingress}_\texttt{v}$ and the (dotted) blue arrow.
Colors that occur on both in- and outgoing arrows extend the data flow. In transition $\mathit{fwd}_{\texttt{v}\rightarrow \texttt{u}}$, the (dotted) blue arrows 
indicate the extension of the data flow from $sw_\texttt{v}$ to $sw_\texttt{u}$. The (dashed) green arrows between $\mathit{fwd}_{\texttt{v}\rightarrow \texttt{u}}$ and $sw_\texttt{u}$ indicate that,
in addition to the incoming data flow from $sw_\texttt{v}$, there may be data flows that have previously reached $sw_\texttt{u}$ and have not yet departed from  $sw_\texttt{u}$. These flows stay in $sw_\texttt{u}$. 

Notice that $\mathit{ingress}_\mathtt{v}$, $\mathit{fwd}_{\texttt{v}\rightarrow \texttt{u}}$, and $\mathit{fwd}_{\texttt{v}\rightarrow \texttt{x}}$ do not actually move tokens because of the double-headed arrows.
None of these transitions change the control state, they only model the data flow. 
As the switches  \texttt{u}, \texttt{v},  and \texttt{x} remain continuously active, their tokens in $sw_u$, $sw_v$, and $sw_x$ are never moved.
By contrast, $\mathit{update}_\texttt{v}$ moves the token from \texttt{v.fwd(u)} to \texttt{v.fwd(x)}, thus disabling the data flow from $sw_\texttt{v}$ to $sw_\texttt{u}$ and enabling the
data flow from $sw_\texttt{v}$ to $sw_\texttt{x}$. 
We specify the correctness of our update process with formulas of the temporal logic \emph{Flow-LTL}. 
The formula $\A \ltlEventually \texttt{d}$ expresses \emph{connectivity} requiring that
all data flows ($\A$) eventually ($\ltlEventually$) arrive at the egress switch~\texttt{d}.
Flow-LTL and the specification of data flow properties are discussed in more detail in \refSection{FlowLTL} and \refSection{applications}.
The general construction of the motivating example is formalized in App.~\ref{appendix:motivation}.

\section{Petri Nets with Transits}
\label{sec:petri-nets}
We give the formal definition of \emph{Petri nets with transits}. 
We assume some basic knowledge about standard Petri nets~\cite{DBLP:books/sp/Reisig85a}.
A safe \emph{Petri net with transits} (PNwT) is a structure $\petriNetFl$, where the set of \emph{places}~$\pl$, the set of \emph{transitions}~$\tr$, 
the (\emph{control}) \emph{flow relation}~$\fl \subseteq (\pl \times \tr) \cup (\tr \times \pl)$, and the \textit{initial marking}~$\init\subseteq\pl$ are as in \emph{safe} Petri nets.
In safe Petri nets, each reachable marking contains at most one token per place. 
We add the \emph{transit relation}~$\tokenflow$ of tokens for transitions to obtain Petri nets with transits.
For each transition $t \in \transitions$, we postulate that $\tokenflow(t)$ is a relation of type $\tokenflow(t) \subseteq (\pre{\pNet}{t}\cup\{\startfl\}) \times \post{\pNet}{t}$, written in infix notation, where the symbol $\startfl$ denotes a \emph{start}.
$p\ \tokenflow(t)\ q$ defines that the token in place~$p$ \emph{transits} via transition~$t$ to place~$q$ and $\startfl\ \tokenflow(t)\ q$ defines that the token in place~$q$ marks the start of a new data flow via transition~$t$.
The graphic representation of $\tokenflow(t)$ 
in Petri nets with transits uses a \emph{color coding} as can be seen in 
\refFig{SDN_singleupdate}.
Black arrows represent the usual \emph{control flow}.
Other matching colors per transition are used to represent the transits of tokens.
Transits allow us to specify which
data flows are moved forward, split, and merged, which data flows are
removed, and which data flows are newly created.

Data flows can be of infinite length
and can be created at any point in time.
Hence, the number of data flows existing in a place during an execution depends on the causal past of the place.
Therefore, we recall informally the notions of unfoldings and runs~\cite{DBLP:journals/acta/Engelfriet91,DBLP:series/eatcs/EsparzaH08}
and apply them to Petri nets with transits. 
In the unfolding of a Petri net $\pNet$, every transition stands for the unique occurrence
(instance) of a transition of $\pNet$ during an execution.
To this end, every loop in~$\pNet$ is unrolled and every  join of transitions in a place
is expanded by duplicating the place.
Forward branching, however, is preserved. Formally, an \emph{unfolding}
is a branching process $\beta^U = (\pNet^U, \lambda^U)$ 
consisting of an occurrence net~$\pNet^U$ and a 
homomorphism $\lambda^U$ that labels the places and transitions in $\pNet^U$
with the corresponding elements of $\pNet$.
The unfolding exhibits concurrency, causality, and nondeterminism (forward branching)
of the unique occurrences of the transitions in $\pNet$ during all possible
executions.
A \emph{run} of $\pNet$ is a subprocess $\runPN = (\pNet^R, \rho)$ of $\beta^U$,
where \(\forall p\in\pl^R:|\post{\pNet^R}{p}|\leq 1\) holds,
i.e., all nondeterminism has been resolved but concurrency is preserved.
Thus, a run formalizes one concurrent execution of $\pNet$.
We introduce the \emph{unfolding of Petri nets with transits} by lifting the transit relation to the
unfolding $\beta^U = (\pNet^U, \lambda^U)$ .
We define the relation $\tokenflow^U$ as follows:
For any $t \in \transitions^U$, the transit relation 
$\tokenflow^U(t) \subseteq (\pre{\pNet^U}{t} \cup \{\rhd\}) \times \post{\pNet^U}{t}$
is defined for all $p,q \in \places^U$ by $p\ \tokenflow^U(t)\ q \Leftrightarrow \lambda^U(p)\ \tokenflow(\lambda^U(t))\ \lambda^U(q)$.

We use the transit relation in the unfolding to introduce (data) flow chains.
A (\emph{data}) \emph{flow chain} in $\beta^U$ is a maximal sequence
$\xi =  p_0, t_0, p_1, t_1, p_2, \dots$
of places in $\places^U$ with connecting transitions in $\transitions^U$ such that
\begin{enumerate}
	\item $\exists t \in \transitions^U: \startfl\ \tfl^U(t)\ p_0$,
	
	\item if $\xi$ is infinite then for all $i \ge 0$ 
	      the transit relation 
	      $p_i\ \tfl^U(t_i)\  p_{i+1}$ holds,
  
  \item if $\xi$ is finite, say $p_0,t_0, \dots, t_{n-1}, p_n$ 
	      for some $n \ge 0$,
	      then for all $i$ with $0 \le i < n$ the transit relation 
	      $p_i\ \tfl^U(t_i)\  p_{i+1}$ holds, and
	      there is no place $q \in \pl^U$ and no transition $t \in \tr^U$ 
	      with $p_n\ \tokenflow^U(t)\ q$.
\end{enumerate}

\section{Flow-LTL for Petri Nets with Transits}
\label{sec:FlowLTL}
We recall LTL applied to Petri nets and
define our extension \emph{Flow-LTL} to specify the behavior of flow chains
in Petri nets with transits.
We fix a Petri net with transits $\petriNetFl$ throughout the section.

\subsection{Linear Temporal Logic for Petri nets}

We define $AP = \places \cup \transitions$ as the set of \emph{atomic propositions}. 
The set $\mathtt{LTL}$ of \emph{linear temporal logic}~(LTL) formulas over $AP$ has the following syntax $\phiLTL ::= true \mid a \mid \neg \phiLTL \mid \phiLTL_1 \land \phiLTL_2 \mid \ltlNext \phiLTL \mid \phiLTL_1 \U \phiLTL_2$,
where $a \in AP$. Here, $\ltlNext$ is the \emph{next} operator and {\sf U} is 
the \emph{until} operator. 
We use the abbreviated temporal operators $\ltlEventually$ (\emph{eventually}) and $\ltlAlways$ (\emph{always}) as usual.
A \emph{trace} is a mapping $\sigma: \mathbb{N} \longrightarrow 2^{AP}$.
The trace $\sigma^i: \mathbb{N} \longrightarrow 2^{AP}$, defined by
$\sigma^i(j) = \sigma(i+j)$ for all $j \in \mathbb{N}$, is the 
$i$th \emph{suffix} of \(\sigma\).

We define the traces of a Petri net based on its runs.
Consider a run $\runPN = (\mathcal N^R, \rho)$ of $\pNet$ and 
a finite or infinite firing sequence 
$\zeta = M_0\firable{t_0}M_1\firable{t_1}M_2\cdots $ of $\pNet^R$ with 
$M_0 = \initialmarking^{R}$. This sequence \emph{covers} $\runPN$ if 
\(
  (\forall p \in \places^{R}\ \exists i \in \mathbb{N}: p \in M_i)
  \land
  (\forall t \in \transitions^{R}\ \exists i \in \mathbb{N}: t = t_i),
\)
i.e., all places and transitions in $\mathcal N^R$ appear in $\zeta$. 
Note that several firing sequences may cover $\runPN$. 
To each firing sequence~$\zeta$ covering $\runPN$, we associate an
infinite trace $\sigma(\zeta): \mathbb{N} \longrightarrow 2^{AP}$.
If $\zeta$ is finite, say $\zeta = M_0\firable{t_0} \cdots \firable{t_{n-1}}M_n$ 
for some $n \ge 0$, we define
1.\ $\sigma(\zeta)(i) = \rho(M_i) \cup \{ \rho(t_i)\}$ for $0 \le i < n$ and
2.\ $\sigma(\zeta)(j) = \rho(M_n)$ for $j \ge n$.
Thus, we record for $0 \le i < n$ (case 1) all places of the original net $\mathcal{N}$ 
that label the places in the marking $M_i$ in $\pNet^R$ and the transition of $\mathcal{N}$
that labels the transition $t_i$ in $\pNet^R$ outgoing from $M_i$.
At the end (case 2), we \emph{stutter} by repeating the set of places recorded in $\sigma(\zeta)(n)$
from $n$ onwards, but repeat no transition. 
If $\zeta$ is infinite 
we apply case~1 for all $i \ge 0$ as no stuttering is needed to generate
an infinite trace $\sigma(\zeta)$.

We define the \emph{semantics} of $\mathtt{LTL}$ on Petri nets by
\(\pNet \models_\mathtt{LTL} \phiLTL\) iff for all runs 
\(\runPN \mbox{ of } \pNet:\ \runPN \models_\mathtt{LTL} \phiLTL\),
which means that
for all firing sequences \(\zeta \mbox{ covering } \runPN:\ \sigma(\zeta) \models_\mathtt{LTL} \phiLTL\),
where the latter refers to the usual binary satisfaction relation
$\models_\mathtt{LTL}$ between traces $\sigma$ and formulas $\phiLTL \in \mathtt{LTL}$
 defined by:
 \(\sigma \models_\mathtt{LTL}\ true\),
\(\sigma \models_\mathtt{LTL} a \  \mbox{iff}\ a \in \sigma(0)\),
\(\sigma \models_\mathtt{LTL} \neg \phiLTL \ \mbox{iff}\  \mbox{not}\ \sigma \models_\mathtt{LTL} \phiLTL\),
\(\sigma \models_\mathtt{LTL} {\phiLTL}_1 \land {\phiLTL}_2 \ \mbox{iff}\ \sigma \models_\mathtt{LTL} {\phiLTL}_1 \ \mbox{and}\  \sigma \models_\mathtt{LTL} {\phiLTL}_2\),
\(\sigma \models_\mathtt{LTL} \ltlNext\, \phiLTL \ \mbox{iff}\ \sigma^1 \models_\mathtt{LTL} \phiLTL \), 
\(\sigma \models_\mathtt{LTL} {\phiLTL}_1\, {\sf U}\, {\phiLTL}_2 \ \mbox{iff}\ \mbox{there exists a } j\ge 0 \mbox{ with } \sigma^j \models_\mathtt{LTL} {\phiLTL}_2 \mbox{ and } \mbox{for all }  i \mbox{ with } 0 \le i < j \mbox{ the following holds: } \sigma^i \models_\mathtt{LTL} {\phiLTL}_1\).

\subsection{Definition of Flow-LTL for Petri Nets with Transits}

For Petri nets with transits, we wish to express requirements on several separate timelines.
Based on the global timeline of the system run, global conditions like fairness and maximality can be expressed. 
Requirements on individual data flows, 
e.g., that the data flow does not enter a loop,
are expressed in terms of the timeline of that specific data flow. 
\emph{Flow-LTL} comprises of \emph{run formulas} $\phiRun$ specifying 
the usual LTL behavior on markings and \emph{data flow formulas} $\varphi_F$ specifying properties of flow chains inside runs:
\[
 \phiRun  ::= \phiLTL \mid \phiRun_1 \land \phiRun_2
                  \mid \phiRun_1 \lor \phiRun_2 
                  \mid \phiLTL \rightarrow \phiRun
                  \mid \varphi_F 
 \quad \text{and} \quad 
 \varphi_F  ::=  \mathbb{A}\ \phiLTL
\]
where formulas $\phiLTL \in \mathtt{LTL}$ may appear both inside $\phiRun$
and $\varphi_F$. 

To each flow chain $\xi$ in a run $\runPN$,
we associate an infinite \emph{flow trace}
$\sigma(\xi): \mathbb{N} \longrightarrow 2^{AP}$.
If $\xi$ is finite, say $\xi = p_0,t_0 \dots, t_{n-1}, p_n$ for some $n \ge 0$,
we define\\
	(1) $\sigma(\xi)(i) = \{\rho(p_i), \rho(t_{i})\} \text{ for } 0 \le i < n$ and 
	(2) $\sigma(\xi)(j) = \{\rho(p_n)\} \text{ for } j \ge n$.

Thus, we record for $0 \le i < n$ (case 1) the place and the transition of the original net $\mathcal{N}$ 
that label the place $p_i$ in $\pNet^R$ and the transition $t_{i}$ in $\pNet^R$ 
outgoing from $p_i$. 
At the end (case 2), we \emph{stutter} by repeating the place recorded in $\sigma(\xi)(n)$
infinitely often. 
No transition is repeated in this case. 

If $\xi$ is infinite 
we apply case 1 for all $i \ge 0$. Here, no stuttering is needed to generate
an infinite flow trace $\sigma(\xi)$ and each element of the trace consists of a place and a transition. 

A Petri net with transits $\pNet$ satisfies $\phiRun$,
abbr.\ $\pNet \models \phiRun$, if the following holds:
\[
\begin{array}{lll}
  \pNet \models \phiRun
     &\ \mbox{iff} \ & \mbox{for all runs } \runPN \mbox{ of } \pNet:
                \ \runPN \models \phiRun
  \\[1mm]
   \runPN \models \phiRun  
     &\ \mbox{iff} \ & \mbox{for all firing sequences } \zeta \mbox{ covering } \runPN: \ \runPN, \sigma(\zeta) \models \phiRun
  \\[1mm]
  \runPN, \sigma(\zeta) \models \phiLTL 
     &\ \mbox{iff} \ & \sigma(\zeta) \models_\mathtt{LTL} \phiLTL
  \\[1mm]
  \runPN, \sigma(\zeta)  \models \phiRun_1 \land \phiRun_2  
     &\ \mbox{iff} \ & \runPN, \sigma(\zeta) \models \phiRun_1 \ \mbox{and}\  
                \runPN, \sigma(\zeta)  \models \phiRun_2
  \\[1mm]
  \runPN, \sigma(\zeta)  \models \phiRun_1 \lor \phiRun_2  
     &\ \mbox{iff} \ & \runPN, \sigma(\zeta) \models \phiRun_1 \ \mbox{or}\  
                \runPN, \sigma(\zeta)  \models \phiRun_2
  \\[1mm]
  \runPN, \sigma(\zeta)  \models \phiLTL \rightarrow \phiRun
     &\ \mbox{iff} \ & \runPN, \sigma(\zeta) \models \phiLTL \ \mbox{implies}\  
                \runPN, \sigma(\zeta)  \models \phiRun 
  \\[1mm]
  \runPN, \sigma(\zeta) \models \mathbb{A}\ {\phiLTL}   
     &\ \mbox{iff} \ & \mbox{for all flow chains } \xi \mbox{ of } \runPN: 
                          \ \sigma(\xi) \models_\mathtt{LTL} \phiLTL
\end{array}
\]

\section{Example Specifications}
\label{sec:applications}
We illustrate Flow-LTL with examples from the literature on software-defined networking. 
Specifications on data flows like loop and drop freedom are encoded as data flow formulas.
Fairness assumptions for switches are given as run formulas.

\subsection{Data Flow Formulas}

We show how properties from the literature can be encoded as data flow formulas.
For a network topology, let $\mathit{Sw}$ be the set of all switches, $\mathit{Ingr} \subseteq \mathit{Sw}$ the ingress switches, and $\mathit{Egr} \subseteq \mathit{Sw}$ the egress switches with $\mathit{Ingr} \cap \mathit{Egr} = \emptyset$. 
The connections between switches are given by $\mathit{Con} \subseteq \mathit{Sw} \times \mathit{Sw}$.

\noindent
{\bf Loop freedom.} Loop freedom~\cite{DBLP:conf/sigcomm/LiuWZYWM13} requires that a data flow visits every switch at most once. 
In \refSection{motivation}, we outlined that arbitrarily ordered updates can lead to loops in the network.
The following data flow formula expresses that each data flow is required to not visit a non-egress switch anymore after it has been forwarded and therefore left that switch (realized via the $\U$-operator):
\[
\A \ltlAlways\, (\bigwedge_{\texttt{s}\in\mathit{Sw}\setminus\mathit{Egr}} \texttt{s} \rightarrow (\texttt{s} \U \ltlAlways \neg \texttt{s}))
\]

\begin{figure}[t]
\centering
	\begin{tikzpicture}[node distance=1cm, on grid]%
		\node [envplace, tokens=1]							(s0)  [label={[yshift=-1mm]above:$\mathit{sw_\texttt{w}}$}] {};
		\node [envplace] (d0) [opacity=0, right of=s0, right of=s0, above of=s0] {};
		\node [envplace] (d1) [opacity=0, right of=s0, right of=s0] {};
		\node [envplace] (d2) [opacity=0, right of=s0, right of=s0, below of=s0] {};
		\node [envplace, tokens=1]  	at ($(d0)!0.5!(d1)$)					(s1)  [label=above:$\mathit{sw}_\texttt{x}$] {};
		\node [envplace, tokens=1]  	at ($(d1)!0.5!(d2)$)					(s2)  [label=below:$\mathit{sw}_\texttt{y}$] {};
		\node [transition, left=of s1] 								(t1)  [] 	{};
		\draw[thick,shorten >=1pt,<->,shorten <=1pt] (s0) -- (t1);
		\draw[thick,shorten >=1pt,<->,shorten <=1pt] (s1) -- (t1);
		\node [transition, left=of s2] 								(t2)  [] 	{};
		\draw[thick,shorten >=1pt,<->,shorten <=1pt] (s0) -- (t2);
		\draw[thick,shorten >=1pt,<->,shorten <=1pt] (s2) -- (t2);

		\draw[thick,->,dotted, cdc_Blue, shorten >=2pt, shorten <=2pt] (s0.0) -- (t1.240);
		\draw[thick,->,dotted, cdc_Blue, shorten >=1pt, shorten <=1pt]   ([yshift=-0.2cm]t1.east) -- ([yshift=-0.2cm]s1.west);
		\draw[thick,<->, dashed, cdc_Green, shorten >=1pt, shorten <=1pt] ([yshift=0.2cm]t1.east) -- ([yshift=0.2cm]s1.west);
		
		\draw[thick,->,dotted, cdc_Blue, shorten >=2pt, shorten <=2pt] (s0.310) -- (t2.185);
		\draw[thick,->,dotted, cdc_Blue, shorten >=1pt, shorten <=1pt]   ([yshift=-0.2cm]t2.east) -- ([yshift=-0.2cm]s2.west);
		\draw[thick,<->, dashed, cdc_Green, shorten >=1pt, shorten <=1pt] ([yshift=0.2cm]t2.east) -- ([yshift=0.2cm]s2.west);
		
		\node [transition, left of=s0]							(t00)  [label=above:] 	{};
		\draw[thick,shorten >=1pt,<->,shorten <=1pt] (s0) -- (t00);
		
		\draw[thick, ->,dotted, cdc_Blue, shorten >=1pt, shorten <=1pt] ([yshift=-0.2cm]t00.east) -- ([yshift=-0.2cm]s0.west);
		\draw[thick,<->,dashed, cdc_Green, shorten >=1pt, shorten <=1pt] ([yshift=0.2cm]t00.east) -- ([yshift=0.2cm]s0.west);	
		
		\node [envplace, right of=s0, right of=s0, right of=s0, right of=s0, tokens=1]			(s3)  [label={[yshift=-1mm]above:$\mathit{sw_\texttt{z}}$}] {};
		\node [transition, right=of s1] 								(t3)  [] 	{};
		\draw[thick,shorten >=1pt,<->,shorten <=1pt] (s1) -- (t3);
		\draw[thick,shorten >=1pt,<->,shorten <=1pt] (s3) -- (t3);
		\node [transition, right=of s2] 								(t4)  [] 	{};
		\draw[thick,shorten >=1pt,<->,shorten <=1pt] (s2) -- (t4);
		\draw[thick,shorten >=1pt,<->,shorten <=1pt] (s3) -- (t4);
		
		\draw[thick,->,dotted, cdc_Blue, shorten >=1pt, shorten <=1pt] ([yshift=-0.2cm]s1.east) -- ([yshift=-0.2cm]t3.west);
		\draw[thick,->,dotted, cdc_Blue, shorten >=2pt, shorten <=2pt]   (t3.300) -- (s3.175);
		\draw[thick,<->, dashed, cdc_Green, shorten >=2pt, shorten <=2pt] (t3.0) -- (s3.135);
		
		\draw[thick,->,dotted, cdc_Blue, shorten >=1pt, shorten <=1pt] ([yshift=-0.2cm]s2.east) -- ([yshift=-0.2cm]t4.west);
		\draw[thick,->,dotted, cdc_Blue, shorten >=2pt, shorten <=2pt]   (t4.0) -- (s3.220);
		\draw[thick,<->, dashed, cdc_Green, shorten >=2pt, shorten <=2pt] (t4.60) -- (s3.180);
		
		\node [envplace, left of=t1, left of=t1, left of=t1, tokens=1, label=above:$\texttt{w.fwd(x)}$]					(u0)  [] {};
		\node [envplace, left of=t2, left of=t2, left of=t2 , label=below:$\texttt{w.fwd(y)}$]							(u1)  [] {};
		
		\node [envplace, right of=t3, right of=t3, tokens=1, label=above:$\texttt{x.fwd(z)}$]					(u2)  [] {};
		\node [envplace, right of=t4, right of=t4, label=below:$\texttt{y.fwd(z)}$]							(u3)  [] {};
		
		\draw[thick,shorten >=1pt,<->,shorten <=1pt] (u0) to [bend left=25] (t1);
		\draw[thick,shorten >=1pt,<->,shorten <=1pt] (u1) to [bend right=25] (t2);
		\draw[thick,shorten >=1pt,<->,shorten <=1pt] (u2) to [bend right] (t3);
		\draw[thick,shorten >=1pt,<->,shorten <=1pt] (u3) to [bend left] (t4);
		
		\node [transition, right of=s3, right of=s3] 								(t5)  [label=right:$\mathit{upd}2$] 	{}
			edge [pre]  (u2)
			edge [post] (u3);
		
		\node [transition, left of=t00, left of=t00] 								(t6)  [label=left:$\mathit{upd}1$] 	{}
			edge [pre]  (u0)
			edge [post] (u1);
	\end{tikzpicture}
	\caption{Concurrent network update that does not preserve drop freedom.
	}
	\label{fig:applications1}
\end{figure}
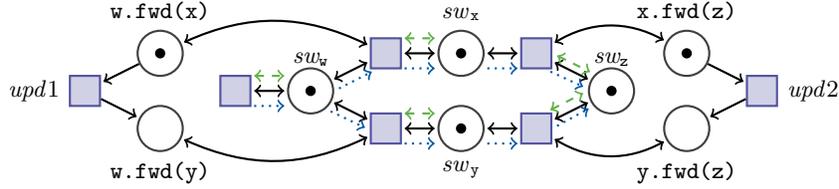

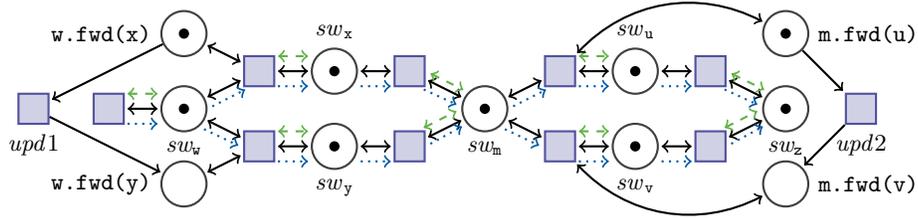
\begin{figure}[t]
\centering
	\begin{tikzpicture}[node distance=1cm, on grid]%
		\node [envplace, tokens=1]							(s0)  [label=below:$\mathit{sw_\texttt{w}}$] {};
		\node [envplace] (d0) [opacity=0, right of=s0, right of=s0, above of=s0] {};
		\node [envplace] (d1) [opacity=0, right of=s0, right of=s0] {};
		\node [envplace] (d2) [opacity=0, right of=s0, right of=s0, below of=s0] {};
		\node [envplace, tokens=1]  	at ($(d0)!0.5!(d1)$)					(s1)  [label=above:$\mathit{sw}_\texttt{x}$] {};
		\node [envplace, tokens=1]  	at ($(d1)!0.5!(d2)$)					(s2)  [label=below:$\mathit{sw}_\texttt{y}$] {};
		\node [transition, left=of s1] 								(t1)  [] 	{};
		\draw[thick,shorten >=1pt,<->,shorten <=1pt] (s0) -- (t1);
		\draw[thick,shorten >=1pt,<->,shorten <=1pt] (s1) -- (t1);
		\node [transition, left=of s2] 								(t2)  [] 	{};
		\draw[thick,shorten >=1pt,<->,shorten <=1pt] (s0) -- (t2);
		\draw[thick,shorten >=1pt,<->,shorten <=1pt] (s2) -- (t2);

		\draw[thick,->,dotted, cdc_Blue, shorten >=2pt, shorten <=2pt] (s0.0) -- (t1.240);
		\draw[thick,->,dotted, cdc_Blue, shorten >=1pt, shorten <=1pt]   ([yshift=-0.2cm]t1.east) -- ([yshift=-0.2cm]s1.west);
		\draw[thick,<->, dashed, cdc_Green, shorten >=1pt, shorten <=1pt] ([yshift=0.2cm]t1.east) -- ([yshift=0.2cm]s1.west);
		
		\draw[thick,->,dotted, cdc_Blue, shorten >=2pt, shorten <=2pt] (s0.310) -- (t2.185);
		\draw[thick,->,dotted, cdc_Blue, shorten >=1pt, shorten <=1pt]   ([yshift=-0.2cm]t2.east) -- ([yshift=-0.2cm]s2.west);
		\draw[thick,<->, dashed, cdc_Green, shorten >=1pt, shorten <=1pt] ([yshift=0.2cm]t2.east) -- ([yshift=0.2cm]s2.west);
		
		\node [envplace, tokens=1, right of=s0, right of=s0, right of=s0, right of=s0]							(s44)  [label=below:$\mathit{sw_\texttt{m}}$] {};
		\node [transition, right=of s1] 								(t11)  [] 	{};
		\draw[thick,shorten >=1pt,<->,shorten <=1pt] (s1) -- (t11);
		\draw[thick,shorten >=1pt,<->,shorten <=1pt] (t11) -- (s44);
		\node [transition, right=of s2] 								(t22)  [] 	{};
		\draw[thick,shorten >=1pt,<->,shorten <=1pt] (s2) -- (t22);
		\draw[thick,shorten >=1pt,<->,shorten <=1pt] (t22) -- (s44);
		
		\draw[thick,->,dotted, cdc_Blue, shorten >=1pt, shorten <=1pt] ([yshift=-0.2cm]s1.east) -- ([yshift=-0.2cm]t11.west);
		\draw[thick,->,dotted, cdc_Blue, shorten >=2pt, shorten <=2pt]   (t11.300) -- (s44.175);
		\draw[thick,<->, dashed, cdc_Green, shorten >=2pt, shorten <=2pt] (t11.0) -- (s44.135);
		
		\draw[thick,->,dotted, cdc_Blue, shorten >=1pt, shorten <=1pt] ([yshift=-0.2cm]s2.east) -- ([yshift=-0.2cm]t22.west);
		\draw[thick,->,dotted, cdc_Blue, shorten >=2pt, shorten <=2pt]   (t22.0) -- (s44.220);
		\draw[thick,<->, dashed, cdc_Green, shorten >=2pt, shorten <=2pt] (t22.60) -- (s44.180);
			
		\node [envplace, tokens=1, right=of t11, right=of t11, right of=t11, right of=t11]  			(s55)  [label=above:$\mathit{sw}_\texttt{u}$] {};
		\node [envplace, tokens=1, right=of t22, right=of t22, right of=t22, right of=t22]  			(s66)  [label=below:$\mathit{sw}_\texttt{v}$] {};	
		
		\node [envplace, right of=s0, right of=s0, right of=s0, right of=s0, right of=s0, right of=s0, right of=s0, right of=s0, tokens=1]							(s3)  [label=below:$\mathit{sw_\texttt{z}}$] {};
		
		\node [transition, right=of s55] 								(t55)  [] 	{};
		\draw[thick,shorten >=1pt,<->,shorten <=1pt] (s3) -- (t55);
		\draw[thick,shorten >=1pt,<->,shorten <=1pt] (t55) -- (s55);
		\node [transition, right=of s66] 								(t66)  [] 	{};
		\draw[thick,shorten >=1pt,<->,shorten <=1pt] (s3) -- (t66);
		\draw[thick,shorten >=1pt,<->,shorten <=1pt] (t66) -- (s66);
		
		\draw[thick,->,dotted, cdc_Blue, shorten >=1pt, shorten <=1pt] ([yshift=-0.2cm]s55.east) -- ([yshift=-0.2cm]t55.west);
		\draw[thick,->,dotted, cdc_Blue, shorten >=2pt, shorten <=2pt]   (t55.300) -- (s3.175);
		\draw[thick,<->, dashed, cdc_Green, shorten >=2pt, shorten <=2pt] (t55.0) -- (s3.135);
		
		\draw[thick,->,dotted, cdc_Blue, shorten >=1pt, shorten <=1pt] ([yshift=-0.2cm]s66.east) -- ([yshift=-0.2cm]t66.west);
		\draw[thick,->,dotted, cdc_Blue, shorten >=2pt, shorten <=2pt]   (t66.0) -- (s3.220);
		\draw[thick,<->, dashed, cdc_Green, shorten >=2pt, shorten <=2pt] (t66.60) -- (s3.180);
		
		\node [transition, left of=s0]							(t00)  [label=above:] 	{};
		\draw[thick,shorten >=1pt,<->,shorten <=1pt] (s0) -- (t00);
		
		\draw[thick,->,dotted, cdc_Blue, shorten >=1pt, shorten <=1pt] ([yshift=-0.2cm]t00.east) -- ([yshift=-0.2cm]s0.west);
		\draw[thick,<->,dashed, cdc_Green, shorten >=1pt, shorten <=1pt] ([yshift=0.2cm]t00.east) -- ([yshift=0.2cm]s0.west);	
		
		\node [transition, right of=s1, , right of=s1, , right of=s1] 								(t3)  [] 	{};
		\draw[thick,shorten >=1pt,<->,shorten <=1pt] (s44) -- (t3);
		\draw[thick,shorten >=1pt,<->,shorten <=1pt] (s55) -- (t3);
		\node [transition, right of=s2, right of=s2, right of=s2] 								(t4)  [] 	{};
		\draw[thick,shorten >=1pt,<->,shorten <=1pt] (s44) -- (t4);
		\draw[thick,shorten >=1pt,<->,shorten <=1pt] (s66) -- (t4);
		
		\draw[thick,->,dotted, cdc_Blue, shorten >=2pt, shorten <=2pt] (s44.0) -- (t3.240);
		\draw[thick,->,dotted, cdc_Blue, shorten >=1pt, shorten <=1pt]   ([yshift=-0.2cm]t3.east) -- ([yshift=-0.2cm]s55.west);
		\draw[thick,<->, dashed, cdc_Green, shorten >=1pt, shorten <=1pt] ([yshift=0.2cm]t3.east) -- ([yshift=0.2cm]s55.west);
		
		\draw[thick,->,dotted, cdc_Blue, shorten >=2pt, shorten <=2pt] (s44.310) -- (t4.185);
		\draw[thick,->,dotted, cdc_Blue, shorten >=1pt, shorten <=1pt]   ([yshift=-0.2cm]t4.east) -- ([yshift=-0.2cm]s66.west);
		\draw[thick,<->, dashed, cdc_Green, shorten >=1pt, shorten <=1pt] ([yshift=0.2cm]t4.east) -- ([yshift=0.2cm]s66.west);
		
		\node [envplace, above of=s0, tokens=1, label=left:$\texttt{w.fwd(x)}$]					(u0)  [] {};
		\node [envplace, below of=s0, label=left:$\texttt{w.fwd(y)}$]							(u1)  [] {};
		
		\node [envplace, above of=s3, tokens=1, label=right:$\texttt{m.fwd(u)}$]					(u2)  [] {};
		\node [envplace, below of=s3, label=right:$\texttt{m.fwd(v)}$]							(u3)  [] {};
		
		\draw[thick,shorten >=1pt,<->,shorten <=1pt] (u0) -- (t1);
		\draw[thick,shorten >=1pt,<->,shorten <=1pt] (u1) -- (t2);
		
		\draw[thick,shorten >=1pt,<->,shorten <=1pt] (u2) to [bend right=37] (t3);
		\draw[thick,shorten >=1pt,<->,shorten <=1pt] (u3) to [bend left=37] (t4);
		
		\node [transition, right=of s3] 								(t5)  [label=below:$\mathit{upd}2$] 	{}
			edge [pre]  (u2)
			edge [post] (u3);
		
		\node [transition, left=of t00] 								(t6)  [label=below:$\mathit{upd}1$] 	{}
			edge [pre]  (u0)
			edge [post] (u1);
	\end{tikzpicture}
	\caption{Concurrent network update that does not preserve packet coherence.
	}
	\label{fig:applications2}
\end{figure}

\noindent
{\bf Drop freedom.} Drop freedom~\cite{DBLP:conf/sigcomm/ReitblattCGF13} requires that no data packets are dropped. 
Packets are dropped by a switch if no forwarding is configured.
We specify that  all data flows not yet at the egress switches are extended by transitions from a set $\mathit{Fwd}$ encoding the connections~$\mathit{Con}$ between switches (details of the encoding can be found in App.~\ref{appendix:motivation-data-plane}).
We obtain the following data flow formula:
\[
\A\ltlAlways\, (\bigwedge_{\texttt{e}\in\mathit{Egr}}\neg \texttt{e} \rightarrow \bigvee_{f\in\mathit{Fwd}} f)
\]
\begin{example}
Figure~\ref{fig:applications1} shows an example update that violates drop freedom.
Packets are forwarded from switch~\texttt{w} to switch~\texttt{z} either via switch~\texttt{x} or via switch~\texttt{y}.
If the forwarding of \texttt{x} is deactivated by firing transition~$\mathit{upd}2$ \emph{before} the forwarding of switch~\texttt{w} is updated by firing~$\mathit{upd}1$, then all packets still forwarded from \texttt{w} to \texttt{x} are dropped as no outgoing transitions from \texttt{x} will be enabled.
\end{example}
\noindent
{\bf Packet coherence.}
Packet coherence~\cite{DBLP:conf/pldi/BallBGIKSSV14} requires that every data flow follows one of two paths: either the path according to the routing before the update or the path according to the routing after the update. 
The paths \(\textit{Path}_1\) and \(\textit{Path}_2\)
are defined as the sets of switches of the forwarding route before and after the update. 
This results in the following data flow formula:
\[
\A ( \ltlAlways\, ( \bigvee_{\texttt{s}\in\textit{Path}_1}\texttt{s}) \lor \ltlAlways\, (\bigvee_{\texttt{s}\in\textit{Path}_2}\texttt{s}) )
\]
\begin{example}
In \refFig{applications2}, the encoding of an update to a double-diamond network topology~\cite{DBLP:conf/wdag/CernyFJM16} is depicted as a simple example for a packet incoherent update.
Before firing the update transitions $\mathit{upd}1$ and $\mathit{upd}2$, packets are forwarded via switches $\texttt{x}$, $\texttt{m}$, and $\texttt{u}$, after the complete update, via switches $\texttt{y}$, $\texttt{m}$, and $\texttt{v}$.
If $\texttt{m}$ is updated by firing transition $\mathit{upd}2$ while packets have been forwarded to $\texttt{x}$ then these packets are forwarded along the incoherent path $\texttt{x}$, $\texttt{m}$, and $\texttt{v}$.
\end{example}
We note that loop and drop freedom are incomparable requirements.
Together, they imply that all packets reach one egress switch.
Connectivity, in turn, implies drop freedom but not loop freedom, because an update can allow some loops.

\subsection{Run Formulas}

Data flow formulas require behavior on the maximal flow of packets and
switches are assumed to forward packets in a fair manner. 
Both types of assumptions are expressed in Flow-LTL as run formulas. 
We typically consider implications between run formulas and data flow formulas.\\
\noindent
{\bf Maximality.} A run $\runPN$ is \emph{interleaving-maximal} if, whenever some transition is enabled, some transition will be taken:
 $\runPN \models \ltlAlways\,( \bigvee_{t \in \transitions} \pre{}{t} \rightarrow 
                        \bigvee_{t \in \transitions} t)$.\\
A run $\runPN$ is \emph{concurrency-maximal} if, when a transition $t$ is from a moment on always
enabled, infinitely often a transition $t'$ (including $t$ itself) sharing a precondition with $t$ 
is taken:
 $\runPN \models \bigwedge_{t \in \transitions}
  ( \ltlEventually \ltlAlways\, \pre{}{t} \rightarrow 
    \ltlAlways \ltlEventually \bigvee_{{\small\begin{array}{c}
                             p\in \pre{}{t},
                             t'\in \post{}{p}
                           \end{array}}} t')$.
\noindent
{\bf Fairness.}
A run $\runPN$ is \emph{weakly fair} w.r.t.\ a transition $t$ if, whenever $t$ is always enabled after some point, $t$ is taken infinitely often:
{\(\runPN \models \ltlEventually \ltlAlways\, \pre{}{t} \rightarrow 
    \ltlAlways \ltlEventually t\)}.\\
A run $\runPN$ is \emph{strongly fair} w.r.t.\ $t$ if, whenever $t$ is enabled infinitely often, $t$ is taken infinitely often:
{\(\runPN \models \ltlAlways \ltlEventually\, \pre{}{t} \rightarrow \ltlAlways \ltlEventually t\)}.

\section{Model Checking Flow-LTL on Petri Nets with Transits}
\label{sec:generalization}
We solve the model checking problem of a Flow-LTL formula \(\phiRun\) on a Petri net with transits \(\pNet\) in three steps:
\begin{enumerate}
	\item \(\pNet\) is encoded as a Petri net \(\pNetMC\) without transits obtained by composing
		suitably modified copies of \(\pNet\) such that each flow subformula in \(\phiRun\) can be checked for correctness using the corresponding copy.
	\item \(\phiRun\) is transformed to an LTL-formula \(\phiMC\) which skips the uninvolved composition copies when evaluating run and flow parts, respectively.
	\item \(\pNetMC\) and  \(\phiMC\) are encoded in a circuit and fair reachability is checked 
with a hardware model checker
to answer if \(\pNet\models\phiRun\) holds.
\end{enumerate}
Given a Petri net with transits \(\petriNetFl\)
and a Flow-LTL formula~\(\phiRun\)
with subformulas \(\varphi_{F_i}=\A\,\phiLTL_i\), 
where \(i=1,\ldots,n\) for some \(n\in\N\),
we produce a
Petri net \(\pNetMC=(\plMC,\trMC,\flMC,\inhibitorFlMC,\initMC)\)
with inhibitor arcs (denoted by \(\inhibitorFlMC\))
and an LTL formula \(\phiMC\).
An \emph{inhibitor arc} is a directed arc from a place \(p\) to a transition \(t\), which only 
enables \(t\) if \(p\) contains no token. 
Graphically, those arcs are depicted as arrows equipped with a circle on their arrow tail.

\subsection{From Petri Nets with Transits to P/T Petri nets}
\label{sec:pntfl2pn}
We informally introduce the construction of \(\pNetMC\)
and Fig.~\ref{fig:algoOverview} visualizes the process by an example.
Details for this and the following constructions,
as well as all proofs corresponding to \refSection{generalization}
can be found in App.~\ref{appendixB}.

The \emph{original part} of \(\pNetMC\) (denoted by \(\pNetMCO\))
is the original net \(\pNet\) without transit relation and is used to check the run part of the formula.
To \(\pNetMCO\), a \emph{subnet} for each subformula \(\A\,\phiLTL_i\) of \(\phiRun\)
is composed (denoted by \(\pNetMCSub{i}\), with places \(\plMC_i\) and transitions \(\trMC_i\)), 
which serves for checking the corresponding data flow part of the formula \(\phiRun\).
The subnet introduces the possibility to decide for the tracking of up to one specific flow chain
by introducing a copy \(\subnet{p}_i\) of each place \(p\in\pl\)
and transitions simulating the transits. The place \(\initsub_i\) serves for starting the tracking.
Each run of a subnet simulates one possible flow chain of \(\pNet\),
i.e., every firing sequence covering any run of \(\pNet\) yields a flow chain. 

\begin{figure}
\centering
\scalebox{0.85}{
\begin{tikzpicture}[node distance=1.25cm,>=stealth',bend angle=45,auto,scale=0.8, on grid]%
\tikzset{
inhibitorArc/.style = {o->}
}
\node [envplace, tokens=1]	 			(a) [label=right:$\mathit{in}$]                  	   	{};
\node [transition, below=of a, yshift=-1cm]		(t) [label={[label distance=-1.5mm]above left:\(\mathit{t}\)}]	{};
\node [envplace, tokens=1, below=of t, yshift=-1cm] 	(b) [label=right:$\mathit{out}$]                   	 	{};
\node [envplace, right=of t, tokens=1,xshift=-2mm] 			(actO) [label={[label distance=-2mm]above left:\(\act{o}\)}]                   	 	{};
\node [transition, above=of a,xshift=0mm]		(create) [label={[label distance=0mm]above:\(\nameStartingTransition\)}]	{};

\path[->] 
	(t) edge[pre] (a)
	    edge[post] (a)
	    edge[pre] (b)
	    edge[post] (b)
	    edge[pre] (actO)

	(create) edge[pre] (a)
	    edge[post] (a)
	  (create.east)  edge[pre, bend left=25] (actO)
	
;

\node [envplace, right=of a, xshift=26mm]		(sa) [label=left:$\subnet{\mathit{in}}_1$]                  	   	{};
\node [envplace, right=of b, xshift=26mm] 		(sb) [label=left:$\subnet{\mathit{out}}_1$]                   	 	{};
\node [transition, below=of sa]				(t1) [label={[label distance=-1mm]below:}]	{\(t\)};
\node [transition, below left=of t1,yshift=2.5mm]	(t3) [label={[label distance=-2mm]below right:}]	{\(t\)};
\node [transition, above=of sb]				(t2) [label={[label distance=-1mm]above:}]	{\(t\)};
\node [transition, above right=of t2,yshift=-2.5mm]	(t4) [label={[label distance=-2mm]above left:\(t_{\skipT_1}\)}]			{\(t\)};
\node [envplace, below left=of t3]			(actt) [label=below:\({\act{t}}_1\)]                   	 	{};
\node [envplace, above=of actt, yshift=25mm]		(acttS) [label=below:\({\act{\nameStartingTransition}}_1\)]                   	 	{};

\node [transition, above=of sa,yshift=6mm]		(c1) [label={[label distance=-2mm]above right:}]	{\(\nameStartingTransition\)};
\node [transition, above=of sa,xshift=8mm]		(c2) [label={[label distance=-1mm]above:\(\nameStartingTransition_{\skipT_1}\)}]	{\(\nameStartingTransition\)};
\node [transition, above=of sa,xshift=-6mm, yshift=-3mm](c3) [label={[label distance=-1mm, yshift=2mm]right:}]	{\(\nameStartingTransition\)};

\node [envplace, tokens=1,above=of sa, yshift=15mm]		(inittfl)  [label=left:$\initsub_1$]                  	   	{};

\path[->] 
	  (t)  edge[bend right=15,post] (actt)
	  (create)  edge[bend right=0,post] (acttS);

\path[->] 
	(t1) edge[bend left=20, pre] (sa)
	     edge[bend right=20, post] (sa)
	     edge[bend right=50, pre] (actt)

	(t2) edge[bend left=20, pre] (sb)
	     edge[bend right=20, post] (sb)
	     edge[bend left=20, pre] (actt)

	(t3) edge[bend left=20, pre] (sa)
	     edge[bend right=20, post] (sb)
	     edge[bend right=20, pre] (actt)

	(sa) edge[bend left=20, inhibitorArc] (t4)
	(sb) edge[bend right=20, inhibitorArc] (t4)
	(t4) edge[bend left=0, pre] (actt);

\path[->] 
	(c1) edge[bend left=0, pre] (inittfl)
	     edge[bend right=0, post] (sa)
	     edge[bend right=30, pre] (acttS)

	(c3) edge[bend left=20, pre] (sa)
	     edge[bend right=20, post] (sa)
	     edge[bend left=0, pre] (acttS)

	(sa) edge[bend left=0, inhibitorArc] (c2)
	(c2) edge[bend right=15, pre] (acttS);

 \draw [fill] (7.95,-1.5) circle [radius=2pt];
 \draw [fill] (8.35,-1.5) circle [radius=2pt];
 \draw [fill] (8.75,-1.5) circle [radius=2pt];
\node [envplace, right=of actt, xshift=2.5cm]			(actt1) [label=below:\({\act{t}}_2\)]                   	 	{};
\node [envplace, above=of actt1, yshift=25mm]		(acttS1) [label=below:\({\act{\nameStartingTransition}}_2\)]                   	 	{};

\path[->]
	 (t1)  edge[bend left=30, post] (actt1)
	 (t2)  edge[bend right=10, post] (actt1)
	 (t3)  edge[bend left=20, post] (actt1)
	 (t4)  edge[bend right=0, post] (actt1)
;

\path[->]
	 (c1)  edge[bend left=30, post] (acttS1)
	 (c2)  edge[bend right=0, post] (acttS1)
	 (c3)  edge[bend right=20, post] (acttS1)
;

\node [envplace, right=of sa, xshift=42.5mm]		(sa1) [label=left:$\subnet{\mathit{in}}_n$]                  	   	{};
\node [envplace, right=of sb, xshift=42.5mm] 		(sb1) [label=left:$\subnet{\mathit{out}}_n$]                   	 	{};
\node [transition, below=of sa1]			(t11) [label={[label distance=-1mm]below:\(\mathit{t_1}\)}]	{\(t\)};
\node [transition, below left=of t11,yshift=2.5mm]	(t31) [label={[label distance=-1mm]right:\(\mathit{t_2}\)}]	{\(t\)};
\node [transition, above=of sb1]			(t21) [label={[label distance=-1mm]above:\(\mathit{t_3}\)}]	{\(t\)};
\node [transition, above right=of t21,yshift=-2.5mm]	(t41) [label={[label distance=-2mm]above left:\(t_{\skipT_n}\)}]	{\(t\)
																};
\node [envplace, below left=of t31]			(acttn) [label=below:\({\act{t}}_n\)]                   	 	{};

\node [envplace, above=of acttn, yshift=25mm]		(acttSn) [label=below:\({\act{\nameStartingTransition}}_n\)]                   	 	{};

\node [transition, above=of sa1,yshift=6mm]		(c11) [label={[label distance=-1mm]right:\(\nameStartingTransition_1\)}]	{\(\nameStartingTransition\)};
\node [transition, above=of sa1,xshift=8mm]		(c21) [label={[label distance=-1mm]above:\(\nameStartingTransition_{\skipT_n}\)}]	{\(\nameStartingTransition\)};
\node [transition, above=of sa1,xshift=-6mm, yshift=-3mm](c31) [label={[label distance=-1mm]right:\(\nameStartingTransition_2\)}]	{\(\nameStartingTransition\)};

\node [envplace, tokens=1,above=of sa1, yshift=15mm]		(inittfl1)  [label=left:$\initsub_n$]                  	   	{};

\path[->] 
	(t11) edge[bend left=20, pre] (sa1)
	     edge[bend right=20, post] (sa1)
	     edge[bend right=50, pre] (acttn)

	(t21) edge[bend left=20, pre] (sb1)
	     edge[bend right=20, post] (sb1)
	     edge[bend left=20, pre] (acttn)

	(t31) edge[bend left=20, pre] (sa1)
	     edge[bend right=20, post] (sb1)
	     edge[bend right=20, pre] (acttn)

	(sa1) edge[bend left=20, inhibitorArc] (t41)
	(sb1) edge[bend right=20, inhibitorArc] (t41)
	(t41) edge[bend left=0, pre] (acttn);
\path[->] 
	(c11) edge[bend left=0, pre] (inittfl1)
	     edge[bend right=0, post] (sa1)
	     edge[bend right=30, pre] (acttSn)

	(c31) edge[bend left=20, pre] (sa1)
	     edge[bend right=20, post] (sa1)
	     edge[bend left=0, pre] (acttSn)

	(sa1) edge[bend left=0, inhibitorArc] (c21)
	(c21) edge[bend right=15, pre] (acttSn);

\path[->, dashed,opacity=0.2]
	 (t11)  edge[bend left=50, post] (actO)
	 (t21)  edge[bend right=40, post] (actO)
	 (t31)  edge[bend left=50, post] (actO)
	 (t41)  edge[bend right=40, post] (actO)
;
\path[->, dashed,opacity=0.2]
	 (c11)  edge[bend right=10, post] (actO)
	 (c21)  edge[bend right=10, post] (actO)
	 (c31)  edge[bend right=10, post] (actO)
;

\node [envplace,left=of a, tokens=1, xshift=-16mm]	 			(aO) [label={[label distance=-1mm]right:$\mathit{in}$}]                  	   	{};
\node [transition,left=of t, xshift=-16mm]				(tO) [label={[label distance=0mm]right:\(\mathit{t}\)}]	{};
\node [envplace, left=of b, tokens=1, xshift=-16mm] 			(bO) [label={[label distance=-1mm]right:$\mathit{out}$}]                   	 	{};
\node [transition, left=of create, xshift=-16mm]		(createO) [label={[label distance=0mm]above:\(\nameStartingTransition\)}]	{};

\path[control] 
		(aO) edge[<->] (createO)
		(aO) edge[<->] (tO)
		(tO) edge[<->] (bO)
;
\path[flowC] 
		([xshift=-2mm]createO.south) edge ([xshift=-2mm]aO.north)
;
\path[flowA]
	 ([xshift=2mm]createO.south) edge[<->] ([xshift=2mm]aO.north)

	 ([xshift=2mm]aO.south) edge[<->] ([xshift=2mm]tO.north)
	 ([xshift=2mm]tO.south) edge ([xshift=2mm]bO.north)
;
\path[flowB]
	 ([xshift=-2mm]tO.south) edge[<->] ([xshift=-2mm]bO.north)
;
\node [above=of createO, yshift=6mm, xshift=-.5mm] (oO) {\color{gray}\(\pNet\)};
\draw [->,decorate,decoration={snake,post length=2.5mm},ultra thick] (-2.8,-1.5) -- (-0.8,-1.5) ;

\node [above=of create, yshift=6mm, xshift=-7.2mm] (n) {\color{gray}\(\pNetMC\)};
\node [above=of create, yshift=2mm, xshift=-1mm] (o) {\color{orange}\(\pNetMCO\)};
\node [right=of o, xshift=13mm, yshift=0mm] (phi1) {\color{DarkBlue}\(\pNetMCSub{1}\)};
\node [right=of phi1, xshift=44mm, yshift=0mm] (phi2) {\color{ganttGreen}\(\pNetMCSub{n}\)};
%
\begin{pgfonlayer}{background}
\draw [-, rectangle,rounded corners,gray,fill=gray!15,pattern=north east lines,opacity=0.4,pattern color=gray]
 	([xshift=-2mm,yshift=-4mm]bO.south west) --
	([xshift=8mm,yshift=-4mm]bO.south east) --
	([xshift=-69.8mm,yshift=4mm]inittfl.north west) --
	([xshift=-85mm,yshift=4mm]inittfl.north west) --
	cycle;
\draw [-, rectangle,rounded corners,gray,fill=gray!15,pattern=north west lines,opacity=0.4,pattern color=gray]
 	([xshift=-10.5mm,yshift=-4mm]b.south west) --
	([xshift=15mm,yshift=-4mm]sb1.south east) --
	([xshift=20.6mm,yshift=4mm]inittfl1.north west) --
	([xshift=-58.6mm,yshift=4mm]inittfl.north west) --
	cycle;

\draw [-, rectangle,rounded corners,orange,fill=orange!15]
 	([xshift=-2.5mm,yshift=-2mm]b.south west) --
	([xshift=15mm,yshift=-2mm]b.south east) --
	([xshift=-27.6mm,yshift=2mm]inittfl.north west) --
	([xshift=-50.6mm,yshift=2mm]inittfl.north west) --
	cycle;

\draw [-, rectangle,rounded corners,DarkBlue,fill=cdc_BlueL!15]
 	([xshift=-24.5mm,yshift=-2mm]sb.south west) --
	([xshift=26.5mm,yshift=-2mm]sb.south east) --
	([xshift=32mm,yshift=2mm]inittfl.north west) --
	([xshift=-24.5mm,yshift=2mm]inittfl.north west) --
	cycle;
\draw [-, rectangle,rounded corners,ganttGreen,fill=cdc_GreenL!15]
 	([xshift=-24.5mm,yshift=-2mm]sb1.south west) --
	([xshift=13mm,yshift=-2mm]sb1.south east) --
	([xshift=18.6mm,yshift=2mm]inittfl1.north west) --
	([xshift=-24.5mm,yshift=2mm]inittfl1.north west) --
	cycle;
\end{pgfonlayer}
\end{tikzpicture}
}
\caption{
An overview of the constructed P/T Petri net \({\pNetMC}\) (on the right) for
an example Petri net with transits \(\pNet\) (on the left)
and \(n\) flow subformulas \(\A\,\phiLTL_i\).
}
\label{fig:algoOverview}
\end{figure}
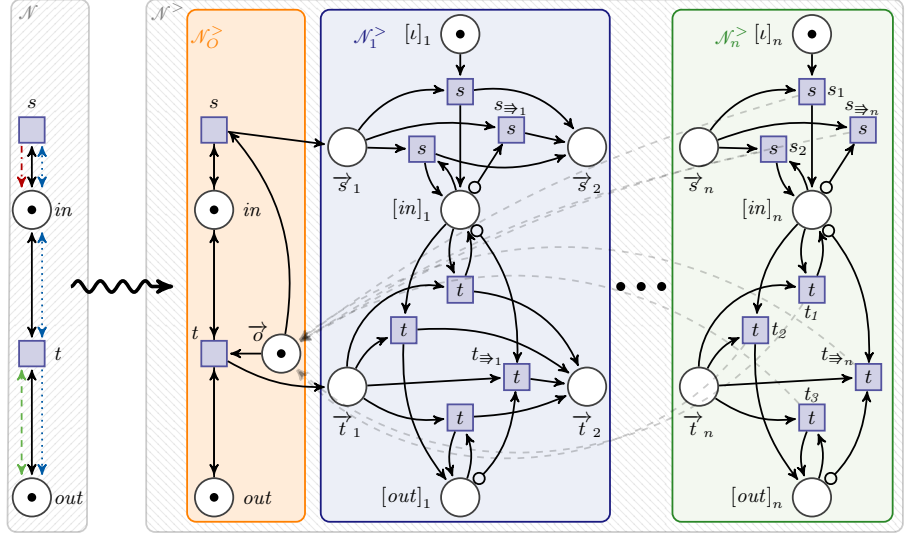

An \emph{activation token} iterates sequentially through these components via places \(\act{t}\) for \(t\in\tr\).
In each step, the active component has to fire exactly one transition and 
pass the active token to the next component.
The sequence starts by \(\pNetMCO\) firing a transition \(t\)
and proceeds through every subnet simulating the data flows according to the transits of \(t\).
This implies that the subnets have to either move their data flow via a \(t\)-labelled transition \(t'\) (\(\lambda(t')=t\))
or use the skipping transition \(t_{\skipT_i}\) if their chain is not involved in the firing of \(t\)
or a newly created chain should not be considered in this run.

\begin{lemma}[Size of the Constructed Net]
\label{lem:sizePN}
The constructed Petri net \(\pNetMC\) has
\(\oclass{|\pNet|\cdot n + |\pNet|}\) places 
and \(\oclass{|\pNet|^3\cdot n + |\pNet|}\) transitions.
\end{lemma}

\subsection{From Flow-LTL Formulas to LTL Formulas}
\label{sec:flowltl2ltl}
The two different kinds of timelines of \(\phiRun\) are encoded in the LTL formula \(\phiMC\).
On the one hand, the data flow formulas \(\A\,\phiLTL_i\) in \(\phiRun\)
are now checked on the corresponding subnets \(\pNetMC_i\)
and, on the other hand, the run formula part of \(\phiRun\) is checked on the original part of the net \(\pNetMC_O\).
In both cases, we need to ignore places and transitions from other parts of the composition.
This is achieved by replacing each next operator \(\ltlNext\,\phi\) and atomic proposition \(t\in\tr\) inside \(\phiRun\)
with an until operator.
Transitions which are not representing the considered timeline are called \emph{unrelated}, others \emph{related}.
Via the until operator, all unrelated transitions can fire until a related transition is fired.
This is formalized in \refTable{formulas} using the sets
\(O=\trMC\setminus\tr\) and 
\(O_i=(\trMC\setminus\trMCSub{i})\cup\{t_{\skipT_i}\in\trMCSub{i}\with t\in\tr\}\),
for the unrelated transitions of the original part and of the subnets, respectively.
The related transitions of the original part are given by \(\tr\) and for the subnets by
\(M_i(t)=\{t'\in\trMCSub{i}\setminus\{t_{\skipT_i}\}\with \lambda(t')=t\}\)
and \(M_i=\trMCSub{i}\setminus\{t_{\skipT_i}\in\trMCSub{i}\with t\in\tr\}\).
\begin{table}
\caption{
Row 1 considers the substitutions in the run part of \(\phiRun\),
row 2 the substitutions in each subformula \(\varphi_{F_i}\).
Column 1 considers simultaneously substitutions,
column 2 substitutions from the inner- to the outermost occurrence.
} 
\label{tab:formulas}
\centering
\begin{tabular}{l|l}
\(t\in\tr\)  & \(\ltlNext\,\phi\) \\\hline
\((\bigvee_{t'\in{}O} t')\U t\) & \(((\bigvee_{t\in{}O} t)\U ((\bigvee_{t'\in\tr} t')\wedge\ltlNext\,\phi))\vee (\ltlAlways\, (\neg(\bigvee_{t'\in\tr} t'))\wedge\phi)\)\\
\((\bigvee_{t_o\in O_i} t_o)\U (\bigvee_{t_m\in M_i(t)}t_m)\) & \(((\bigvee_{t\in O_i} t)\U ((\bigvee_{t\in M_i}t)\wedge\ltlNext\,\phi))\vee (\ltlAlways\, (\neg(\bigvee_{t\in M_i}t))\wedge\phi)\)
\end{tabular}
\end{table}

Additionally, every atomic proposition \(p\in\pl\)
in the scope of a flow operator is simultaneously substituted with its corresponding place \(\subnet{p}_i\) of the subnet.
Every flow subformula \(\A\,\phiLTL_i\) is substituted with
\(\ltlAlways\initsub_i\vee(\initsub_i\ltlUntil(\neg \initsub_i\wedge\phiLTL_i'))\),
where \(\ltlAlways\initsub_i\) represents that no flow chain is tracked and
\(\phiLTL_i'\) is the result of the substitutions of atomic propositions and next operators described before.
The until operator in \(\initsub_i\ltlUntil(\neg \initsub_i\wedge\phiLTL_i')\) ensures to only check
the flow subformula at the time the chain is created.
Finally, restricting runs to not end in any of the subnets yields the final formula
\[\phiMC=(\ltlAlways\ltlEventually\act{o})\rightarrow \phiRun^\A\]
with 
\(\act{o}\) being the activation place of the original part of the net
and \(\phiRun^\A\) the result of the substitution of all flow subformulas.

\begin{lemma}[Size of the Constructed Formula]
\label{lem:sizeFormula}
The size of the constructed LTL formula \(\phiMC\) is in \(\oclass{|\pNet|^3\cdot n \cdot |\phiRun| + |\phiRun|}\).
\end{lemma}

\begin{lemma}[Correctness of the Transformation]
\label{lem:correctnessPNLTL}
For a Petri net with transits \(\pNet\) and a Flow-LTL formula \(\phiRun\),
there exists a safe P/T Petri net \(\pNetMC\) with inhibitor arcs and an LTL formula \(\phiMC\)
such that
\(\pNet \models \phiRun\text{ iff } \pNetMC\models_\mathtt{LTL}\phiMC\).
\end{lemma}

\subsection{Petri Net Model Checking with Circuits}
\label{sec:pnMCwithCircuits}
We translate the model checking of an LTL formula \(\phiLTL\) with places and transitions as atomic propositions on 
a safe P/T Petri net with inhibitor arcs \(\pNet\) to a model checking problem on a circuit.
We define the circuit \(\circuit_\pNet\) simulating \(\pNet\) and an adapted formula \(\phiLTL'\),
which can be checked by modern model checkers~\cite{DBLP:conf/fmcad/ClaessenES13,DBLP:conf/cav/FinkbeinerRS15,abcAsTheyWantIt}.

A \emph{circuit} \(\circuit=(\cin,\cout,\clatches,\cformula)\)
consists of boolean variables \(\cin\), \(\cout\), \(\clatches\) for input, output, latches,
and a boolean formula \(\cformula\) over \(\cin\times \clatches \times \cout\times \clatches\),
which is deterministic in \(\cin\times \clatches\).
The formula \(\cformula\) can be seen as transition relation from a valuation of the input variables
and the current state of the latches to the valuation of the output variables and the next state of the latches.
A circuit \(\circuit\) can be interpreted as a Kripke structure
such that the satisfiability of a formula \(\phiLTL'\)
(denoted by \(\circuit\models \phiLTL'\)) can be defined by the satisfiability in the Kripke structure.

The desired circuit \(\circuit_\pNet\) has
a latch for each place \(p\in\pl\) to store the current \emph{marking}, 
a latch \(\mathtt{i}\) for \emph{initializing} this marking with \(\init\) in the first step, and
a latch \(\mathtt{e}\) for handling \emph{invalid} inputs.
The inputs \(\cin\) consider the firing of a transition \(t\in\tr\).
The latch \(\mathtt{i}\) is true in every but the first step.
The latch \(\mathtt{e}\) is true whenever invalid values are applied on the inputs, i.e.,
the firing of not enabled, or more than one transition.
The marking latches are updated according to the firing of the valid transition.
If currently no valid input is applied, the marking is kept from the previous step.
There is an output for each place (the current marking),
for each transition (the transition leading to the next marking), and for the current value of the invalid latch.
We create \(\phiLTL'\) by skipping the initial step and allowing invalid inputs only at the end of a trace:
\[\phiLTL'=\ltlNext\,(\ltlAlways(\mathtt{e}\rightarrow\ltlAlways \mathtt{e})\rightarrow\phiLTL).\]
This allows for finite firing sequences.
The concrete formula \(\cformula\), the Kripke structure, and the corresponding proofs can be found in App.~\ref{appendixB}.
The circuit \(\circuit_\pNet\) can be encoded as an and-inverter graph in the Aiger format~\cite{Biere-FMV-TR-11-2}.
\begin{lemma}[Correctness of the Circuit]
\label{lem:corCircuit}
For a safe P/T Petri net with inhibitor arcs \(\pNet\) and an LTL formula \(\phiLTL\), 
there exists a circuit \(\circuit_\pNet\)
with \(|\pl|+2\) latches and \(\oclass{|\pNet|^2}\) gates, and \(\phiLTL'\) 
of size \(\oclass{|\phiLTL|}\)
such that \(\pNet\models_\mathtt{LTL} \phiLTL\)~iff~\(\circuit_\pNet\models \phiLTL'\).
\end{lemma}
\begin{theorem}
\label{theo:correctness}
A safe Petri net with transits \(\pNet\) can be checked against a Flow-LTL formula \(\phiRun\)
in single-exponential time in the size of \(\pNet\) and \(\phiRun\).
\end{theorem}
Checking a safe Petri net with transits against Flow-LTL has a PSPACE-hard lower bound because checking a safe Petri net against LTL is a special case of this problem and reachability of safe Petri nets is PSPACE-complete.

\section{Implementation Details and Experimental Results}
\label{sec:experimentalResults}
We implemented our model checking approach in a prototype tool based on the tool \textsc{Adam}~\cite{DBLP:conf/cav/FinkbeinerGO15}. 
Our tool takes as input a Flow-LTL specification and a Petri net with transits, and carries out the transformation described in \refSection{generalization} to
obtain an LTL formula and an Aiger circuit.
We then use MCHyper \cite{DBLP:conf/cav/FinkbeinerRS15} to combine the circuit and the LTL formula into another Aiger circuit.
MCHyper is a verification tool for HyperLTL~\cite{DBLP:conf/post/ClarksonFKMRS14}, which subsumes LTL.
The actual model checking is carried out by the hardware model checker ABC \cite{abcAsTheyWantIt}.
ABC provides a toolbox of state-of-the-art verification and falsification techniques like 
IC3~\cite{DBLP:conf/vmcai/Bradley11}/PDR \cite{DBLP:conf/fmcad/EenMB11}, interpolation (INT)~\cite{DBLP:conf/sas/McMillan03},
and bounded model checking~\cite{DBLP:conf/cav/BiereCRZ99} (BMC, BMC2, BMC3).
We prepared an artifact to replicate our experimental results~\cite{GiesekingH19}.

Our experimental results cover two benchmark families (SF/RP) and a case study (RU) from software-defined networking
on
real-world network topologies:
\\\textbf{Switch Failure} (SF) (\emph{Parameter:} \(n\) switches):
		From a sequence of $n$ switches with the ingress at the beginning and the egress at the end, a failing switch is chosen at random and removed. Then, data flows are bypassed from the predecessor to the successor of the failing switch.
		 Every data flow reaches the egress node no matter of the update (connectivity).\\
\textbf{Redundant Pipeline} (RP) (\emph{Parameters:} \(n_1\) switches in pipeline one / \(n_2\) switches in pipeline two / \(v\) version):
		The \emph{base version (B)} contains two disjoint sequences of switches from the ingress to the egress,
		possibly with differing length.
		For this and the next two versions, it is required that each data flow reaches the egress node (connectivity) and is only forwarded via the first or the second pipeline (packet coherence).
		\emph{Update version (U):}
		Two updates are added that can concurrently remove the first node of any pipeline and return the data flows to the ingress.
		If both updates happen, data flows do not reach the egress.
		Returning the data flows violates packet coherence.
		\emph{Mutex version (M):}
		A mutex is added to the update version such that at most one pipeline can be broken.
		Updates can happen sequentially such that data flows are in a cycle through the ingress.
		\emph{Correct version (C):}
		The requirements are weakened such that each data flow only has to reach the egress when updates do not occur infinitely often.\\
\textbf{Routing Update} (RU) is a case study based on realistic software-defined networks.
		We picked 31 real-world network topologies from~\cite{DBLP:journals/jsac/KnightNFBR11}.
		For each network, we choose at random an ingress switch, an egress switch, and a loop- and drop-free initial configuration between the two.
		For a different, random final configuration, we build a sequential update in reverse from egress to ingress.
		The update overlaps with the initial configuration at some point during the update or is activated from the ingress in the last step.
		It is checked if all packets reach the egress 
		(T) and if all packets reach another specific switch as an egress 
		(F).

{
\setlength{\tabcolsep}{1pt}
\begin{table}
\caption{Experimental results from the benchmark families {Switch Failure} (SF) and Redundant Pipeline (RP), and the case study {Routing Update} (RU). The results are the average over five runs on an Intel i7-2700K CPU with 3.50~GHz, 32~GB RAM, and a timeout of 30~minutes.}
  \label{table:benchmarks}
  \centering
\begin{tabular}{c|c|r|rrr|rrr|rr||rr||c}
& & &\multicolumn{3}{c|}{PNwT} & \multicolumn{3}{c|}{Translated PN} & \multicolumn{2}{c||}{Circuit} & \multicolumn{2}{c||}{Result}\\
Ben. & Par. & \(\#S\) &\(|\pl|\) & \(|\tr|\) & \(|\phiRun|\) & \(|\plMC|\)  & \(|\trMC|\)  & \(|\phiLTL'|\)  & Lat.  & Gat.  & Sec. & Algo. & \(\models\)\\\hline
SF 
   &   3  & 4 & 4 & 5 & 35 & 17 & 22 & 60 & 90 & 2796 &  2.7 & IC3 & \cmark \\
   &\(\cdots\)&\multicolumn{9}{c||}{\(\cdots\)} & \multicolumn{2}{c||}{\(\cdots\)} &\\
   &   9  & 10 & 10 & 11 & 95 & 35 & 46 & 138 & 186 & 8700 & 1359.9 & IC3 & \cmark \\
   &   10  & 11 & 11 & 12 & 105 & 38 & 50 & 151 & 202 & 9964 & TO & - & ? \\\hline 
RP  
  &    1/1/B   &  4  &  4 &  5  &  43  &  17  &  22  &  68  &  100  &  2989   &   4.0   & IC3 &   \cmark \\
  &\(\cdots\)&\multicolumn{9}{c||}{\(\cdots\)} & \multicolumn{2}{c||}{\(\cdots\)} &\\
  &    4/4/B   &  10  &  10 &  11  &  103  &  35  &  46  &  146  &  196  &  8893   &   646.4   & IC3 &   \cmark \\
  &    4/5/B   &  11  &  11 &  12  &  113  &  38  &  50  &  159  &  212  &  10157   &   TO   & - &   ? \\
\cline{2-13}
&   1/1/U   & 6 & 6 & 9 & 63 & 25 & 36 & 100 & 136 & 5535 &   1.6   & BMC2 &   \xmark \\
&\(\cdots\)&\multicolumn{9}{c||}{\(\cdots\)} & \multicolumn{2}{c||}{\(\cdots\)} &\\
&   5/4/U   & 13 & 13 & 16 & 133 & 46 & 64 & 191 & 248 & 14523 &   945.1   & BMC3 &   \xmark \\
&   5/5/U   & 14 & 14 & 17 & 143 & 49 & 68 & 204 & 264 & 16127 &   TO   & - &   ? \\
\cline{2-13}
&   1/1/M  & 6 & 9 & 11 & 63 & 30 & 42 & 106 & 146 & 6908 &   8.1   & BMC3 &   \xmark \\
&\(\cdots\)&\multicolumn{9}{c||}{\(\cdots\)} & \multicolumn{2}{c||}{\(\cdots\)} &\\
&   4/3/M  & 11 & 14 & 16 & 113 & 45 & 62 & 171 & 226 & 13573 &   1449.6   & BMC2 &   \xmark \\
&   4/4/M  & 12 & 15 & 17 & 123 & 48 & 66 & 184 & 242 & 15146 &   TO   & - &   ? \\
\cline{2-13}
  	&    1/1/C   & 6 &  9  &  11  &  70  &  30  &  42  &  113  &  151  &  7023   &   63.1   &    IC3 &   \cmark \\
   	&\(\cdots\)&\multicolumn{9}{c||}{\(\cdots\)} & \multicolumn{2}{c||}{\(\cdots\)} &\\
  	&    3/3/C   & 10&  13  &  15  &  110  &  42  &  58  &  165  &  215  &  12195   &   1218.0   &    IC3 &   \cmark \\
  	&    3/4/C   & 11&  14  &  16  &  120  &  45  &  62  &  178  &  231  &  13688   &    TO   &    - &   ? \\
\cline{1-13}
RU
  &    Arpanet196912T  & 4  & 14 & 10 & 117 & 31 & 39 & 154 & 188 & 7483  &   22.7   & IC3  &  \cmark \\ 
  &    Arpanet196912F  & 4  & 14 & 10 & 117 & 31 & 39 & 154 & 188 & 7483  &   2.0    & BMC3 &   \xmark \\
  &    NapnetT         & 6  & 23 & 17 & 199 & 48 & 64 & 254 & 292 & 15875 &   95.1  & IC3 &   \cmark \\ 
  &    NapnetF         & 6  & 23 & 17 & 199 & 48 & 64 & 254 & 292 & 15875 &   4.7   & BMC3 &   \xmark \\ 
  &\(\cdots\)&\multicolumn{9}{c||}{\(\cdots\)} & \multicolumn{2}{c||}{\(\cdots\)} &\\
  &    NetrailT        & 7  & 30 & 23 & 271 & 62 & 88 & 344 & 380 & 26101 & 145.3 & IC3 &   \cmark \\ 
  &    NetrailF        & 7  & 30 & 23 & 271 & 62 & 88 & 344 & 380 & 26101 &  58.3 & BMC3 &   \xmark \\
  &    Arpanet19706T  & 9  & 33 & 24 & 281 & 67 & 89 & 354 & 400 & 27619 & 507.8 & IC3 &   \cmark \\ 
  &    Arpanet19706F  & 9  & 33 & 24 & 281 & 67 & 89 & 354 & 400 & 27619 & 49.7 & BMC3 &   \xmark \\
  &    NsfcnetT        & 10 & 31 & 22 & 261 & 65 & 87 & 334 & 376 & 26181 &    304.8   & IC3 &  \cmark  \\ 
  &    NsfcnetF        & 10 & 31 & 22 & 261 & 65 & 87 & 334 & 376 & 26181 &   8.4   & BMC3 &   \xmark \\
\cline{2-13}
  &\(\cdots\)&\multicolumn{9}{c||}{\(\cdots\)} & \multicolumn{2}{c||}{\(\cdots\)} &\\
  &    TwarenF         & 20 & 65 & 45 & 531 & 130 & 170 & 664 & 736 & 87493 & 461.5 & BMC3 &   \xmark \\
  &    MarnetF         & 20 & 77 & 57 & 679 & 156 & 224 & 854 & 908 & 138103 & 746.1 & BMC3 &   \xmark \\
  &    JanetlenseF     & 20 & 91 & 71 & 847 & 184 & 280 & 1064 & 1104 & 203595 & 514.2 & BMC2 &   \xmark \\
  &    HarnetF         & 21 & 71 & 50 & 593 & 143 & 193 & 744 & 812 & 108415 & 919.0 & BMC3 &   \xmark \\
  &    Belnet2009F     & 21 & 71 & 50 & 597 & 145 & 199 & 754 & 816 & 113397 & 1163.3 & BMC2 &   \xmark \\
  &\(\cdots\)&\multicolumn{9}{c||}{\(\cdots\)} & \multicolumn{2}{c||}{\(\cdots\)} &\\
  &    UranF           & 24 & 56 & 38 & 449 & 106 & 133 & 552 & 618 & 57950 & 143.2 & BMC3 &   \xmark \\
  &    KentmanFeb2008F & 26 & 82 & 56 & 669 & 167 & 223 & 844 & 920 & 142291 &  111.2 & BMC3 &   \xmark \\
  &    Garr200212F     & 27 & 86 & 59 & 703 & 174 & 232 & 884 & 964 & 153509 & 324.2 & BMC3 &   \xmark \\
  &    IinetF          & 31 & 104 & 73 & 871 & 210 & 288 & 1094 & 1176 & 227153 & 1244.5 & BMC3 &   \xmark \\
  &    KentmanJan2011F & 38 & 117 & 79 & 943 & 236 & 312 & 1184 & 1288 & 269943 &  112.6 & BMC3 &   \xmark \\
\cline{1-13}
 \end{tabular}
\end{table}
}

Table~\ref{table:benchmarks} presents our experimental results and 
indicates for each benchmark the model checking approach with the best performance (cf.\ App.~\ref{appendix:com_ex_results} for the full table).
In the benchmarks where the specification is satisfied (\cmark), IC3 is the clear winner,
in benchmarks where the specification is violated (\xmark), the best approach is bounded model checking with dynamic unrolling (BMC2/3).
The results are encouraging: hardware model checking is effective for circuits constructed by our transformation with 
up to 400 latches and 27619 gates; falsification is possible for larger circuits with up to 1288 latches and 269943 gates. 
As a result, we were able to automatically verify with our prototype implementation updates for networks with topologies of 
up to 10 switches (\(\# S\)) and to falsify updates for topologies with up to 38~switches within the time bound of 30 minutes.

We investigated the cost of specifications drop and loop freedom compared with connectivity and packet coherence.
Table~\ref{table:benchmarks2} exemplarily shows the results for network topology \emph{Napnet} from RU.
Connectivity, packet coherence, and loop freedom have comparable runtime due to similar formula and circuit sizes.
Drop freedom is defined over transitions and, hence, expensive for our transformation.

\begin{table}
\caption{For the network topology Napnet and a concurrent update between two randomly generated topologies, our four standard requirements are checked.}
  \label{table:benchmarks2}
  \centering
\begin{tabular}{c|c|ccc|ccc|cc||cc||c}
& & \multicolumn{3}{c|}{PN w.\ Transits} & \multicolumn{3}{c|}{Translated PN} & \multicolumn{2}{c||}{Circuit} & \multicolumn{2}{c||}{Result}\\
Ben. & Req. &\(|\pl|\) & \(|\tr|\) & \(|\phiRun|\) & \(|\plMC|\)  & \(|\trMC|\)  & \(|\phiLTL'|\)  & Latches  & Gates  & Sec. & Algo. & \(\models\)\\\hline
Napnet &  connectivity & 23 & 17 & 199 & 48 & 64 & 254 & 292 & 15875 &  95.1  & IC3  &  \cmark \\
 &  p.\ coherence  & 23 & 17 & 208 & 48 & 64 & 267 & 298 & 16041 &   31.9  & IC3  &  \cmark \\
 &   loop-free  & 23 & 17 & 237 & 48 & 64 & 296 & 305 & 16289 &   52.6  &  INT  &  \cmark \\
 &   drop-free  & 23 & 17 & 257 & 48 & 64 & 2288 & 325 & 30449 &   165.9  &  IC3  &  \cmark \\
\cline{1-13}
 \end{tabular}
\end{table}

\section{Related Work}
\label{sec:related-work}
There is a large body of work on software-defined networks, see~\cite{DBLP:journals/pieee/KreutzRVRAU15} for a good introduction.
Specific solutions that were proposed for the network update problem include \emph{consistent
updates}~\cite{DBLP:conf/sigcomm/ReitblattFRSW12,DBLP:conf/wdag/CernyFJM16} (cf.\ the introduction), \emph{dynamic scheduling}~\cite{Jin:2014:DSN:2619239.2626307}, and \emph{incremental updates}~\cite{Katta:2013:ICU:2491185.2491191}.
Model checking, including both explicit and SMT-based approaches, has
previously been used to verify software-defined
networks~\cite{Canini:2012:NWT:2228298.2228312,6987609,Mai:2011:DDP:2043164.2018470,DBLP:conf/icnp/WangMLTS13,DBLP:conf/pldi/BallBGIKSSV14,DBLP:conf/popl/PadonIKLSS15}.
Closest
to our work are models of networks as Kripke structures 
to use model checking for synthesis of correct network
updates~\cite{DBLP:conf/cav/El-HassanyTVV17,DBLP:conf/cav/McClurgHC17}. While
they pursue \emph{synthesis}, rather than verification
of network updates, the approach is still based on a model checking
algorithm that is called in each step of the construction of a
sequence of updates.  
The model checking
subroutine of the synthesizer assumes that each packet sees at most
one switch that was updated after the packet entered the network.
This restriction is implemented with explicit waits,
which can afterwards often be removed by heuristics.
Our model checking routine does not require this assumption. As it therefore
allows for more general updates, it would be very interesting to add the
new model checking algorithm into the synthesis procedure.
Flow correctness also plays a role in other application areas like \emph{access control} in physical spaces. Flow properties that are of interest in this setting, such as
``from every room in the building there is a path to exit the building'', have been formalized in a temporal logic~\cite{7536393}.

There is a significant number of model checking tools~%
(e.g., \cite{DBLP:conf/apn/Schmidt00,DBLP:conf/tacas/Thierry-Mieg15,DBLP:conf/tacas/KantLMPBD15}) 
for Petri nets and an annual model checking contest~\cite{mcc:2019}.
In this contest, however, only LTL formulas with places as atomic propositions are checked.
To the best of our knowledge, other model checking tools for Petri nets do not provide places and transitions as
atomic propositions.
Our encoding needs to reason about places and transitions to pose fairness conditions on the firing of transitions.

\section{Conclusion}
\label{sec:conclusion}
We have presented a model checking approach for the verification of data flow correctness in networks during concurrent updates of the network configuration.
Key ingredients of the approach are Petri nets with transits, which superimpose the transit relation of data flows onto the flow relation of Petri nets,
and Flow-LTL, which combines the specification of local data flows with the specification of global control.
The model checking problem for Petri nets with transits and Flow-LTL specifications reduces to a circuit model checking problem.
Our prototype tool implementation can verify and falsify realistic concurrent updates of software-defined networks with specifications like packet coherence.

In future work, we plan to extend this work to the synthesis of concurrent updates. 
Existing synthesis techniques use model checking as
a subroutine to verify the correctness of the individual update steps~\cite{DBLP:conf/cav/El-HassanyTVV17,DBLP:conf/cav/McClurgHC17}. 
We plan to study Flow-LTL specifications in the setting of Petri games~\cite{DBLP:journals/iandc/FinkbeinerO17}, which describe the existence of controllers for
asynchronous distributed processes. 
This would allow us to synthesize concurrent network updates without a central controller.

\bibliographystyle{splncs04}
\bibliography{ms}

\appendix
\newpage

\section*{Appendix}

\section{Encoding of Concurrent Updates for SDN}
\label{appendix:motivation}

In this section of the appendix, we outline how updates to a configured network topology can be encoded into Petri nets with transits.
First, we recall network topologies, network configurations, and concurrent updates.
Second, we encode a network topology as data plane into Petri nets with transits.
Third, we encode a corresponding network configuration and a concurrent update to this configuration into usual Petri nets (without transits).

\subsection{Network Topology, Configurations, and Updates}

A \emph{network topology} $T$ is given as a finite, connected, undirected graph $T=(\mathit{Sw}, \mathit{Con})$ with the non-empty set of switches $\mathit{Sw}$ as vertices and the non-empty set of connections $\mathit{Con} \subseteq \mathit{Sw} \times \mathit{Sw}$ as edges.
The \emph{configuration} of a network topology is defined by a static NetCore~\cite{DBLP:conf/icfp/FosterHFMRSW11,DBLP:conf/nsdi/MonsantoRFRW13} program with the following syntax: 

$\texttt{ingress = } \mathit{Ingr}\texttt{;}$

$\texttt{forwarding}\texttt{;}$

$\texttt{egress = } \mathit{Egr}\texttt{;}$\\
such that $\mathit{Ingr} \subseteq \mathit{Sw}$ is the non-empty set of \emph{ingress switches} where packets enter the network topology and $\mathit{Egr} \subseteq \mathit{Sw}$ is the non-empty set of \emph{egress switches} to which packets should be forwarded. 
It is required  that $\mathit{Ingr} \cap \mathit{Egr} = \emptyset$.
$\texttt{forwarding}$ is a list of forwarding rules of the form $\texttt{x.fwd(y);}$ with $\texttt{x}, \texttt{y} \in \mathit{Sw}$. 
Such a rule defines that switch~$\texttt{x}$ forwards packets to switch~$\texttt{y}$.
Each switch occurs at most once on the left hand side of the list of forwarding rules, making the forwarding decision of switches \emph{deterministic}.
If a switch is not configured, then all packets are dropped at that switch.
 
A \emph{concurrent update} to the $\texttt{forwarding}$ rules is given by the following syntax:
\[
\begin{array}{lll}
	\text{switch update} &\ ::= \ & \texttt{upd(x.fwd(z))} \hfill (\texttt{x}, \texttt{z} \in \mathit{Sw})\\
	\text{sequential update} &\ ::= \ & \texttt{(}\text{update } \texttt{>>} \text{ update } \texttt{>>} \text{ ... } \texttt{>>} \text{ update}\texttt{)} \\
	\text{parallel update} &\ ::= \ & \texttt{(} \text{update } \texttt{||} \text{ update } \texttt{||} \text{ ... } \texttt{||} \text{ update} \texttt{)} \\
	\text{update} &\ ::= \ & \text{switch update } | \text{ sequential update } | \text{ parallel update }\\
\end{array}
\]
where, for efficiency reasons, each switch is updated at most once.

Update statements of the form \texttt{upd(x.fwd(z))} define that a given configuration is altered by changing the forwarding rule $\texttt{x.fwd(y)}$ to $\texttt{x.fwd(z)}$ or by adding $\texttt{x.fwd(z)}$ to the given configuration if switch~$\texttt{x}$ is not configured in $\texttt{forwarding}$.
The first case updates the forwarding rule of switch~$\texttt{x}$ from switch~$\texttt{y}$ to switch~$\texttt{z}$ whereas the second case configures the switch~$\texttt{x}$ to now forward packets to switch~$\texttt{z}$ (instead of dropping the packets).
The sequential update defines a sequence of updates, where the next update is only carried out after the current update is finished.
The parallel update defines that all updates can happen in parallel and no assumptions about their order can be made.

Given a network topology $T = (\mathit{Sw}, \mathit{Con})$, an initial configuration $\texttt{ingress =}\\\mathit{Ingr} \texttt{; } \texttt{forwarding; egress = } \mathit{Egr}\texttt{;}$, and an $\mathit{update}$, we describe the construction of the corresponding Petri net with transits $\pNet=(\pl^{\mathit{D}} \cup \pl^C, \tr^{\mathit{D}} \cup \tr^C, \fl^{\mathit{D}} \cup \fl^C,\init^{\mathit{D}} \cup \init^C,\tfl^{\mathit{D}}) $ consisting of a sub-Petri net with transits\\ $\pNet^\mathit{D}= (\pl^{\mathit{D}}, \tr^{\mathit{D}}, \fl^{\mathit{D}},\init^{\mathit{D}},\tfl^{\mathit{D}}) $ to encode the topology and initial configuration, and a sub-Petri net without transits $\pNet^C = ( \pl^C, \tr^C, \fl^C, \init^C)$ encoding the update to the initial configuration.

\subsection{Data Plane as Petri Net with Transits}
\label{appendix:motivation-data-plane}

As the flow of packets is modeled by transits, we model switches and their connections with tokens that remain in corresponding places.
The initial configuration puts tokens in additional places such that only the transitions corresponding to the configured forwarding are enabled.
Specific transitions model ingress switches where new data flows begin.
These data flows are then extended by firing the enabled forwarding rules without moving any tokens.
We thereby model any order of newly generated packets and their forwarding.
We assume that for all connections between two switches~\texttt{x} and \texttt{y}, both $(\texttt{x},\texttt{y})$ and $(\texttt{y},\texttt{x})$ are in $\mathit{Con}$.

We create a place for each switch and for each direction of connections between two switches: \quad
$\pl^\mathit{D} =\{\mathit{sw}_\texttt{x} \mid \texttt{x} \in \mathit{Sw} \} \cup \{\texttt{x.fwd(y)} \mid (\texttt{x}, \texttt{y}) \in \mathit{Con}  \}$.

We create a transition for each direction of connections between switches and for each ingress switch: \quad
$\tr^\mathit{D} = \{ \mathit{fwd}_{\texttt{x}\rightarrow \texttt{y}} \mid (\texttt{x}, \texttt{y}) \in \mathit{Con} \} \cup \{\mathit{ingress}_\texttt{x} \mid \texttt{x} \in \mathit{Ingr}\}$. 
In \refSection{applications}, we call this set of transitions $\mathit{Fwd}$ when defining \emph{drop freedom}.

We define the flow relation such that each transition $\mathit{fwd}_{\texttt{x}\rightarrow\texttt{y}}$ has the places $\mathit{sw}_\texttt{x}$, $\mathit{sw}_\texttt{y}$, and $\texttt{x.fwd(y)}$ in its preset \emph{and} postset, and each transition $\mathit{ingress}_\texttt{x}$ has the place $\mathit{sw}_\texttt{x}$ in its preset \emph{and} postset: 
$\fl^\mathit{D} = \{ (\mathit{sw}_\texttt{x}, \mathit{fwd}_{\texttt{x}\rightarrow\texttt{y}}), (\mathit{sw}_\texttt{y}, \mathit{fwd}_{\texttt{x}\rightarrow\texttt{y}}),\\ (\texttt{x.fwd(y)}, \mathit{fwd}_{\texttt{x}\rightarrow \texttt{y}}), (\mathit{fwd}_{\texttt{x}\rightarrow \texttt{y}}, \mathit{sw}_\texttt{x}),
(\mathit{fwd}_{\texttt{x}\rightarrow\texttt{y}}, \mathit{sw}_\texttt{y}), (\mathit{fwd}_{\texttt{x}\rightarrow\texttt{y}}, \texttt{x.fwd(y)})
\mid (\texttt{x}, \texttt{y}) \in \mathit{Con} \} \cup
\{ (\mathit{sw}_\texttt{x}, \mathit{ingress}_\texttt{x}), (\mathit{ingress}_\texttt{x}, \mathit{sw}_\texttt{x}) \mid \texttt{x} \in \mathit{Ingr} \}$.
All transitions are weak fair.

The initial marking contains all switches and the initial forwarding rules from $\texttt{forwarding}$: \quad
$\init^\mathit{D} = \{ \mathit{sw}_x \mid x \in \mathit{Sw} \} \cup \texttt{forwarding}$.

The transit relation of transitions of the form $\mathit{fwd}_{\texttt{x}\rightarrow\texttt{y}}$ defines a data flow from switch~$\texttt{x}$ to switch~$\texttt{y}$ and maintains data flows in switch~$\texttt{y}$. Transitions of the form $\mathit{ingress}_\texttt{x}$ create a new data flow in switch~$\texttt{x}$ and maintain data flows in switch~\texttt{x}: \quad
$\tfl^\mathit{D} = \{ (\mathit{sw}_\texttt{x}\tfl(\mathit{fwd}_{\texttt{x}\rightarrow\texttt{y}})\mathit{sw}_\texttt{y}), (\mathit{sw}_\texttt{y}\tfl(\mathit{fwd}_{\texttt{x}\rightarrow\texttt{y}})\mathit{sw}_\texttt{y}) \mid (\texttt{x},\texttt{y}) \in \mathit{Con} \} \cup  \{ (\startfl \tfl(\mathit{ingress}_\texttt{x}) sw_\texttt{x}), ( sw_\texttt{x} \tfl(\mathit{ingress}_\texttt{x}) sw_\texttt{x}) \mid \texttt{x} \in \mathit{Ingr}\} $.

\subsection{Control Plane as Petri Net}
\label{appendix:motivation-control-plane}

The subnet for the update has no transit relation but moves tokens from and to places of the form $\texttt{x.fwd(y)}$ such that transitions corresponding to other forwarding rules become enabled.
The order of these updates is defined by the nesting of the sequential and the parallel operator in the update.

The update is realized by initially one token moving through dedicated places.
A parallel update temporarily increases the number of tokens and reduces it upon completion to one.
The update and each of its sub-expressions have a unique starting and finishing place. 
The following construction defines the update behavior between start and finish places and connects finish and start places to ensure sequentiality and concurrency depending on the sub-expression structure. 
For a given $\mathit{update}$, let $\mathit{SwU}$ be the set of switch updates in it, $\mathit{SeU}$ the set of sequential updates in it, and $\mathit{PaU}$ the set of parallel updates in it.
Depending on $\mathit{update}$'s type, it is added to the respective set.

We introduce start and finish places for all elements of the three sets $\mathit{SwU}$, $\mathit{SeU}$, and $\mathit{PaU}$:
$\pl^C = \{ \mathit{upd}^\mathit{start}, \mathit{upd}^\mathit{finish} \mid \mathit{upd} \in \mathit{SwU} \} \cup \{ \mathit{seq}^\mathit{start}, \mathit{seq}^\mathit{finish} \mid \mathit{seq} \in \mathit{SeU} \} \cup \{ \mathit{par}^\mathit{start}, \mathit{par}^\mathit{finish} \mid \mathit{par} \in \mathit{PaU} \}$.
All transitions are weak fair.

Each switch update is represented by one transition.
Each sequential update of length $n$ is represented by $n+1$ transitions, where the first and last transition start and finish the sequential update, and the $n - 1$ transitions in between connect the $n$ steps of the sequential update.
Each concurrent update is represented by an open and a close transition:
$\tr^C =\{ \mathit{upd} \mid \mathit{upd} \in \mathit{SwU}\} \cup \{ \mathit{seq}^0, \ldots, \mathit{seq}^n  \mid \mathit{seq} \in \mathit{SeU} \wedge |\mathit{seq}| = n \} \cup\{ \mathit{par}^\mathit{open}, \mathit{par}^\mathit{close} \mid \mathit{par} \in \mathit{PaU}\}$.

There are two mutually exclusive cases for the flow of switch updates where (1) a previous configuration is updated and (2) a new configuration is added:
For (1), we search for a previous configuration $\texttt{x.fwd(y)}$ and move the token to the new configuration $\texttt{x.fwd(z)}$: 
$\fl^C_{\mathit{SwU}1} =  \{ (\texttt{upd(x.fwd(z))}^\mathit{start}, \texttt{upd(x.fwd(z))}) ,  \\( \texttt{x.fwd(y)} , \texttt{upd(x.fwd(z))}), (\texttt{upd(x.fwd(z))}, \texttt{x.fwd(z)}), (\texttt{upd(x.fwd(z))},\\ \texttt{upd(x.fwd(z))}^\mathit{finish})  \mid \texttt{upd(x.fwd(z))} \in \mathit{SwU} \wedge \texttt{upd(x.fwd(y))} \in \texttt{forwarding}  \} $.

If no corresponding previous configuration exists (2), a new token is created at $\texttt{x.fwd(z)}$:
$\fl^C_{\mathit{SwU}2} =  \{ (\texttt{upd(x.fwd(z))}^\mathit{start}, \texttt{upd(x.fwd(z))}) , \\(\texttt{upd(x.fwd(z))}, \texttt{upd(x.fwd(z))}^\mathit{finish}), (\texttt{upd(x.fwd(z))}, \texttt{x.fwd(z)}) \mid \\ \texttt{upd(x.fwd(z))} \in \mathit{SwU} \wedge \nexists \texttt{y} \in \mathit{Sw} : \texttt{upd(x.fwd(y))} \in \texttt{forwarding} \} $.

For each sequential update, we add flows to the corresponding transitions such that the token can move from the start place of the sequential update to the start place of its first direct sub-formula, then from the finish of each direct sub-formula ($i$) to the start place of the next direct sub-formula ($i+1$), and in the end, from the finish place of the last direct sub-formula to the finish place of the sequential update:
$\fl^C_\mathit{SeU} = \{ (\mathit{seq}^\mathit{start} , \mathit{seq}^0), (\mathit{seq}^0, \mathit{sub}^\mathit{start}_1) \mid \mathit{seq} \in \mathit{SeU} \wedge \mathit{sub}_1 \text{ is the first direct sub-expression of } \mathit{seq}\} \cup \{(\mathit{sub}^\mathit{finish}_{i} , \mathit{seq}^i), (\mathit{seq}^i, \mathit{sub}^\mathit{start}_{i + 1}) \mid \mathit{seq} \in \mathit{SeU} \wedge |\mathit{seq}| = n \wedge  1 \leq i < n \wedge \mathit{sub}_{i}$ \text{ and } $\mathit{sub}_{i+1} \text{ are the $i$-th and $i+1$th}\\ \text{direct sub-formula of } \mathit{seq}\} \cup \{ (\mathit{sub}^\mathit{finish}_n , \mathit{seq}^{n}), (\mathit{seq}^{n}, \mathit{seq}^\mathit{finish}) \mid \mathit{seq} \in \mathit{SeU} \wedge\\ |\mathit{seq} | = n \wedge \mathit{sub}_n \text{ is the last direct sub-expression of } \mathit{seq}\} $.

Each concurrent update is opened to the starting places of its direct sub-formulas and closed from the finish places of its direct sub-formulas:\\
$\fl^C_\mathit{PaU} =\{ (\mathit{par}^\mathit{start}, \mathit{par}^\mathit{open} ), (\mathit{par}^\mathit{close}, \mathit{par}^\mathit{finish} )\}~\cup~\{(\mathit{par}^\mathit{open}, sub^\mathit{start}),\\ (\mathit{sub}^\mathit{finish} ,\mathit{par}^\mathit{close}) \mid \mathit{par} \in \mathit{PaU} \wedge \mathit{sub} \text{ is direct sub-expression of }\mathit{PaU} \}$.

We add all these sets into the flow of the Petri net $\fl^C = \fl^C_{\mathit{SwU}1} \cup \fl^C_{\mathit{SwU}2} \cup \fl^C_\mathit{SeU} \cup \fl^C_\mathit{PaU} $ and define the initial marking to contain one token in the starting place of the update: $\init^C = \{\mathit{update}^\mathit{start}\}$.

This construction gives us a Petri net with transits that models the concurrent update for the given network topology and initial configuration.
The set of egress switches can be used to formulate requirements.

\section{Proofs for the Transformation}
\label{appendixB}
In this section of the appendix, we provide details to \refSection{generalization}.
Firstly, we give a formal definition of the construction of the P/T Petri net with inhibitor arcs \(\pNetMC\)
from a Petri net with transits \(\pNet\) and a Flow-LTL formula \(\phiRun\) described in \refSection{pntfl2pn}.
Secondly, the formal transformation from a Flow-LTL formula \(\phiRun\) to an LTL formula \(\phiMC\)
(described in \refSection{flowltl2ltl}) is given. 
Thirdly, the correctness of these transformations are proven by mutually transforming the counterexamples.
Finally, the formal construction of the circuit (outlined in \refSection{pnMCwithCircuits}), its corresponding Kripke structure, and its correctness
are presented.
By this, the proofs for all lemmata and the theorem of this paper are made available.

\subsection{Formal Construction of the Net Transformation}
We introduce a set of identifiers \(\netIds\) and an injective naming function
\(\netId_\pNet:\pl\cup\tr\to\netIds\) for every
Petri net \(\petriNet\) (and all of its extensions)
which uniquely identifies every place and transition of a given net.
If the net \(\pNet\) is clear from the context, we omit the subscript and only write \(\netId\).
Furthermore, we often omit the predicate \(p\in\pl\wedge\netId(p)=\mathit{identifier}\)
and \(t\in\tr\wedge\netId(t)=\mathit{identifier}\), respectively, in formulas
and only use \(\mathit{identifier}\) instead of \(p\) or \(t\), respectively,
within the formula to keep the presentation short.

The construction of a Petri net with transits to a standard P/T Petri net with inhibitor arcs
is given by the following definition.
\begin{definition}[Petri Net with Transits to a P/T Petri Net]
\label{def:pntfl2pn}
Let \(\petriNetFl\) be a Petri net with transits and \(\phiRun\) a Flow-LTL formula
with \(n\in\N\) subformulas \(\A\,\phiLTL_i\), for \(i=1,\ldots,n\).
We define a P/T-Petri net with inhibitor arcs \(\pNetMC=(\plMC,\trMC,\flMC,\inhibitorFlMC,\initMC)\),
with
\begin{align*}
 \plMC\defEQ \plMC_o\cup\bigcup_{i\in\{1,\ldots,n\}}\plMC_i,\qquad \trMC\defEQ \trMC_o\cup\bigcup_{i\in\{1,\ldots,n\}}\trMC_i
\end{align*}
and a partial function
\(\lambda:~\trMC\cup\plMC\to \tr\cup\pl\)
which maps the elements to its corresponding original ones.
The smallest sets and sets \(\plMC_o,\plMC_i,\trMC_0\), \(\trMC_i\), \(\flMC\), and \(\inhibitorFlMC\) fulfilling
the following constraints define the net \(\pNetMC\).

By constraint \textit{(o)} it is ensured that all places, transitions, and flows of the original net \(\pNet\)
are also existent in \(\pNetMC\). The labeling and the identifiers are copied.
\begin{itemize}
\item[(o)] 	\(\plMC_o\supset\pl\wedge\trMC_o=\tr\wedge\flMC\supset\fl \wedge
	 	\forall p\MC\in\pl:\lambda(p\MC)=p\wedge\forall t\MC\in\tr:\lambda(t\MC)=t\wedge		 
	 	\forall p\MC\in\pl:\netId(p\MC)=\netId(p)\wedge\forall t\MC\in\tr:\netId(t\MC)=\netId(t)\)
\end{itemize}
With the following constraints starting with \textit{s}, the additional places, transitions,
and flows of the subnets are defined for each subformula \(\A\,\phiLTL_i\). Let \(\subnetIndices=\{1,\ldots,n\}\).
In \textit{(s1)}, a copy of every original place for each subnet is demanded for tracking the flow chains.
\begin{itemize}
\item[(s1)]	\(\forall i\in\subnetIndices:\forall p\in\pl:\exists p\MC\in\plMC_i: \lambda(p\MC)=p\wedge \netId(p\MC)=\subnet{\netId(p)}_i\)
\end{itemize}
To consider if any chain has been tracked, an initial place \(\initsub_i\) for every subnet is defined
via constraint \textit{(s2)}.
\begin{itemize}
\item[(s2)] \(\forall i\in\subnetIndices:\exists p\MC\in\plMC_i:\netId(p\MC)=\initsub_i\)
\end{itemize}
Constraint \textit{(s3)} ensures the existence of transitions 
simulating the creation of a flow chain during the run.
Hence, for every starting flow chain there is a transition in each subnet, which takes the initial token 
from \(\initsub_i\) and moves it according to the corresponding transit.
\begin{itemize}
\item[(s3)] \(\forall i\in\subnetIndices:\forall t\in\tr:\forall (\startfl,q)\in\tfl(t):
				\exists t\MC\in\trMC_i:(\initsub_i,t\MC)\in\flMC\wedge(t\MC,\subnet{\netId(q)}_i)\in\flMC\wedge\lambda(t\MC)=t
\wedge\netId(t\MC)=\netId(t)_{\netId(q)_i}\)
\end{itemize}
In constraint \textit{(s4)}, this is similarly done for each transit of each transition.
\begin{itemize}
\item[(s4)] \(\forall i\in\subnetIndices:\forall t\in\tr:\forall (p,q)\in\tfl(t):
				\exists t\MC\in\trMC_i:(\subnet{\netId(p)}_i,t\MC)\in\flMC\wedge(t\MC,\subnet{\netId(q)}_i)\in\fl\MC\wedge\lambda(t\MC)=t\wedge\netId(t\MC)=\netId(t)_{(\netId(p),\netId(q))_i}\)
\end{itemize}
Constraint \textit{(s5)} treats the situation that the currently tracked flow chain is independent of the 
firing of the previous original transition.
Thus, a skipping transition is demanded for every original transition \(t\) in each subnet
which is only allowed to fire when the considered chain of the subnet is not moved by the transits of \(t\), i.e.,
no corresponding places of the preset of \(t\) are occupied.
\begin{itemize}
\item[(s5)] \(\forall i\in\subnetIndices:\forall t\in\tr:
				\exists t\MC\in\trMC_i:\forall p\in\preset{t}: (\subnet{\netId(p)}_i,t\MC)\in\inhibitorFlMC\wedge\lambda(t\MC)=t\wedge\netId(t\MC)=\netId(t)_{\skipT_i}\)
\end{itemize}
The following four constraints are used to connect the components sequentially.
Constraint \textit{(a)} ensures one activation place \(\act{o}\) for the original net,
and one activation place \(\act{\netId(t)_i}\) for every original transition \(t\in\tr\) for each subnet.
\begin{itemize} 		
		\item[(a)]\(\exists p\MC\in\plMC_o: \netId(p\MC)=\act{o}\wedge
				\forall i\in\subnetIndices:\forall t\in\tr:\exists p\MC\in\plMC_i:\netId(p\MC)=\act{\netId(t)_i}\)
\end{itemize}
Constraint \textit{(mO)} let every original transition \(t\in\trMC_o\) take the activation token from \(\act{o}\) and moves
it to the activation place \(\act{\netId(t)_0}\) of all transition of the first subnet which are labelled with \(t\) 
to activate this subnet.
\begin{itemize}			
		\item[(mO)] \(\forall t\in\tr_o:(\act{o},t)\in\flMC\wedge(t,\act{\netId(t)_0})\in\flMC\)
\end{itemize}
With the constraints \textit{(mSi)} and \textit{(mSn)}, we move the active token through the subnets, or back to 
the original net, respectively.
Therefore, we let all equally labelled transitions of the subnet take their corresponding activation token from the place
\(\act{\netId(t)_i}\) for a label \(t\in\tr\) and move it to the next subnet, i.e., place \(\act{\netId(t)_{i+1}}\) or
\(\act{o}\), respectively.
\begin{itemize} 											
		\item[(mSi)]\(\forall i\in\{1,\ldots,n-1\}:\forall t\in\tr:\\\forall t\MC\in\trMC_i:
				\lambda(t\MC)=t\implies ((\act{\netId(t)_i},t\MC)\in\flMC\wedge(t\MC,\act{\netId(t)_{i+1}})\in\fl\MC)\)
											
		\item[(mSn)]\(\forall t\in\tr:\forall t\MC\in\trMC_n:
				\lambda(t\MC)=t\implies ((\act{\netId(t)_n},t\MC)\in\flMC\wedge(t\MC,\act{o})\in\fl\MC)\)
\end{itemize}
The initial marking of \(\pNetMC\) is defined by constraint \textit{(in)}.
We only activate the original part of the net and allow all subnets to track a chain.
\begin{itemize}
\item[(in)] \(\initMC=\{\act{o}\}\cup\{\initsub_i\with i\in\subnetIndices\}\cup\init\).
\end{itemize}
Newly introduced identifiers, e.g., \(\iota\) and \(\act{o}\), are unique and do not occur in \(\pNet\).
\end{definition}

The results of \refLemma{sizePN} directly follows from this construction. 
Having a copy of each original place for each flow subformula and
an activation place for each original transition
yields the quadratic number of places in the size of the net and the number of subformulas.
The cubic number of the transitions in the size of the net and the number of subformulas
follows from the possible quadratic number of transits of each transition
and that for each flow subformula and for each transit a new transition is added.
\begin{proof}[\refLemma{sizePN} (Size of the Constructed Net)]
The injectivity of \(\netId\) yields that each demanded place or transition is a unique element of \(\plMC\) or \(\trMC\) 
respectively. That we demand the smallest sets \(\plMC\) and \(\trMC\) fulfilling the constraints allows to only
consider the explicitly stated elements.

Therewith, constraint \textit{(o)} together with constraint \textit{(a)} yield \(|\plMC_o|=|\pl|+1\).
Constraint \textit{(s1)} demands \(|\pl|\) places for each subnet, constraint \textit{(s2)} demands one
initial place for every subnet, and constraint \textit{(a)} requires one active place for each original 
transition for every subnet.
Hence, \(|\bigcup_{i=1,\ldots,n}\plMC_i|=n\cdot|\pl|+n+n\cdot|\tr|\) and so \(\plMC=n\cdot(|\pl|+|\tr|+2)+|\pl|+1\).

For the transitions of the original part of the net, constraint \textit{(o)} directly yields \(|\trMC_o|=|\tr|\).
Constraint \textit{(s3)} demands a transition for every newly created flow chain of the net, i.e.,
\(|\{(t,p)\in\tr\times\pl\with (\startfl,p)\in\tfl(t)\}|\) many transitions, which 
are at most \(|\tr|\cdot|\pl|\) transitions for each subnet.
Constraint \textit{(s4)} does the same for every transit of the net, i.e.,
\(|\{(p,t,q)\in\pl\times\tr\times\pl\with (p,q)\in\tfl(t)\}|\) many transitions, which 
are at most \(|\pl|\cdot|\tr|\cdot|\pl|\) transitions for each subnet.
Constraint \textit{(s5)} requires one transition for skipping and directly moving the active token 
to the next subnet.
Hence, \(|\bigcup_{i=1,\ldots,n}\trMC_i|=n\cdot|\tr|\cdot|\pl|+n\cdot|\pl|\cdot|\tr|\cdot|\pl|+n\cdot|\tr|\)
and thus, \(\trMC\) contains \(n\cdot(|\pl|^2\cdot|\tr|+|\pl|\cdot|\tr|+|\tr|) + |\tr|\) transitions.\hfill\qed
\end{proof}

\subsection{Formal Construction of the Formula Transformation}
We formally introduce the transformation of a Flow-LTL formula \(\phiRun\) to an LTL formula \(\phiMC\).
We use 
the abbreviation \(\phi_1\LTLweakuntil\phi_2\) for \(\ltlAlways \phi_1 \vee (\phi_1\ltlUntil \phi_2)\)
and
the operator \([\phi_1'/\phi_1,\ldots, \phi_m'/\phi_m]\) on formulas
for the simultaneous substitution of \(\phi_j\) by \(\phi_j'\).
To substitute formulas from the inner- to the outermost,
we utilise the function \(d\), which calculates the depth of a formula.
The \emph{depth-function} \(d\) is inductively defined: \(d(a)=0\) for every atomic proposition~\(a\)
and \(d(\circ\, \phi_1)=1+d(\phi_1)\), \(d(\phi_1\widetilde{\circ}\,\phi_2)=1+\textit{max}\{d(\phi_1),d(\phi_1)\}\)
for every unary operator \(\circ\), binary operator \(\widetilde{\circ}\), and formulas \(\phi_1\) and \(\phi_2\).

\begin{definition}[Flow-LTL to LTL]
\label{def:flowLTL2ltl}
Let \(\petriNetFl\) be a Petri net with transits,
\(\phiRun\) a Flow-LTL formula with \(n\in\N\) subformulas \(\A\,\phiLTL_i\), for \(i=1,\ldots,n\),
and \(\pNetMC=(\plMC,\trMC,\flMC,\inhibitorFlMC,\initMC)\) the with \refDef{pntfl2pn} created P/T Petri net with inhibitor arcs.
The corresponding LTL formula \(\phiMC\) is created by the following steps.

\noindent{}\textbf{Flow part:}
For each flow formula \(\A\,\phiLTL_i\), for \(i=1,\ldots,n\), we 
create a new formula \(\phiLTL_i^{{\pl\tr\mathcal{X}}_{d_\mathtt{max}}}\) which adequately
copes with the different timelines of the corresponding flow chains.
Since in our approach each flow formula is checked on the corresponding subnet \(\pNetMC_i\),
the places and transitions of the other components are ignored.

The atomic propositions \(p\in\pl\) are substituted with the 
corresponding places of the subnet:
\begin{itemize}
\item[(pF)] \(\phiLTL_i^\pl\defEQ\phiLTL_i\left[\begin{array}{c}\subnet{\netId(p_1)}_i/p_1,
							\\\ldots,
							\\\subnet{\netId(p_m)}_i/p_m\end{array}
					\right]\) for \(\pl=\{p_1,\ldots,p_m\}\).
\end{itemize}
For the atomic propositions \(t\in\tr\), we skip all transitions not concerning the extension 
of the current flow chain via an until operator.
That means the firing of \emph{unrelated} transitions
\(O_i=\left(\trMC\setminus\trMC_i\right)\cup\{\netId(t)_{\skipT_i}\in\trMC_i\with t\in\tr\}\),
i.e., transitions of the other components or own skipping transitions, are skipped
until a \emph{related}, i.e., one of the set of transitions extending the flow chain
\(M_i(t)=\{t\MC\in\trMC_i\setminus\{\netId(t)_{\skipT_i}\}\with \lambda(t\MC)=t\}\),
is fired.
\begin{itemize}
\item[(tF)] \(\phiLTL_i^{\pl\tr}\defEQ\phiLTL_i^\pl\left[\begin{array}{c}(\bigvee_{t_o\in O_i} t_o)\U (\bigvee_{t\in M_i(t_1)}t)/t_1,\\
								\ldots,\\
								(\bigvee_{t_o\in O_i} t_o)\U (\bigvee_{t\in M_i(t_{m'})}t)/t_{m'}
						\end{array}\right]\)
for \(\tr=\{t_1,\ldots,t_{m'}\}\).
\end{itemize}
The next operator is treated similarly.
Let \(\phiLTL_i^{\pl\tr}\) contain \(m_1\) subformulas \(\ltlNext \phiLTL_j'\) for \(j=1,\ldots,m_1\).
In this case, the \emph{related} transitions are all transitions of the subnet \(\trMC_i\) but the skipping transitions:
\(M_i=\trMC_i\setminus\{\netId(t)_{\skipT_i}\in\trMC_i\with t\in\tr\}\).
We define the disjunction of all related transitions by \(\widetilde{M}_i=\bigvee_{t\in M_i}t\) and apply the same ideas
as in the previous case.
To adequately cope with situations where no related transition \(t\in M_i\) would ever fire
again (\(\Box(\neg\widetilde{M}_i)\)), i.e., the stuttering at the end of a chain,
we require the immediate satisfaction.
To replace the formulas from the inner- to the outermost, we organize 
the formulas in groups \(\{\ltlNext\phiLTL_1^l,\ldots,\ltlNext\phiLTL_{k_l}^l\}\) according to their depth:
\begin{itemize}
\item[(nF)] Let \(\{\ltlNext\phiLTL_1^l,\ldots,\ltlNext\phiLTL_{k_l}^l\}=\{\ltlNext\phiLTL_j'\with j\in\{1,\ldots,m_1\}\wedge \depth(\ltlNext\phiLTL_j')=l\}\)
and \(d\in\{2,\ldots,d_\mathtt{max}\}\) with \(d_\mathtt{max}\defEQ\mathtt{max}\{\depth(\ltlNext \phiLTL_j')\with j\in\{1,\ldots,m_1\}\}\).
Then, the substitutions of the next operators is given by
\begin{align*}
\phiLTL_i^{{\pl\tr\mathcal{X}}_1}&\defEQ\phiLTL_i^{\pl\tr}\left[\begin{array}{c}
	((\bigvee_{t\in O_i} t)\U (\widetilde{M}_i\wedge\ltlNext\phiLTL_1^1))\vee (\Box(\neg\widetilde{M}_i)\wedge\phiLTL_1^1)/\ltlNext\phiLTL_1^1,\\
			\ldots,\\
	((\bigvee_{t\in O_i} t)\U (\widetilde{M}_i\wedge\ltlNext\phiLTL_{k_1}^1))\vee (\Box(\neg\widetilde{M}_i)\wedge\phiLTL_{k_1}^1)/\ltlNext\phiLTL_{k_1}^1
\end{array}\right] 
\end{align*}
and 
\begin{align*}
\phiLTL_i^{{\pl\tr\mathcal{X}}_d}&\defEQ\phiLTL_i^{{\pl\tr\mathcal{X}}_{d-1}}\left[\begin{array}{c}
	((\bigvee_{t\in O_i} t)\U (\widetilde{M}_i\wedge\ltlNext\phiLTL_1^d))\vee (\Box(\neg\widetilde{M}_i)\wedge\phiLTL_1^d)/\ltlNext\phiLTL_1^d,\\
		\ldots,\\
	((\bigvee_{t\in O_i} t)\U (\widetilde{M}_i\wedge\ltlNext\phiLTL_{k_d}^d))\vee (\Box(\neg\widetilde{M}_i)\wedge\phiLTL_{k_d}^d)/\ltlNext\phiLTL_{k_d}^d
\end{array}\right]
\end{align*}
\end{itemize}

\noindent{}\textbf{Run part:}
They same ideas are applied for the run part of the formula.
Here the atomic propositions \(p\in\pl\) do not need to be substituted since \(\plMC_o=\pl\) holds.
For the atomic propositions \(t\in\tr\) however, the skipping procedure has to be applied as well.
The \emph{unrelated} transitions in this case are all transitions of the subnet \(O=\trMC\setminus\tr\).
To only substitute occurrences in the run part of the formula,
we introduce the substitution operator \(/_{\bar{\A}}\) 
which does not change anything within the scope of the flow operator \(\A\).
\begin{itemize}
\item[(tR)] \(\phiRun^\tr\defEQ\phiRun\left[\begin{array}{c}(\bigvee_{t\in O} t)\U t_1/_{\bar{\A}}t_1,\\
							\ldots,\\
								(\bigvee_{t\in O} t)\U t_m/_{\bar{\A}}t_m
						\end{array}\right]\).
\end{itemize}
For the next operator in the run part of the formula,
the \emph{related} transition are all transitions of \(\pNetMC_O\), i.e., \(\tr\).
We define the disjunction of these transitions by \(T=\bigvee_{t\in \tr}t\).
To adequately cope with situations where no transition \(t\in\tr\) would ever fire
again (\(\Box(\neg T)\)), i.e., the stuttering in the traces for finite firing sequences,
we require the immediate satisfaction.
Let the run part of \(\phiRun^{\tr}\) contain \(m_2\) subformulas \(\ltlNext \phiRun_j\) for \(j=1,\ldots,m_2\).
We again organize the formulas according to their depth, to replace them from the inner- to the outermost.
\begin{itemize}
\item[(nR)] Let \(\{\ltlNext{\phiRun}_1^{l},\ldots,\ltlNext\phiRun_{k_l}^{l}\}=\{\ltlNext\phiRun_j\with j\in\{1,\ldots,m_2\}\wedge \depth(\ltlNext\phiRun_j)=l\}\)
and \(d'\in\{2,\ldots,d_\mathtt{max}'\}\) with \(d_\mathtt{max}'\defEQ\mathtt{max}\{\depth(\ltlNext \phiRun_j)\with j\in\{1,\ldots,m_2\}\}\).
Then, the substitutions of the next operators is given by
\begin{align*}
\phiRun^{{\tr\mathcal{X}}_1}&\defEQ\phiRun^{\tr}\left[\begin{array}{c}
	((\bigvee_{t\in O} t)\U (T\wedge\ltlNext\phiRun_1^1))\vee (\Box(\neg T)\wedge\phiRun_1^1)/_{\bar{\A}}\ltlNext\phiRun_1^1,\\
			\ldots,\\
	((\bigvee_{t\in O} t)\U (T\wedge\ltlNext\phiRun{k_1}^1))\vee (\Box(\neg T)\wedge\phiRun_{k_1}^1)/_{\bar{\A}}\ltlNext\phiRun_{k_1}^1
\end{array}\right] 
\end{align*}
and 
\begin{align*}
\phiRun^{{\tr\mathcal{X}}_{d'}}&\defEQ\phiRun^{{\tr\mathcal{X}}_{d'-1}}\left[\begin{array}{c}
	((\bigvee_{t\in O} t)\U (T\wedge\ltlNext\phiRun_1^{d'}))\vee (\Box(\neg T)\wedge\phiRun_1^{d'})/_{\bar{\A}}\ltlNext\phiRun_1^{d'},\\
		\ldots,\\
	((\bigvee_{t\in O} t)\U (T\wedge\ltlNext\phiRun_{k_{d'}}^{d'}))\vee (\Box(\neg T)\wedge\phiRun_{k_{d'}}^{d'})/_{\bar{\A}}\ltlNext\phiRun_{k_{d'}}^{d'}
\end{array}\right] 
\end{align*}
\end{itemize}
Since flow chains can be created at any time during the run, we skip to their creation point, i.e., 
to the time \(\initsub_i\) is not occupied anymore.
To also allow for not tracking any chain the weak until operator is used.
\begin{itemize}
\item[(\(\A\)R)] \(\phiRun^\A\defEQ\phiRun^{\tr\mathcal{X}_{d_\mathtt{max}'}}\left[\begin{array}{c}\initsub_1\LTLweakuntil(\neg \initsub_1\wedge\phiLTL_1^{{\pl\tr\mathcal{X}}_{d_\mathtt{max}}})/\A\,\phiLTL_1,\\\ldots,\\\initsub_n\LTLweakuntil(\neg \initsub_n\wedge\phiLTL_n^{{\pl\tr\mathcal{X}}_{d_\mathtt{max}}})/\A\,\phiLTL_n\end{array}\right]\)
\end{itemize}
Since we do not want any run or firing sequence to stop in any subnet \(\pNetMC_i\),
the final formula restricts the considered traces to those which infinitely often visit
the activation place \(\act{o}\) of the original part of the net.
\begin{itemize}
\item[(nSub)] \(\phiMC\defEQ \Box\Diamond\act{o}\rightarrow \phiRun^\A\)
\end{itemize}
\end{definition}
%
The idea of the proof of \refLemma{sizeFormula} 
directly results from \refLemma{sizePN}, since the substitution of elements of \(\phiRun\)
depends on the size of \(\pNetMC\).
This net is quartic in the size of \(\pNet\) and the number of flow subformulas in \(\phiRun\).
\begin{proof}[\refLemma{sizeFormula} (Size of the Constructed Formula)]
For each occurrence of an atomic proposition \(t\in\tr\) in the \emph{run part} of \(\phiRun\),
we have additionally
\(1 + 2\cdot|\trMC\setminus\tr|-1\)
subformulas in \(\phiMC\) by constraint \textit{(tR)}.
Since \(\trMC\) is quartic in the size of \(\pNet\) and \(n\), \(\phiMC\) is already 
quintic in the size of \(\pNet\) and \(\phiRun\).
For each occurrence of an atomic proposition \(t\in\tr\) in a \emph{flow subformula} of \(\phiRun\)
we have additionally 
\(1+2\cdot(|\trMC\setminus\tr_i|+|\tr|)-1 + 2\cdot(|M_i(t)|)-1\)
subformulas in \(\phiMC\) (constraint \textit{(tF)}). Since \(M_i(t)\) are at most \(|\pl|^2+|\pl|\) transits of a transition
\(\trMC\) is still the biggest part and we are quintic in the size of \(\pNet\) and \(\phiRun\).
For each occurrence of a next operator in the \emph{run part} of \(\phiRun\)
we have additionally
\(2 + 1 + 2\cdot|\trMC\setminus\tr|-1 + 3 + |\tr| + 1 + 3 + 1 + |\tr|\)
subformulas in \(\phiMC\) (constraint \textit{(nR)}). Hence, \(\phiMC\) is still quintic in the size of \(\pNet\) and \(\phiRun\).
For each occurrence of a next operator in a \emph{flow subformula} of \(\phiRun\) we have
additionally
\(2 + 1 + 2\cdot(|\trMC\setminus\tr_i|+|\tr|)-1 + 3 + |\tr_i|-|\tr| + 1 + 3 + 1 + |\tr_i|-|\tr|\)
subformulas in \(\phiMC\) (constraint \textit{(nF)}). Hence, \(\phiMC\) is still quintic in the size of \(\pNet\) and \(\phiRun\).
\hfill\qed
\end{proof}

\subsection{Correctness Proof of the Transformations}
In this section, we prove \refLemma{correctnessPNLTL}.
Hence, we fix a Petri net with transits \(\petriNet\),
the corresponding Petri net with inhibitor arcs \(\pNetMC=(\plMC,\trMC,\flMC,\inhibitorFlMC,\initMC)\)
which is created by \refDef{pntfl2pn},
a Flow-LTL formula \(\phiRun\) with \(n\in\N\) subformulas \(\A\,\phiLTL_i\), for \(i=1,\ldots,n\), and the corresponding LTL formula \(\phiMC\)
which is created by \refDef{flowLTL2ltl}.

The general idea is to show the contraposition of the statement:
\[\pNet \not\models \phiRun\text{ iff } \pNetMC\not\models_\mathtt{LTL}\phiMC.\]
Therefore, we transform the counterexamples mutually by \refDef{cexPNwt2PN} and \refDef{cexPN2PNwT}.
For \refDef{cexPNwt2PN}, we sequentially pump up the firing sequence serving as counterexample for \(\pNet \models\phiRun\)
by one transition for each subnet in every step.
If the subnet has to consider a flow chain, i.e., \(\runPN,\sigma(\firingSequence)\not\models\A\,\psi_i\)
and the original transition transits the flow chain, the corresponding transition of the subnet is used.
Otherwise, the corresponding skipping transition is added.
The markings are extended to the additional tokens of \(\pNetMC\).

\begin{definition}[CEX: From PNwT to PN]
\label{def:cexPNwt2PN}
Let \(\runPN=(\pNet^R, \rho)\) be a run of \(\pNet\),
\(\firingSequence=M_0\firable{t_0}M_1\firable{t_1}\cdots\) a covering firing sequence, and
for every \(\runPN,\sigma(\firingSequence)\not\models\A\psi_i\)
be \(\flowChain^i=p_0^i,t_0^i,p_1^i, t_1^i,\ldots\)
the corresponding flow chain with \(\sigma(\flowChain^i)\not\models_\mathtt{LTL} \psi_i\).
We create a tuple \({\runPN}\MC=({\pNet\MC}^R, \rho\MC)\) and a
sequence \(\firingSequence\MC=M_0\MC\firable{t_0\MC}M_1\MC\firable{t_1\MC}\cdots\) iteratively.
We lift the function \(\rho\MC\) to sets \(X\) by \(\rho\MC(X)=\{\rho\MC(x)\with x\in X\}\).

\begin{itemize}
\item[(i)] The marking \(M_0\MC\) corresponds to the initial marking of \(\pNetMC\), i.e., \(\rho\MC(M_0\MC)=\initMC\).
\item[(ii)] Every \((n+1)\)th transition is the next transition of \(\firingSequence\).
Thus, \(t_{j\cdot(n+1)}\MC=t_j\) with \(\rho\MC(t_{j\cdot(n+1)})=\rho(t_j)\) for every \(j\in\N\)
(as long as \(t_j\) is existent in \(\firingSequence\)).
\item[(iii)] Every other transition \(t=t_{j\cdot(n+1)+i}\MC\), for the existing \(t_j\) 
and \(i\in\{1,\ldots,n\}\), is a fresh transition .
The mapping of \(t\) is dependent on the previous original transition \(t_o=\rho(t_{j\cdot(n+1)}\MC)\in\trMC_o=\tr\).
In the case of \(\runPN,\sigma(\firingSequence)\models\A\,\psi_i\), no chain has to be considered, thus
\(t\) is mapped to the corresponding skipping transition (i.e., \(\rho\MC(t)=\netId({t_o})_{\skipT_i}\)).
In the case of \(\runPN,\sigma(\firingSequence)\not\models\A\,\psi_i\) the mapping is done iteratively
according to \(\flowChain^i\).
The first occurrence of \((\startfl,\rho(p_0^i))\in\tfl(t_o)\) yields \(\rho\MC(t)=\netId(t_o)_{\netId(\rho(p_0^i))_i}\).
We remember this position by \(\Theta_i(\firingSequence,0)=j(n+1)+i+1\).
Then, whenever the next transition \(t_o\) with \((\rho(p_{k}^i),\rho(p_{k+1}^i))\in\tfl(t_o)\) occurs,
\(\rho\MC(t)=\netId(t_o)_{(\netId(p_{k}^i),\netId(p_{k+1}^i))_i}\) is used.
This position is remembered with \(\Theta_i(\firingSequence,k)=j(n+1)+i+1\).
In all other cases, \(\rho\MC(t)=\netId({t_o})_{\skipT_i}\) holds again.
\item[(iv)] Each marking \(M_k\MC\) of \(\firingSequence\MC\) for \(k\in\N\setminus\{0\}\) is
corresponding to a marking \({M'_k}\MC=\rho(M_{k-1}\MC)\setminus\pre{{\pNetMC}}{\rho(t_{k-1}\MC)}\cup\post{{\pNetMC}}{\rho(t_{k-1}\MC)}\)
(as long as \(t_{k-1}\MC\) has been created), i.e., \(\rho\MC(M_k\MC)={M'_k}\MC\).
\end{itemize}
The net \({\pNetMC}^R\) is created iteratively out of the places of the markings \(M_j\MC\) and the transitions \(t_j\MC\)
of \(\firingSequence\MC\) which are connected according to the pre- and postsets of \(\rho\MC(t_j\MC)\).
We denote this construction with \(\Theta(\firingSequence)=\firingSequence\MC\).
\end{definition}
Note that the constructed tuple \({\runPN}\MC=({\pNetMC}^R,\rho\MC)\) 
is indeed a run of \(\pNetMC\) 
and the constructed sequence \(\firingSequence\MC\) is a firing sequence covering \({\runPN}\MC\).

The firing sequence serving as counterexample for \(\pNet\models\phiRun\) is gained from 
a covering \(\firingSequence\MC\) of a run of \(\pNetMC\) by projecting onto the elements of \(\pNet\).
The flow chains serving as counterexample for \(\runPN,\sigma(\firingSequence)\models\A\,\psi_i\) are created 
by iteratively concatenating the corresponding places and the transitions different to the skipping transition
of each subnet.
\begin{definition}[CEX: From PN to PNwT]
\label{def:cexPN2PNwT}
Let \({\runPN}\MC=({\pNet\MC}^R, \rho\MC)\) be a run of \(\pNetMC\)
and \(\firingSequence\MC=M_0\MC\firable{t_0\MC}M_1\MC\firable{t_1\MC}\cdots\) a covering firing sequence.
\begin{enumerate}
\item[(i)] We create a sequence \(\firingSequence=M_0\firable{t_0}M_1\firable{t_1}\cdots\) by projecting onto the elements of \(\pNet\),
i.e.,
\(M_j=\{p\in M_{j\cdot(n+1)}\MC\with\rho\MC(p)\in\pl\}\)
and
\(t_j=t_{j\cdot(n+1)}\MC\)
for all \(j\in\N\) (as long as \(M_{j\cdot(n+1)}\MC\) and \(t_{j\cdot(n+1)}\MC\) are existent in \(\firingSequence\MC\)).
\item[(ii)] The net \({\pNet}^R\) is analogously created as in \refDef{cexPNwt2PN} and \(\rho\) is defined by
\(\rho(p)=\rho\MC(p)\) for all \(p\in M_j\) and \(\rho(t_j)=\rho\MC(t_{j\cdot(n+1)}\MC)\) for all \(j\in\N\) (as long as the elements exist).
\item[(iii)] The flow chain \(\flowChain^i=p_0^i,t_0^i,p_1^i, t_1^i,\ldots\) for a subnet \(i\in\{1,\ldots,n\}\)
is only created when there is any transition \(t_j\MC\) in \(\firingSequence\MC\) with \(\rho\MC(t_j\MC)=t\MC\) 
for any \(t\MC\in\trMC_i\) with \(\netId(t\MC)\neq\netId(\lambda(t\MC))_{\skipT_i}\).
In the case that there are such \(t_j\MC\),
we iteratively collect the corresponding transitions and their to the transit belonging corresponding places of the pre- and postset.
This means, 
if \(\netId(t\MC)=\netId(\lambda(t\MC))_{\netId(q)_i}\) for any \(q\in\pl\), add \(q\), and if
\(\netId(t\MC)=\netId(\lambda(t\MC))_{(\netId(p),\netId(q))_i}\) for any \(p,q\in\pl\),
then add \(\lambda(t\MC),q\) to the sequence.
For each adding step \(k\), we remember the position in \(\firingSequence\MC\)
by \(\Theta_i\MC(\firingSequence\MC,k)=j+1\).
 We denote the construction by \(\Theta_{\flowChain^i}\MC(\firingSequence\MC)=\flowChain^i\).
\end{enumerate}
For this construction we denote \(\Theta_R\MC(\firingSequence\MC)=(\pNet^R,\rho)\)
and \(\Theta\MC(\firingSequence\MC)=\firingSequence\).
\end{definition}
Note that the constructed tuple \(\runPN=({\pNetMC},\rho)\)
indeed is a run of \(\pNet\),
the constructed sequences \(\flowChain^j\)
are flow chains of \(\pNet\), and
the constructed sequence \(\firingSequence\)
is a covering firing sequence of \(\runPN\).

We prove \refLemma{correctnessPNLTL} via a nested structural induction over the Flow-LTL formula \(\phiRun\).
Therefore, we use the LTL part and the flow part of \(\phiRun\) as induction base for the outer induction,
and prove each part separately by structural induction.

\begin{lemma}[LTL Part]
\label{lem:LTLpart}
Given an LTL formula \(\phiLTL\MC\) created by \refDef{flowLTL2ltl} without condition \textit{(nSub)}
from the LTL part \(\phiLTL\) (not within the scope of a flow operator \(\A\)) of a  flow-LTL formula \(\phiRun\).
\begin{itemize}
\item[(s)]
Given a run \(\runPN\) and a covering firing sequence \(\firingSequence\), with
 \(\sigma(\Theta(\firingSequence))^i\models_\mathit{LTL}\ltlAlways\ltlEventually\act{o}\) for all \(i\in\N\), then
\(\sigma(\firingSequence)\not\models_\mathit{LTL}\phiLTL\implies \sigma(\Theta(\firingSequence))\not\models_\mathtt{LTL} \phiLTL\MC\)
holds.
\item[(c)]
Given a firing sequence \(\firingSequence\MC\) of a run of the net \(\pNetMC\) with
 \(\sigma(\firingSequence\MC)^i\models_\mathit{LTL}\ltlAlways\ltlEventually\act{o}\) for all \(i\in\N\), then
\(\sigma(\firingSequence\MC)\not\models_\mathit{LTL}\phiLTL\MC\implies \sigma(\Theta\MC(\firingSequence\MC))\not\models_\mathtt{LTL} \phiLTL\)
holds.
\end{itemize}
\end{lemma}
\begin{proof}[via structural induction over \(\phiLTL\)]
For proving \textit{(s)} and \textit{(c)}, we show 
\emph{soundness}, i.e.,
\[\forall i\in\N: \sigma(\firingSequence)^i\not\models_\mathtt{LTL}\phiLTL \implies \forall j\in\{0,\ldots,n_0\}:\sigma(\Theta(\firingSequence))^{i(n+1)-j}\not\models_\mathtt{LTL} \phiLTL\MC,\]
and \emph{completeness}, i.e.,
\[\forall i\in\N: \exists j\in\{0,\ldots,n_0\}:\sigma(\firingSequence\MC)^{i(n+1)-j}\not\models_\mathtt{LTL} \phiLTL\MC \implies \sigma(\Theta\MC(\firingSequence\MC))^i\not\models_\mathtt{LTL}\phiLTL\]
with \(n_0 = \left\{\begin{array}{ll}0 &\text{if }i=0\\ n&\text{otherwise}\end{array}\right.\)
hold by an induction over the structure of \(\phiLTL\)
.\\
(IB)
\textbf{Case} \(\phiLTL=p\in\pl\). Definition~\ref{def:flowLTL2ltl} yields \(\phiLTL\MC=p\).\\
\noindent\textit{Regarding {soundness}}:
Let \(\firingSequence=M_0\firable{t_0}M_1\firable{t_1}\cdots\) and \(\Theta(\firingSequence)=M_0\MC\firable{t_0\MC}M_1\MC\firable{t_1\MC}\cdots\)
with the corresponding mapping functions \(\rho\) and \(\rho\MC\), respectively.
The premise of the statement yields \(p\nin\sigma(\firingSequence)(i)\).
For \(i=0\), condition \textit(in) of \refDef{pntfl2pn} together with condition \textit{(i)} of \refDef{cexPNwt2PN}
ensures that \(\rho(M_0)\subset\rho\MC(M_0\MC)\) and that there cannot be any other \(p'\in M_0\MC\) with 
\(\rho\MC(p')\in\pl\). Hence, \(p\nin\sigma(\Theta(\firingSequence))(0)\).
For \(i>0\), condition \textit{(iv)} of \refDef{cexPNwt2PN} yields that all other markings \(M_k\MC\) are mapped to 
markings which are created by the firing of transitions of \(\pNetMC\).
Definition~\ref{def:pntfl2pn} ensures that tokens residing on places \(p\in \pl_o\MC\cap\pl\) of the original part of the net 
are not moved by transitions of the subnet. Hence, by \(p\nin\sigma(\firingSequence)(i)\) we know that
\(p\) cannot get occupied while firing any transition of the subnet.
With condition \textit{(ii)} and \textit{(iii)} of \refDef{cexPNwt2PN} we know \(p\nin\sigma(\Theta(\firingSequence))(i(n+1)-j)\)
for all \(j\in\{0,\ldots,n\}\).\\[2mm]
\noindent\textit{Regarding {completeness}}:
Given \(\firingSequence\MC=M_0\MC\firable{t_0\MC}M_1\MC\firable{t_1\MC}\cdots\) and its transformation
\(\Theta\MC(\firingSequence\MC)=M_0\firable{t_0}M_1\firable{t_1}\cdots\)
with the corresponding mapping functions \(\rho\) and \(\rho\MC\), respectively.
For \(i=0\), the premise yields \(p\nin\sigma(\firingSequence\MC)(0)\).
From condition \textit{(i)} of \refDef{cexPN2PNwT}, we know that 
\(M_0=\{p\MC\in M_0\MC \with \rho\MC(p\MC)\in\pl\}\) and since \(p\in\pl\) and
condition \textit{(ii)} of the definition ensures that both mapping functions coincide,
\(p\nin\sigma(\Theta\MC(\firingSequence\MC)\).
For all other \(i\in\N\), the premise yields that there is a \(j\in\{0,\ldots,n_0\}\)
such that \(p\nin\sigma(\firingSequence\MC)(i(n+1)-j)\).
Definition~\ref{def:pntfl2pn} ensures that every firing sequence of \(\pNetMC\) repeatedly
has an original transition \(t_o\in\trMC_o\) and then sequentially \(n\) transitions, one for each subnet.
Furthermore, no transition of any subnet moves a token of the original net (apart from \(\act{o}\)).
Thus, for all \(j'\in\{0,\ldots,n_0\}\) the places \(p\in\pl\) of the markings of \(\sigma(\firingSequence\MC)(i(n+1)-j')\)
stay the same. Since \refDef{cexPN2PNwT} creates the markings for \(\Theta\MC(\firingSequence\MC)\)
exactly by choosing those corresponding places and maps them accordingly, \(p\nin\sigma(\Theta\MC(\firingSequence\MC)(i(n+1))\).
\\
\textbf{Case} \(\phiLTL=t\in\tr\). Definition~\ref{def:flowLTL2ltl} yields \(\phiLTL\MC=\bigvee_{t'\in\tr\MC\setminus\tr}t'\ltlUntil t\).\\
\noindent\textit{Regarding {soundness}}:
For all \(i\in\N\) the premise yields \(t\nin\sigma(\firingSequence)(i)\).
Since condition \textit{(ii)} of \refDef{cexPNwt2PN} copies every \((n+1)\)th transition
we firstly know that \(t\nin\sigma(\Theta(\firingSequence))(i(n+1))\) and secondly 
\(\sigma(\Theta(\firingSequence))(i(n+1))\cap\trMC=\emptyset\).
Thus, \(\sigma(\Theta(\firingSequence))^{i(n+1)}\not\models_\mathtt{LTL}\phiLTL\MC\).
For \(i=0\) this already yields the conclusion.
In the case that \(i>0\), condition \textit{(iii)} of \refDef{cexPNwt2PN} ensures
that for all \(j\in\{0,\ldots,n\}\) a transition is added which maps to a transition of the subnet \(\tr\MC\setminus\tr\).
Hence, \(\sigma(\Theta(\firingSequence))(i(n+1)-j)\cap\tr=\empty\) for all \(j\in\{0,\ldots,n\}\),
and so \(\sigma(\Theta(\firingSequence))^{i(n+1)-j}\not\models_\mathtt{LTL}\phiLTL\MC\).
\\[2mm]
\noindent\textit{Regarding {completeness}}:
The premise yields that there is a \(j\in\{0,\ldots,n_0\}\) such that
\(\sigma(\firingSequence\MC)^{i(n+1)-j}\not\models_\mathtt{LTL}\bigvee_{t'\in\tr\MC\setminus\tr}t'\ltlUntil t\).
Definition~\ref{def:pntfl2pn} yields that every trace of \(\pNetMC\) must contain at position
\(i(n+1)\) a transition \(t_o\in\trMC_o=\tr\) or no transition at all.
Thus, for \(i=0\) we know from the premise \(t\nin\sigma(\firingSequence\MC)(0)\).
Definition~\ref{def:cexPN2PNwT} keeps for all \(i(n+1)\) positions the transitions and the mapping
for \(\Theta\MC(\firingSequence\MC)\). Hence, \(t\nin\sigma(\Theta\MC(\firingSequence\MC))(0)\).
For \(i>0\), from \refDef{pntfl2pn} follows that \(\sigma(\firingSequence\MC)(i(n+1)-j)\) belong
to some subnet and for all next positions until position \(i(n+1)-1\) only transitions \(t'\)
belonging to subnets can possibly be fired, i.e., \(t'\in\tr\MC\setminus\tr\).
Thus, with the premise \(t\nin\sigma(\firingSequence\MC)(i(n+1)\).
With the same argument as for the \(i=0\) case, \(t\nin\sigma(\Theta\MC(\firingSequence\MC))(i)\).
\\
(IS) Let \(\phiLTL\MC_1\) and \(\phiLTL\MC_2\) be created from flow formulas \(\phiLTL_1\) and \(\phiLTL_2\)
by \refDef{flowLTL2ltl} without \textit{(nSub)}.\\
\textbf{Case} \(\phiLTL=\neg \phiLTL_1\). Since \refDef{flowLTL2ltl} does not concern the negation
\(\phiLTL\MC=\neg\phiLTL\MC_1\).
\\
\noindent\textit{Regarding {soundness}}:
The premise yields \(\sigma(\firingSequence)^i\models_\mathtt{LTL} \phiLTL_1\) and because \(\Theta\MC(\Theta(\firingSequence))=\firingSequence\),
this also can be stated as \(\sigma(\Theta\MC(\Theta(\firingSequence)))^i\models_\mathtt{LTL} \phiLTL_1\).
Then the contraposition of the \emph{completeness} part of the induction hypothesis yields
\(\neg( \exists j\in\{0,\ldots,n_0\}:\sigma(\Theta(\firingSequence))^{i(n+1)-j}\not\models_\mathtt{LTL} \phiLTL_1\MC)\). Thus,
\(\sigma(\Theta(\firingSequence))^{i(n+1)-j}\models_\mathtt{LTL} \phiLTL_1\MC\) for all of those \(j\),
and so it is not satisfied for the negation, i.e., \(\phiLTL\MC\).
\\[2mm]
\noindent\textit{Regarding {completeness}}:
The premise states the existence of a \(j\in\{0,\ldots,n_0\}\) with \(\sigma(\firingSequence\MC)^{i(n+1)-j}\models_\mathtt{LTL} \phiLTL_1\MC)\).
Since \(\Theta(\Theta\MC(\firingSequence\MC))=\firingSequence\MC\), the contraposition of the \emph{soundness} part of the 
induction hypothesis yields \(\sigma(\Theta\MC(\firingSequence\MC))^i\models_\mathit{LTL}\phiLTL_1\).
Hence, \(\sigma(\Theta\MC(\firingSequence\MC))^i\not\models_\mathit{LTL}\phiLTL\).
\\
\textbf{Case} \(\phiLTL=\phiLTL_1\wedge\phiLTL_2\). Since \refDef{flowLTL2ltl} does not concern the conjunction
operator \(\phiLTL\MC=\phiLTL_1\MC\wedge\phiLTL_2\MC\).\\
\noindent\textit{Regarding {soundness}}:
The premise and the induction hypothesis for \emph{soundness} yield that 
either 
\(\forall j\in\{0,\ldots,n_0\}:\sigma(\Theta(\firingSequence))^{i(n+1)-j}\not\models_\mathtt{LTL} \phiLTL_1\MC\) or
\(\forall j\in\{0,\ldots,n_0\}:\sigma(\Theta(\firingSequence))^{i(n+1)-j}\not\models_\mathtt{LTL} \phiLTL_2\MC\)
holds.
The universal quantifier can be moved to the outer level.
Hence, \(\forall j\in\{0,\ldots,n_0\}:\sigma(\Theta(\firingSequence))^{i(n+1)-j}\not\models_\mathtt{LTL}\phiLTL_1\MC\wedge\phiLTL_2\MC\).
\\[2mm]
\noindent\textit{Regarding {completeness}}:
The premise states that there is 
\(j\in\{0,\ldots,n_0\}\) such that
\(\sigma(\firingSequence\MC)^{i(n+1)-j}\not\models_\mathtt{LTL} \phiLTL_1\MC\) or
\(\sigma(\firingSequence\MC)^{i(n+1)-j}\not\models_\mathtt{LTL} \phiLTL_1\MC\) holds.
We move the existential quantifier to the inner level and together 
with the induction hypothesis for the \emph{soundness} part we know that 
\(\sigma(\Theta\MC(\firingSequence\MC))^i\not\models_\mathtt{LTL}\phiLTL_1\) or
\(\sigma(\Theta\MC(\firingSequence\MC))^i\not\models_\mathtt{LTL}\phiLTL_2\) holds.
Hence, \(\sigma(\Theta\MC(\firingSequence\MC))^i\not\models_\mathtt{LTL}\phiLTL_1\wedge\phiLTL_2\).
\\
\textbf{Case} \(\phiLTL=\ltlNext \phiLTL_1\). Definition~\ref{def:flowLTL2ltl} yields \(\phiLTL\MC=(\bigvee_{t'\in\tr\MC\setminus\tr}t')\ltlUntil (\bigvee_{t\in\tr} t\wedge\ltlNext \phiLTL_1\MC)\vee (\ltlAlways (\neg \bigvee_{t\in\tr} t)\wedge \phiLTL_1\MC)\).
\\
\noindent\textit{Regarding {soundness}}:
The premise yields \(\sigma(\firingSequence)^i\not\models_\mathtt{LTL} \ltlNext\phiLTL_1\), i.e.,
\(\sigma(\firingSequence)^{i+1}\not\models_\mathtt{LTL} \phiLTL_1\).
The \emph{soundness} part of the induction hypothesis ensures that the statement
(*)~\(\sigma(\Theta(\firingSequence))^{(i+1)(n+1)-j}\not\models_\mathtt{LTL}\phiLTL_1\MC\)
holds for all \(j\in\{0,\ldots,n_0\}\).
We show the conclusion by contradiction. Assume, \(\sigma(\Theta(\firingSequence))^{i(n+1)-j'}\models_\mathtt{LTL} \phiLTL\MC\)
for some of those \(j'\).
In the case that the first disjunct is satisfied, 
condition \textit{(iii)} of \refDef{cexPNwt2PN} ensures if \(j'>0\)
that only transitions of the subnet are added.
Those are skipped by the until and 
since at position \(i(n+1)\) condition \textit{(ii)} of \refDef{cexPNwt2PN} ensures that 
only \(t\in\tr\) can occur \(\sigma(\Theta(\firingSequence))^{i(n+1)}\models_\mathtt{LTL}\bigvee_{t\in\tr} t\wedge\ltlNext \phiLTL_1\MC\).
This is directly given if \(j'=0\).
Hence, \(\sigma(\Theta(\firingSequence))^{i(n+1)+1}\models_\mathtt{LTL}\phiLTL_1\MC\) which is a contradiction to (*)
for \(j=n\). If the second disjunct is satisfied, we know by \(\ltlAlways (\neg \bigvee_{t\in\tr} t)\)
that never an original transition will occur in the future.
Because of \refDef{cexPNwt2PN} this is only possible if \(\firingSequence\) is finite, and therewith also \(\Theta(\firingSequence)\).
Thus, from \(i(n+1)-j\) on \(\Theta(\firingSequence)\) is stuttering and so all atomic propositions stay the same in 
the future. This is a contradiction to (*), since \(\phiLTL_1\MC\) currently holds and therewith for the whole future.
\\[2mm]
\noindent\textit{Regarding {completeness}}:
The premise states the existence of a \(j\in\{0,\ldots,n_0\}\) with \(\sigma(\firingSequence\MC)^{i(n+1)-j}\not\models_\mathit{LTL}(\bigvee_{t'\in\tr\MC\setminus\tr}t')\ltlUntil (\bigvee_{t\in\tr} t\wedge\ltlNext \phiLTL_1\MC)\vee (\ltlAlways (\neg \bigvee_{t\in\tr} t)\wedge \phiLTL_1\MC)\).
For \(j>0\), \refDef{pntfl2pn} ensures that the starting position of the trace corresponds to a situation with an active subnet.
Furthermore, it ensures that only transition \(t'\) of the subnet can be used until position \(i(n+1)\), i.e., \(t'\in\trMC\setminus\tr\).
Since those are skipped within the first disjunct, this yields \(\sigma(\firingSequence\MC)^{i(n+1)}\not\models_\mathit{LTL}\bigvee_{t\in\tr} t\wedge\ltlNext \phiLTL_1\MC\).
For \(j=0\) this directly holds, since \refDef{pntfl2pn} forces every trace of \(\pNetMC\) to start with a transition
\(t\in\tr\) or to contain no transition at all.
Since position \(i(n+1)\) belongs to situations where it is the turn of the original part of the net,
either no original transition is ever fired again (see case of the other disjunct) or
\(\sigma(\firingSequence\MC)^{i(n+1)+1}\not\models_\mathit{LTL}\phiLTL_1\MC\).
For \(j'=n\) the induction hypothesis yields \(\sigma(\Theta\MC(\firingSequence\MC))^{i+1}\not\models_\mathit{LTL}\phiLTL_1\)
and so \(\sigma(\Theta\MC(\firingSequence\MC))^{i}\not\models_\mathit{LTL}\phiLTL\).
If no original transition is ever fired again, the second disjunct yields
\(\sigma(\firingSequence\MC)^{i(n+1)-j}\not\models_\mathit{LTL}\phiLTL_1\MC\).
Since \refDef{pntfl2pn} yields that only transitions \(t\in\tr\) are firable in the \(i(n+1)\) positions of the trace,
this situation must correspond to the stuttering of the trace.
Since all atomic propositions stay the same during the stuttering we know
\(\sigma(\firingSequence\MC)^{i(n+1)+1}\not\models_\mathit{LTL}\phiLTL_1\MC\).
Hence, like in the other case the induction hypothesis yields the conclusion.
\\
\textbf{Case} \(\phiLTL=\phiLTL_1\ltlUntil\phiLTL_2\). Since Definition~\ref{def:flowLTL2ltl} does not concern the until operator
\(\phiLTL\MC=\phiLTL_1\MC\ltlUntil\phiLTL_2\MC\).\\
\noindent\textit{Regarding {soundness}}:
The premise yields \(\forall k\geq 0: \sigma(\firingSequence)^{i+k}\not\models_\mathtt{LTL}\phiLTL_2\vee\exists 0\leq l<k:\sigma(\firingSequence)^{i+l}\not\models_\mathtt{LTL}\phiLTL_1\).
The induction hypothesis applied for all these \(k\)s and \(l\)s yields:
\(\forall k\geq 0: \forall j\in\{0,\ldots,n_0\}: \sigma(\Theta(\firingSequence))^{(i+k)(n+1)-j}\not\models_\mathtt{LTL} \phiLTL_2\MC \vee\exists 0\leq l<k:\forall j\in\{0,\ldots,n_0\}:\sigma(\Theta(\firingSequence))^{(i+l)(n+1)-j}\not\models_\mathtt{LTL} \phiLTL_1\MC\).
For any \(k\) we know that starting from \(i(n+1)\) either all \(k(n+1)-j\) steps are not satisfying \(\phiLTL_2\MC\)
or there is an \(0\leq l < k\) such that all \(l(n+1)-j\) steps are not satisfying \(\phiLTL_1\MC\). Since the \(j\)
ensures that this exactly holds for all steps between \(k\) and \(k+1\), we can also shift the indices such that
\(\forall k\geq 0: \forall j\in\{0,\ldots,n_0\}: \sigma(\Theta(\firingSequence))^{i(n+1)+k-j}\not\models_\mathtt{LTL} \phiLTL_2\MC \vee\exists 0\leq l<k:\forall j\in\{0,\ldots,n_0\}:\sigma(\Theta(\firingSequence))^{i(n+1)+l-j}\not\models_\mathtt{LTL} \phiLTL_1\MC\)
holds. By moving the universal quantifier of the \(j\) to the outside, we obtain the semantical definition of the 
until operator. Hence, \(\forall j\in\{0,\ldots,n_0\}:\sigma(\Theta(\firingSequence))^{i(n+1)-j}\not\models_\mathtt{LTL}\phiLTL_1\MC\ltlUntil\phiLTL_2\MC\).\\[2mm]
\noindent\textit{Regarding {completeness}}:
The premise yields that there is a \(j\in\{0,\ldots,n_0\}\) with 
\(\forall k\geq 0: \sigma(\firingSequence\MC)^{i(n+1)-j+k}\not\models_\mathtt{LTL}\phiLTL_2\MC\vee\exists 0\leq l<k:\sigma(\firingSequence\MC)^{i(n+1)-j+l}\not\models_\mathtt{LTL}\phiLTL_1\MC\).
We can move the existential quantifier for the \(j\) inwards, i.e.,
\(\forall k\geq 0:\exists j\in\{0,\ldots,n_0\}:\sigma(\firingSequence\MC)^{i(n+1)-j+k}\not\models_\mathtt{LTL}\phiLTL_2\MC\vee\exists 0\leq l<k:\exists j'\in\{0,\ldots,n_0\}:\sigma(\firingSequence\MC)^{i(n+1)-j'+l}\not\models_\mathtt{LTL}\phiLTL_1\MC\).
Since this holds for every \(k\) and thus, especially for every \(k(n+1)\).
Since the existence of the \(j'\) can be used to alter the index to the previously existing position \(l\) for every \(k\)
we know \(\forall k\geq 0:\exists j\in\{0,\ldots,n_0\}:\sigma(\firingSequence\MC)^{i(n+1)-j+k(n+1)}\not\models_\mathtt{LTL}\phiLTL_2\MC\vee\exists 0\leq l<k:\exists j'\in\{0,\ldots,n_0\}:\sigma(\firingSequence\MC)^{i(n+1)-j'+l(n+1)}\not\models_\mathtt{LTL}\phiLTL_1\MC\).
The induction hypothesis then yields
\(\forall k\geq 0:\sigma(\Theta\MC(\firingSequence\MC))^{i+k}\not\models_\mathtt{LTL}\phiLTL_2\vee\exists 0\leq l<k:\sigma(\Theta\MC(\firingSequence\MC))^{i+l}\not\models_\mathtt{LTL}\phiLTL_1\), which is the semantical definition of the conclusion.
\hfill\qed
\end{proof}

\begin{lemma}[Flow Part]
\label{lem:flowpart}
Given an LTL formula \(\phiLTL\MC\) created by \refDef{flowLTL2ltl} without condition \textit{(nSub)}
from a flow formula \(\A\,\phiLTL_i\) of a  flow-LTL formula \(\phiRun\).
\begin{itemize}
\item[(s)]
Given a run \(\runPN\) and a covering firing sequence \(\firingSequence\), with
 \(\sigma(\Theta(\firingSequence))^j\models_\mathit{LTL}\ltlAlways\ltlEventually\act{o}\) for all \(j\in\N\), then
\(\runPN, \sigma(\firingSequence)\not\models \A\,\phiLTL_i\implies \sigma(\Theta(\firingSequence))\not\models_\mathtt{LTL} \phiLTL\MC\)
holds.
\item[(c)]
Given a firing sequence \(\firingSequence\MC\) of a run of the net \(\pNetMC\) with
 \(\sigma(\firingSequence\MC)^j\models_\mathit{LTL}\ltlAlways\ltlEventually\act{o}\) for all \(j\in\N\), then
\(\sigma(\firingSequence\MC)\not\models_\mathit{LTL}\phiLTL\MC\implies \Theta_R\MC(\firingSequence\MC), \sigma(\Theta\MC(\firingSequence\MC))\not\models \A\,\phiLTL_i\)
holds.
\end{itemize}
\end{lemma}
\begin{proof}[via structural induction over \(\phiLTL_i\)]
We know from \refDef{flowLTL2ltl} that \(\phiLTL\MC=\initsub_i\LTLweakuntil(\neg \initsub_i\wedge\phiLTL_i\MC)\) where \(\phiLTL_i\MC\) is created
from \(\phiLTL_i\) by the constraints \textit{(pF)}, \textit{(tF)}, and \textit{(nF)}.
Since the operator \(\LTLweakuntil\) is an abbreviation \(\phiLTL\MC=(\initsub_i\ltlUntil(\neg \initsub_i\wedge\phiLTL_i\MC)) \vee \ltlAlways\initsub_i\).\\[0.5mm]
\noindent\textit{Regarding \textit{(s)}}.
The premise yields that there is a flow chain \(\flowChain^i=p_0^i,t_0^i,p_1^i, t_1^i,\ldots\) 
such that \(\sigma(\flowChain^i)\not\models_\mathit{LTL} \phiLTL_i\).
Since there is a chain, \refDef{cexPNwt2PN} yields the existence of a transition \(t\)
in \(\Theta(\firingSequence)\) at position \(\Theta_i(\firingSequence,0)-1\)
starting the chain, i.e., \(\rho\MC(t)=\netId(t_o)_{\netId(\rho(p_0^i))_i}\)
for an original transition \(t_o\in\tr_o\MC=\tr\)
with \((\startfl, \rho(p_0^i))\in\tfl(t_o)\).
Condition \textit{(s3)} of \refDef{pntfl2pn} ensures that 
at position \(\Theta_i(\firingSequence,0)\) the place \(\initsub_i\) is unoccupied.
Hence, \(\sigma(\Theta(\firingSequence))\not\models_\mathtt{LTL}\ltlAlways\initsub_i\).
Furthermore, from condition \textit{(in)} follows that \(\initsub_i\) is initially marked.
Definition~\ref{def:cexPNwt2PN} states that \(t\) is the first occurrence of such kind
and \refDef{pntfl2pn} ensure that no other kind takes a token from \(\initsub_i\).
Thus, \(\initsub_i\) is satisfied until position \(\Theta_i(\firingSequence,0)\).
Hence, we only have to show that in such situations
\(\sigma(\Theta(\firingSequence))^{\Theta_i(\firingSequence,0)}\not\models_\mathtt{LTL}\phiLTL_i\MC\)
holds.\\[2mm]
\noindent\textit{Regarding \textit{(c)}}.
The premise states that \(\sigma(\firingSequence\MC)\not\models_\mathit{LTL}(\initsub_i\ltlUntil(\neg \initsub_i\wedge\phiLTL_i\MC)) \vee \ltlAlways\initsub_i\) holds.
Since \(\sigma(\firingSequence\MC)\not\models_\mathit{LTL}\ltlAlways\initsub_i\) \refDef{pntfl2pn} yields
the existence of a transition \(t\) in \(\firingSequence\MC\) with \(\rho\MC(t)=\netId(t_o)_{\netId(q)_i}\)
for an original transition \(t_o\in\tr\MC=\tr\) and a place \(q\in\pl\) with \((\startfl,q)\in\tfl(t_o)\).
Thus, \refDef{cexPN2PNwT} yields the existence of a flow chain \(\Theta_{\flowChain^i}\MC(\firingSequence\MC)\)
and that the first place of the chain corresponds to index \(\Theta_i\MC(\firingSequence\MC,0)\).
Since we know that until this position \(\initsub_i\) is satisfied, the premise yields 
\(\sigma(\firingSequence\MC)^{\Theta_i\MC(\firingSequence\MC,0)}\not\models_\mathit{LTL}\phiLTL_i\MC\).
Hence, we have to show that in those situations \(\sigma(\Theta_{\flowChain^i}\MC(\firingSequence\MC))^0\not\models_\mathtt{LTL}\phiLTL_i\)
holds.

We show these two cases via an induction over \(\phiLTL_i\).
We define two sets \(\Delta=\{\Theta_i(\firingSequence,k),\ldots,\Theta_i(\firingSequence,k+1)\}\) and
\(\Delta\MC=\{\Theta_i\MC(\firingSequence\MC,k),\ldots,\Theta_i\MC(\firingSequence\MC,k+1)\}\) and
show that
\emph{soundness}, i.e.,
\[\forall k\in\N: \sigma(\flowChain)^k\not\models_\mathtt{LTL}\phiLTL_i \implies \forall j\in\Delta:\sigma(\Theta(\firingSequence))^j\not\models_\mathtt{LTL} \phiLTL_i\MC\]
and \emph{completeness}, i.e.,
\[\forall k\in\N: \exists j\in\Delta\MC:\sigma(\firingSequence\MC)^j\not\models_\mathtt{LTL} \phiLTL\MC \implies \sigma(\Theta_{\flowChain^i}\MC(\firingSequence\MC))^k\not\models_\mathtt{LTL}\phiLTL_i\]
hold.
The induction is similarly to the induction in the proof of \refLemma{LTLpart}.
Since in the \emph{(tF)} and \emph{(nF)} constraints of \refDef{flowLTL2ltl} also all time points are skipped
which do not concern the current part of the formula, the different time lines are handled properly.
The difference in this case is that not after exactly \(n\) steps the entry concern
the current part of the formula
but there could be more rounds not concerning the considered flow chain.
This is adequately handled by allowing the
skipping transition in \(O_i\).
We cannot skip in situations where we should not skip,
i.e., in this round a transition \(t_o\in\tr_o\) has fired which moves the current chain, because of the 
inhibitor arcs of the skipping transition.
Another key part is that, as for the LTL part, the tokens of the subnet (apart from the active one) are
independent from every firing of any transition of the other subnets.
\hfill\qed
\end{proof}

\begin{lemma}[Soundness]
\label{lem:sound}
Given a run \(\runPN\), a covering firing sequence \(\firingSequence\),
and an LTL formula \(\phiMC\) created by \refDef{flowLTL2ltl} from a Flow-LTL formula \(\phiRun\).
Then:
\[\runPN,\sigma(\firingSequence)\not\models\phiRun \implies \sigma(\Theta(\firingSequence))\not\models_\mathtt{LTL} \phiMC \]
\end{lemma}
\begin{proof}[via structural induction over \(\phiRun\)]\\
Definition \ref{def:flowLTL2ltl} yields that every transformed formula \({\phiMC}'\) is of the form \(\ltlAlways\ltlEventually\act{o}\rightarrow\phiRun^\A\).
Since \refDef{cexPNwt2PN} adds for every existing transition in \(\firingSequence\) also the \(n\) transitions (one for each subnet)
to \(\Theta(\firingSequence)\) (condition (3)) and by condition \textit{(mSn)} of \refDef{pntfl2pn} every transition of 
the last subnet puts a token onto \(\act{o}\), the statement \(\sigma(\Theta(\firingSequence))^i\models_\mathit{LTL}\ltlAlways\ltlEventually\act{o}\)
is satisfied for every \(i\in\N\). This also holds for finite firing sequences, because \refDef{cexPNwt2PN} ensures that
the last added transition is one of the last subnet. The stuttering then yields the satisfaction.
Thus, we only have to show that \(\sigma(\Theta(\firingSequence))\not\models_\mathtt{LTL}\phiRun^\A\).
Let \(\phiMC\) be such a subformula \(\phiRun^\A\).
\\
(IB) \textbf{Case} \(\phiRun=\phiLTL\), for a standard LTL formula \(\phiLTL\). Lemma~\ref{lem:LTLpart} directly yields
\(\sigma(\Theta(\firingSequence))\not\models_\mathtt{LTL} \phiMC\).
\textbf{Case} \(\phiRun=\A\,\phiLTL_i\).
\refLemma{flowpart} proves this case.\\
(IS)
\textbf{Case} \(\phiRun=\phiRun_1\wedge\phiRun_2\).
Since \refDef{flowLTL2ltl} does not concern the conjunction operator of the run part, there are subformulas
\(\phiMC_1\) and \(\phiMC_2\) with \(\phiMC=\phiMC_1\wedge\phiMC_2\),
which are created from the corresponding \(\phiRun_1\) and \(\phiRun_2\), respectively.
The induction hypothesis yields that \(\sigma(\Theta(\firingSequence))\not\models_\mathtt{LTL} \phiMC_1\) or
\(\sigma(\Theta(\firingSequence))\not\models_\mathtt{LTL} \phiMC_2\), thus \(\sigma(\Theta(\firingSequence))\not\models_\mathtt{LTL} \phiMC_1\wedge\phiMC_2\).\\
\textbf{Case} \(\phiRun=\phiRun_1\vee\phiRun_2\).
Since \refDef{flowLTL2ltl} also does not concern the disjunction operator of the run part, this case is analogously done 
as the previous case.\\
\textbf{Case} \(\phiRun=\phiLTL\rightarrow\phiRun_2\).
Since \refDef{flowLTL2ltl} also does not concern the implication operator there are subformulas
\(\phiMC_1\) and \(\phiMC_2\) with \(\phiMC=\phiMC_1\rightarrow\phiMC_2\),
which are created from \(\phiLTL\) and \(\phiRun_2\), respectively.
The premise of the statement yields \(\runPN,\sigma(\firingSequence)\models\phiLTL\),
thus \(\runPN,\sigma(\firingSequence)\not\models\neg\phiLTL\), and 
\(\runPN,\sigma(\firingSequence)\not\models\phiRun_2\).
Since \(\neg\phiLTL\) is still a standard LTL formula, \refLemma{LTLpart} yields 
\(\sigma(\Theta(\firingSequence))\not\models_\mathtt{LTL} \neg\phiMC_1\) (the \refDef{flowLTL2ltl} does not concern the negation).
The induction hypothesis ensures \(\sigma(\Theta(\firingSequence))\not\models_\mathtt{LTL}\phiRun_2\), and so 
\(\sigma(\Theta(\firingSequence))\not\models_\mathtt{LTL} \phiMC_1\rightarrow\phiMC_2\).\hfill\qed
\end{proof}

\begin{lemma}[Completeness]
\label{lem:completeness}
Given a firing sequence \(\firingSequence\MC\) of a run of \(\pNetMC\),
and an LTL formula \(\phiMC\) created by \refDef{flowLTL2ltl} from a flow-LTL formula \(\phiRun\).
Then:
\[\sigma(\firingSequence\MC)\not\models_\mathtt{LTL}\phiMC\implies\Theta_R\MC(\firingSequence\MC),\sigma(\Theta\MC(\firingSequence\MC))\not\models\phiRun\]
\end{lemma}
\begin{proof}[via structural induction over \(\phiRun\)]\\
Again, every \({\phiMC}'\) is of the form \(\ltlAlways\ltlEventually\act{o}\rightarrow\phiRun^\A\).
The premise of the statement therewith yields
\(\sigma(\firingSequence\MC)\models_\mathtt{LTL}\ltlAlways\ltlEventually\act{o}\).
Every other argument is analog to the arguments of the proof of \refLemma{sound}.
\hfill\qed
\end{proof}

Finally, we are able to prove \refLemma{correctnessPNLTL}.

\begin{proof}[\refLemma{correctnessPNLTL} (Correctness of the Transformation)]\\
\noindent\textit{Regarding soundness:}
We show the contraposition of the statement: \(\pNet \not\models \phiRun\implies \pNetMC\not\models_\mathtt{LTL}\phiMC\).
Hence, there is a run \(\runPN\) of \(\pNet\) and a covering firing sequence \(\firingSequence\)
such that \(\runPN,\sigma(\firingSequence)\not\models\phiRun\).
\refLemma{sound} yields that the firing sequence \(\Theta(\firingSequence)\)
fulfills \(\sigma(\Theta(\firingSequence))\not\models_\mathtt{LTL}\phiMC\).
Thus, there exists a run \({\runPN}\MC\) (created from \(\Theta(\firingSequence)\) by iteratively adding the places of the markings
and the transitions of the firing sequence and connecting them according their corresponding places and transitions in \(\pNetMC\)) which 
is covered by \(\Theta(\firingSequence)\), such that \({\runPN}\MC\not\models_\mathtt{LTL}\phiMC\), and thus \(\pNetMC\not\models_\mathtt{LTL}\phiMC\).

\noindent\textit{Regarding completeness.}
We analogously show the contraposition of the statement: \(\pNetMC \not\models_\mathtt{LTL} \phiMC\implies \pNet\not\models\phiRun\).
Hence, there is a run \({\runPN}\MC\) of \(\pNetMC\) and a covering firing sequence \(\firingSequence\MC\)
such that \(\sigma(\firingSequence\MC)\not\models_\mathtt{LTL}\phiMC\).
\refLemma{completeness} yields that the firing sequence \(\Theta\MC(\firingSequence\MC)\)
fulfills \(\Theta_R\MC(\firingSequence\MC),\sigma(\Theta\MC(\firingSequence\MC))\not\models\phiRun\).
Hence, \(\pNet\not\models\phiRun\).
\hfill\qed
\end{proof}

\subsection{Construction of the Circuit}
We formally define the circuit \(\circuit_\pNet\)
for a P/T Petri net with inhibitor arcs \(\pNet=(\pl,\tr,\fl,\inhibitorFl,\init)\)
and the transformed formula \(\phiLTL'\) of an LTL formula \(\phiLTL\)
described in \refSection{pnMCwithCircuits}.

For a circuit \(\circuit=(\cin,\cout,\clatches,\cformula)\)
with a Boolean formula \(\cformula\) over \(\cin\times \clatches \times \cout\times \clatches\),
we use decoration to express the correspondence of the variables
in the second and fourth component \(\clatches\).
Thus, if \(x\) denotes the current value of a latch in the second component \(\clatches\)
then \(x'\) denotes the new value of that latch after the next clock pulse in the fourth component of \(\clatches\).
This decoration is also lifted to sets.
For a tuple \((I,L,O,L')\in 2^\cin\times 2^\clatches\times 2^\cout\times 2^\clatches\)
we say \((I,L,O,L')\) \emph{satisfies} \(\cformula\) (denoted by \((I,L,O,L')\models \cformula\))
iff \(\cformula\) is satisfied under the valuation which maps each occurring variable to \(\true\) and all others to \(\false\).
\begin{definition}[P/T Petri Net to Circuit]
\label{def:pn2circuit}
For a P/T Petri net with inhibitor arcs \(\pNet=(\pl,\tr,\fl,\inhibitorFl,\init)\),
we define the \emph{circuit} \(\circuit_\pNet=(\cin,\cout,\clatches,\cformula)\) with
the set of \emph{input variables} \(\cin=\tr\),
the set of \emph{output variables} \(\cout=\{p_o\with p\in\pl\}\cup\{t_o\with t\in\tr\}\cup\{e_o\}\), 
the set of \emph{latches} \(\clatches=\pl\cup\{\mathtt{i},\mathtt{e}\}\) with an \emph{initialisation latch} \(\mathtt{i}\) and a latch for \emph{invalid inputs} \(\mathtt{e}\),
and a \emph{boolean formula} \(\cformula=\mathtt{out}_P\wedge\mathtt{out}_T\wedge\mathtt{out_e}\wedge\mathtt{latch_e}\wedge\mathtt{latch_i}\wedge\mathtt{latch}_P\)
over \(\cin\times \clatches\times \cout\times \clatches\)
which is defined with the help of the following formulas:
\begin{align*}
\mathtt{val}(t)&\defEQ t \wedge \bigwedge_{t_1\in\tr\setminus\{t\}}\neg t_1\wedge\bigwedge_{p\in \preset{t}} \left\{\begin{array}{rl}\neg p &\text{if } (p,t)\in\inhibitorFl \\ p & \text{otherwise}\end{array}\right.,\\
\mathtt{noT}&\defEQ \bigwedge_{t\in\tr} \neg \mathtt{val}(t),\\
\mathtt{succ}(p)&\defEQ (\mathtt{noT}\rightarrow p) \wedge (\neg\mathtt{noT}\rightarrow\bigwedge_{t\in\tr} (t_o\rightarrow\left\{\begin{array}{ll} p & \text{if } p\nin\preset{t}\wedge p\nin\postset{t} \\ 0 & \text{if } p\in\preset{t}\wedge p\nin\postset{t}\\1 &\text{otherwise}\end{array}\right.)).
\end{align*}
The formula \(\mathtt{val}(t)\) for a \(t\in\tr\) states the validity of \(t\), i.e., \(t\) is set as input but
no other transition is set and \(t\) is enabled by the current state of the latches.
The formula \(\mathtt{noT}\) is true iff no transition is valid and
the formula \(\mathtt{succ}(p)\) for a place \(p\in\pl\) defines the successor value for \(p\).
If there is no valid input we keep the same marking.
Otherwise, the marking is the successor marking of the current output transition \(t_o\) and the current marking.
Therewith, the conjuncts of \(\cformula\) are defined as follows:
\begin{align*}
\mathtt{out_P}&\defEQ\bigwedge_{p\in\pl} (p_o\leftrightarrow (\neg\mathtt{i}\rightarrow p')\wedge(\mathtt{i}\rightarrow p)),\\
\mathtt{out_T}&\defEQ\bigwedge_{t\in\tr} (t_o\leftrightarrow \mathtt{val}(t)),\\
\mathtt{out_e}&\defEQ e_o \leftrightarrow \mathtt{e},\\ 
\mathtt{latch_e}&\defEQ e' \leftrightarrow \mathtt{i}\wedge\mathtt{noT},\\
\mathtt{latch_i}&\defEQ i' \leftrightarrow \true,\\
\mathtt{latch}_P&\defEQ \bigwedge_{p\in\pl} (p' \leftrightarrow \left\{\begin{array}{ll}\mathtt{i}\rightarrow \mathtt{succ}(p)&\text{if } p\in\init\\\mathtt{i}\wedge\mathtt{succ}(p)&\text{otherwise}\end{array}\right.).
\end{align*}
In all but the initial state, the outputs corresponding to places
are the current values of the latches.
The outputs corresponding to the transitions are at most one valid transition.
The new value for the latches corresponding to places are initially the 
initial marking of \(\pNet\).
Otherwise, if no valid input is applied, the current values of the latches are copied to the new values
and if there is a valid transition, the successor marking of firing this transition in the current values is used for the new values.
\end{definition}
For all subsets \(P\subseteq\pl\) and \(T\subseteq\tr\), we define with \(P_o=\{p_o\in\cout\with p\in P\}\)
and \(T_o=\{t_o\in\cout\with t\in T\}\) the respective sets of the output variables.
We define the \emph{transformed formula} \(\phiLTL'\) by skipping the initialisation step and 
focusing on the valid traces:
\[\phiLTL'=\ltlNext(\ltlAlways(e_o\rightarrow\ltlAlways e_o)\rightarrow \widetilde{\phiLTL})\]
where \(\widetilde{\phiLTL}\) is obtained from \(\phiLTL\) by replacing every place and transition with the 
corresponding output variable.

A \emph{Kripke structure} is a five-tuple
\(\kripke=(\katoms, \kstates, \kinit,\krel,\klab)\), with
a set of \emph{atoms}~\(\katoms\),
a set of \emph{states}~\(\kstates\),
a set of \emph{initial states}~\(\kinit\,\subseteq\kstates\),
a \emph{transition relation}~\(\krel\subseteq\kstates\times\kstates\),
and a \emph{labelling function}~\(\klab:\kstates\to 2^\katoms\).
A \emph{path}~\(\pi=\pi_0\pi_1\cdots\) of a Kripke structure is an
infinite sequence of states \(\pi_i\in\kstates\) for \(i\in\N\)
with \((\pi_i,\pi_{i+1})\in\krel\). The path \(\pi\) is \emph{initial} iff \(\pi_0\in\kinit\).
A Kripke structure \(\kripke\) \emph{satisfies} an LTL formula \(\phiLTL\)
(denoted by \(\kripke\models \phiLTL\)) iff every initial path satisfies \(\phiLTL\).

\begin{definition}[Circuit to Kripke Structure]
\label{def:circuit2kripke}
The \emph{Kripke structure} \(\kripke_\circuit\) of a circuit \(\circuit\) is defined by
\(\kripke_\circuit=(\katoms, \kstates, \kinit,\krel,\klab)\), with
the set of \emph{atoms} \(\katoms=\cout\),
the set of \emph{states} \(\kstates=2^\cin\times 2^\clatches\times 2^\cout\times 2^\clatches\),
the set of \emph{initial states} \(\kinit=\{(I,\emptyset,\init_o,\init'\cup\{\mathtt{i}'\} \with I\subseteq\tr\}\),
the \emph{transition relation} \(\krel=\{((I_1,L_1,O_1,L_1'),(I_2,L_2,O_2,L_2'))\in\kstates\times\kstates \with L_1'=L_2 \wedge (I_2,L_2,O_2,L_2') \models \cformula\}\),
and the \emph{labelling function} \(\klab=\{((I,L,O,L'),O)\in\kstates\times 2^\cout\}\).
\end{definition}

\subsection{Correctness of the Constructions}
We show the correctness of the construction of the circuit \(\circuit_\pNet\) and 
the corresponding Kripke structure \(\kripke_{\circuit_\pNet}\) by proving
\(\pNet\not\models_\mathtt{LTL} \phiLTL\text{ iff }\kripke_{\circuit_\pNet}\not\models \phiLTL'\).
For this purpose, we mutually transform the respective counterexamples and show their 
correspondence on the atomic propositions of \(\phiLTL\) and \(\phiLTL'\), respectively.
This already yields the correctness part of \refLemma{corCircuit}.
Finally, we have everything at hand to prove \refTheo{correctness}.

Given a trace \(\sigma(\firingSequence)\)
of a firing sequence \(\firingSequence\)
covering a run \(\runPN\)
of a safe P/T Petri net with inhibitor arcs \(\pNet\)
and a path \(\pi=s_0s_1\cdots\) of the corresponding Kripke structure 
\(\kripke_{\circuit_\pNet}=(\katoms, \kstates, \kinit,\krel,\klab)\).
We say an entry \(z_i=\sigma(\firingSequence)(i)\in 2^{\pl\cup\tr}\) of the trace
and an element \(s_j=(I_j,L_j,O_j,L_j')\in\kstates\) of the path \emph{coincide} (denoted by \(z_i\sim s_j\)) iff
\(z_{|\tr}=\{t\in\tr\with t_o\in O_j\wedge e_o\nin O_j\}\) and \(z_{|\pl}=\{p\in\pl\with p_o\in O_j\}\).
Where \(z_{|\tr}\) and \(z_{|\pl}\) are the projections onto the respective sets.

\begin{proof}[\refLemma{corCircuit} (Correctness of the Circuit)]
The number of latches and gates directly follows from \refDef{pn2circuit}.
No formula has more than two nesting conjunctions over \(\pl\) or \(\tr\).
The correctness is proven on the corresponding Kripke structure \(\kripke_{\circuit_\pNet}\)
via contraposition.
We show \(\pNet\not\models_\mathtt{LTL} \phiLTL\) iff \(\kripke_{\circuit_\pNet}\not\models \phiLTL'\)
by transforming the counterexamples.\\
\emph{Soundness:} Let \(\pNet\not\models_\mathtt{LTL} \phiLTL\).
Thus, there is a run \(\runPN=(\pNet^R,\rho)\)
and a covering firing sequence \(\firingSequence=M_0\firable{t_0}M_1\firable{t_1}\cdots\)
such that \(\sigma(\firingSequence)\not\models_\mathtt{LTL} \phiLTL\).
We create a path \(\pi=\pi_0\pi_1\cdots\) of \(\kripke_{\circuit_\pNet}\) by 
\(\pi_0=(\emptyset,\emptyset,\rho(M_0)_o,\rho(M_0)'\cup\{\mathtt{i}'\})\)
and if \(\firingSequence\) is infinite
\[\pi_i=(\{\rho(t_{i-1})\},\rho(M_{i-1})\cup\{\mathtt{i}\},\{\rho(t_{i-1})_o\}\cup\rho(M_{i-1})_o,\rho(M_{i})'\cup\{\mathtt{i}'\})\]
for all \(i\in\N\setminus\{0\}\).
If \(\firingSequence=M_0\firable{t_0}M_1\firable{t_1}\cdots\firable{t_{n-1}}M_n\) is finite 
the first \(n\) states are created as above and for the first \(j>n\) we define
\[\pi_j=(\emptyset,\rho(M_{n})\cup\{\mathtt{i}\},\rho(M_{n})_o,\rho(M_{n})'\cup\{\mathtt{i}',\mathtt{e}'\})\]
and for all other we define \(\pi_j\) equally but that the current values of the latches also contain \(\mathtt{e}\)
and that also \(e_o\) is set for the outputs.

We show that the sequence \(\pi\) is indeed an initial path of \(\kripke_{\circuit_\pNet}\).
Since \(\rho(M_0)=\init\), because \(\firingSequence\) covers a run of \(\pNet\), it directly holds \(\pi_0\in\kinit\).
We show that for all \(i\in\N:(\pi_i,\pi_{i+1})\in\krel\), i.e., \(L_i'=L_{i+1}\) and \((I_{i+1},L_{i+1},O_{i+1},L_{i+1}') \models \cformula\)
for \(\pi_i=(I_i,L_i,O_i,L_i')\) and \(\pi_{i+1}=(I_{i+1},L_{i+1},O_{i+1},L_{i+1}')\).
The first clause can directly be seen by the definition of \(\pi\).
We show that all the defined \(\pi_i\) in the finite as well as the infinite case satisfy \(\cformula\) by checking 
each of the conjuncts:
The conjunct \(\mathtt{out}_P\) is satisfied, because the current values \(p\) are set as output (\(\mathtt{i}\) is true)
and this fits to the definition of \(\pi_i\).
The conjunct \(\mathtt{out}_T\) is satisfied, because \(t_{i-1}\) is enabled in \(M_{i-1}\) because of \(\firingSequence\)
and no other transition is set in \(\pi_i\).
The conjunct \(\mathtt{latch}_P\) is satisfied, because \(\mathtt{succ}(p)\) yields 
the marking resulting by firing \(t_{i-1}\) in the current values of the latches (here \(\rho(M_{i-1})\)) because \(\mathtt{noT}\) is not satisfied.
Because of \(\firingSequence\), this is \(\rho(M_i)\), which fits \(\pi_i\).
The other conjuncts are directly satisfied by the construction.
In the case of a finite firing sequence the defined \(\pi_j\) also satisfy \(\cformula\):
The conjunct \(\mathtt{out}_P\) is satisfied, with the same arguments as in the previous case.
The conjunct \(\mathtt{latch}_P\) is satisfied, because \(\mathtt{noT}\) is satisfied and therewith the current values
of the latches are taken 
and \(\mathtt{out}_T\) is also satisfied, because no transition is applied to the input.
Also because of \(\mathtt{noT}\) the conjunct \(\mathtt{latch}_\mathtt{e}\) is satisfied
and with the additional constraints for all but the first \(j>n\) the conjunct \(\mathtt{out}_\mathtt{e}\) is satisfied
for all \(i\in\N\).
The subpath \(\pi^1\) satisfies \(\ltlAlways(e_o\rightarrow \ltlAlways e_o)\), because no \(e_o\) is ever set by the 
construction in the case of an infinite \(\firingSequence\) and in the other case \(e_o\) is set for every \(j+1>n\).
Since by the construction \(\sigma(\firingSequence)(i)\sim \phi_{i+1}\) directly holds,
the path \(\pi\) does not satisfy \(\ltlNext(\ltlAlways(e_o\rightarrow\ltlAlways e_o)\rightarrow \widetilde{\phiLTL})\).
Hence, \(\kripke_{\circuit_\pNet}\not\models \phiLTL'\).
\\
\emph{Completeness:} Let \(\kripke_{\circuit_\pNet}\not\models \phiLTL'\).
Thus, there is an initial path \(\pi=\pi_0\pi_1\cdots\) of \(\kripke_{\circuit_\pNet}\) not satisfying 
\(\ltlNext(\ltlAlways(e_o\rightarrow\ltlAlways e_o)\rightarrow \widetilde{\phiLTL})\).
Hence, the subpath \(\pi^1\) satisfies \(\ltlAlways(e_o\rightarrow\ltlAlways e_o)\) and not \(\widetilde{\phiLTL}\).
In the case that \(e_o\nin O_i\) for all \(i\in\N\) holds, we create a infinite firing sequence
\(\firingSequence=M_0\firable{t_0}M_1\firable{t_1}\cdots\) with
\(\rho(M_i)=\{p\in\pl\with p_o\in O_{i+1}\}\) and \(\rho(t_i)\in\{t\in\tr\with t_o\in O_{i+1}\}\).
Otherwise, if there is an \(i\in\N\) with \(e_o\nin O_i\), say \(i=n+1\) is the first of such occurrences,
we create a finite firing sequence
\(\firingSequence=M_0\firable{t_0}M_1\firable{t_1}\cdots\firable{t_{n-1}}M_n\) 
for the first \(n-1\) steps as before and for the \(n\)th step we
define \(\rho(M_n)=\{p\in\pl\with p_o\in O_{n+1}\}\) as above.

This is indeed a firing sequence, because \(\rho(M_0)=\init\) holds
by the definition of \(\kripke_{\circuit_\pNet}\), the construction of \(\firingSequence\), and since \(\pi\) is initial.
All transitions are firable and yield the respective successor marking because of \(\cformula\),
since \(\pi\) is a path.
This is because the error flag is maximally set in the last step, and thus, there must be exactly one enabled transition applied to the inputs.
Therewith \(\mathtt{succ}(p)\) yields the correct successor marking and this is the output of the next step.
Since \(\pi\) satisfies \(\ltlNext(\ltlAlways(e_o\rightarrow\ltlAlways e_o)\),
we know that also in the finite case of \(\firingSequence\) it holds
\(\sigma(\firingSequence)(i)\sim \phi_{i+1}\), since in \(\pi\) the output markings stay the same in situations where \(e_o\) is
set, because this is only possible if \(\mathtt{notT}\) is true.
Since \(\pi^1\) does not satisfy \(\widetilde{\phiLTL}\), \(\sigma(\firingSequence)\not\models_\mathtt{LTL}\phiLTL\).
Constructing the corresponding run yields the conclusion.
\hfill\qed
\end{proof}

\begin{proof}[\refTheo{correctness}]
For a Petri net with transits \(\pNet\) and a Flow-LTL formula~\(\phiRun\),
Lemma~\ref{lem:sizePN} and \ref{lem:sizeFormula} yield the polynomial size of the constructed P/T Petri net with inhibitor arcs \(\pNetMC\)
and the constructed formula \(\phiMC\). Lemma~\ref{lem:correctnessPNLTL} yields the equivalent satisfiability
and \refLemma{corCircuit} shows the polynomial size of the constructed, satisfiability-equivalent circuit \(\circuit_\pNet\)
and of the formula \(\phiLTL'\).
This results in a Kripke structure of exponential size, which can be checked in linear time
in the size of the state space and in exponential time in the size of the formula~\cite{DBLP:reference/mc/2018}.
\hfill\qed
\end{proof}

\section{Complete Experimental Results}
\label{appendix:com_ex_results}
In this section of the appendix, we present the complete experimental results from the benchmark families presented in \refSection{experimentalResults}. Table~\ref{table:benchmarks} is a snapshot of this table.

\setlength{\tabcolsep}{1pt}
\begin{longtable}{c!{\vline}c!{\vline}r!{\vline}rrr!{\vline}rrr!{\vline}rr!{\vline}!{\vline}rr!{\vline}!{\vline}c}
\caption{Experimental results from the benchmark families {Switch Failure} (SF) and Redundant Pipeline (RP), and the case study {Routing Update} (RU). The results are the average over five runs on an Intel i7-2700K CPU with 3.50~GHz, 32~GB RAM, and a timeout of 30~minutes. For the case study RU, results are only listed when at least the falsifier had no timeout.}\\
& & &\multicolumn{3}{c|}{PNwT} & \multicolumn{3}{c|}{Translated PN} & \multicolumn{2}{c||}{Circuit} & \multicolumn{2}{c||}{Result}\\
Ben. & Par. & \(\#S\) &\(|\mathscr{P}|\) & \(|\mathscr{T}|\) & \(|\varphi|\) & \(|\mathscr{P}^>|\)  & \(|\mathscr{T}^>|\)  & \(|\psi'|\)  & Lat.  & Gat.  & Sec. & Algo. & \(\models\)\\\hline
SF 
   &   3  & 4 & 4 & 5 & 35 & 17 & 22 & 60 & 90 & 2796 &  2.7 & IC3 & \cmark \\
   &   4  & 5 & 5 & 6 & 45 & 20 & 26 & 73 & 106 & 3580 &  8.5 & IC3 & \cmark \\
   &   5  & 6 & 6 & 7 & 55 & 23 & 30 & 86 & 122 & 4444 &  18.8 & IC3 & \cmark \\
   &   6  & 7 & 7 & 8 & 65 & 26 & 34 & 99 & 138 & 5388 &  49.8 & IC3 & \cmark \\
   &   7  & 8 & 8 & 9 & 75 & 29 & 38 & 112 & 154 & 6412 &  128.9 & IC3 & \cmark \\
   &   8  & 9 & 9 & 10 & 85 & 32 & 42 & 125 & 170 & 7516 & 291.9 & IC3 & \cmark \\
   &   9  & 10 & 10 & 11 & 95 & 35 & 46 & 138 & 186 & 8700 & 1359.9 & IC3 & \cmark \\
   &   10  & 11 & 11 & 12 & 105 & 38 & 50 & 151 & 202 & 9964 & TO & - & ? \\\hline 
RP 
  &    1/1/B   &  4  &  4 &  5  &  43  &  17  &  22  &  68  &  100  &  2989   &   4.0   & IC3 &   \cmark \\
  &    1/2/B   &  5  &  5 &  6  &  53  &  20  &  26  &  81  &  116  &  3773   &   11.0   & IC3 &   \cmark \\
  &    1/3/B   &  6  &  6 &  7  &  63  &  23  &  30  &  94  &  132  &  4637   &   16.4   & IC3 &   \cmark \\
  &    1/4/B   &  7  &  7 &  8  &  73  &  26  &  34  &  107  &  148  &  5581   &   66.6   & IC3 &   \cmark \\
  &    1/5/B   &  8  &  8 &  9  &  83  &  29  &  38  &  120  &  164  &  6605   &   99.3   & IC3 &   \cmark \\
\cdashline{2-14}
  &    2/1/B   &  5  &  5 &  6  &  53  &  20  &  26  &  81  &  116  &  3773   &   4.6   & IC3 &   \cmark \\
  &    2/2/B   &  6  &  6 &  7  &  63  &  23  &  30  &  94  &  132  &  4637   &   38.6   & IC3 &   \cmark \\
  &    2/3/B   &  7  &  7 &  8  &  73  &  26  &  34  &  107  &  148  &  5581   &   68.4   & IC3 &   \cmark \\
  &    2/4/B   &  8  &  8 &  9  &  83  &  29  &  38  &  120  &  164  &  6605   &   133.0   & IC3 &   \cmark \\
  &    2/5/B   &  9  &  9 &  10  &  93  &  32  &  42  &  133  &  180  &  7709   &   204.6   & IC3 &   \cmark \\
\cdashline{2-14}
  &    3/1/B   &  6  &  6 &  7  &  63  &  23  &  30  &  94  &  132  &  4637   &   1004.9   & IC3 &   \cmark \\
  &    3/2/B   &  7  &  7 &  8  &  73  &  26  &  34  &  107  &  148  &  5581   &   76.2   & IC3 &   \cmark \\
  &    3/3/B   &  8  &  8 &  9  &  83  &  29  &  38  &  120  &  164  &  6605   &   192.3   & IC3 &   \cmark \\
  &    3/4/B   &  9  &  9 &  10  &  93  &  32  &  42  &  133  &  180  &  7709   &   310.9   & IC3 &   \cmark \\
  &    3/5/B   &  10  &  10 &  11  &  103  &  35  &  46  &  146  &  196  &  8893   &   TO   & - &   ? \\
\cdashline{2-14}
  &    4/1/B   &  7  &  7 &  8  &  73  &  26  &  34  &  107  &  148  &  5581   &   77.4   & IC3 &   \cmark \\
  &    4/2/B   &  8  &  8 &  9  &  83  &  29  &  38  &  120  &  164  &  6605   &   139.7   & IC3 &   \cmark \\
  &    4/3/B   &  9  &  9 &  10  &  93  &  32  &  42  &  133  &  180  &  7709   &   594.7   & IC3 &   \cmark \\
  &    4/4/B   &  10  &  10 &  11  &  103  &  35  &  46  &  146  &  196  &  8893   &   646.4   & IC3 &   \cmark \\
  &    4/5/B   &  11  &  11 &  12  &  113  &  38  &  50  &  159  &  212  &  10157   &   TO   & - &   ? \\
\cdashline{2-14}
  &    5/1/B   &  8  &  8 &  9  &  83  &  29  &  38  &  120  &  164  &  6605   &   126.9   & IC3 &   \cmark \\
  &    5/2/B   &  9  &  9 &  10  &  93  &  32  &  42  &  133  &  180  &  7709   &   580.0   & IC3 &   \cmark \\
  &    5/3/B   &  10  &  10 &  11  &  103  &  35  &  46  &  146  &  196  &  8893   &   TO   & - &   ? \\
  &    5/4/B   &  11  &  11 &  12  &  113  &  38  &  50  &  159  &  212  &  10157   &   TO   & - &   ? \\
  &    5/5/B   &  12  &  12 &  13  &  123  &  41  &  54  &  172  &  228  &  11501   &   TO   & - &   ? \\
\cline{2-14}
&   1/1/U   & 6 & 6 & 9 & 63 & 25 & 36 & 100 & 136 & 5535 &   1.6   & BMC2 &   \xmark \\
&   1/2/U   & 7 & 7 & 10 & 73 & 28 & 40 & 113 & 152 & 6579 &   2.7   & BMC3 &   \xmark \\
&   1/3/U   & 8 & 8 & 11 & 83 & 31 & 44 & 126 & 168 & 7703 &   6.3   & BMC2 &   \xmark \\
&   1/4/U   & 9 & 9 & 12 & 93 & 34 & 48 & 139 & 184 & 8907 &   21.1   & BMC3 &   \xmark \\
&   1/5/U   & 10 & 10 & 13 & 103 & 37 & 52 & 152 & 200 & 10191 &   30.6   & BMC2 &   \xmark \\
\cdashline{2-14}
&   2/1/U   & 7 & 7 & 10 & 73 & 28 & 40 & 113 & 152 & 6579 &   2.3   & BMC3 &   \xmark \\
&   2/2/U   & 8 & 8 & 11 & 83 & 31 & 44 & 126 & 168 & 7703 &   4.2   & BMC3 &   \xmark \\
&   2/3/U   & 9 & 9 & 12 & 93 & 34 & 48 & 139 & 184 & 8907 &   18.4   & BMC3 &   \xmark \\
&   2/4/U   & 10 & 10 & 13 & 103 & 37 & 52 & 152 & 200 & 10191 &   23.7   & BMC &   \xmark \\
&   2/5/U   & 11 & 11 & 14 & 113 & 40 & 56 & 165 & 216 & 11555 &   105.2   & BMC2 &   \xmark \\
\cdashline{2-14}
&   3/1/U   & 8 & 8 & 11 & 83 & 31 & 44 & 126 & 168 & 7703 &   8.6   & BMC2 &   \xmark \\
&   3/2/U   & 9 & 9 & 12 & 93 & 34 & 48 & 139 & 184 & 8907 &   14.4   & BMC2 &   \xmark \\
&   3/3/U   & 10 & 10 & 13 & 103 & 37 & 52 & 152 & 200 & 10191 &   22.6   & BMC3 &   \xmark \\
&   3/4/U   & 11 & 11 & 14 & 113 & 40 & 56 & 165 & 216 & 11555 &   70.4   & BMC2 &   \xmark \\
&   3/5/U   & 12 & 12 & 15 & 123 & 43 & 60 & 178 & 232 & 12999 &   180.9   & BMC3 &   \xmark \\
\cdashline{2-14}
&   4/1/U   & 9 & 9 & 12 & 93 & 34 & 48 & 139 & 184 & 8907 &   15.2   & BMC3 &   \xmark \\
&   4/2/U   & 10 & 10 & 13 & 103 & 37 & 52 & 152 & 200 & 10191 &   49.0   & BMC3 &   \xmark \\
&   4/3/U   & 11 & 11 & 14 & 113 & 40 & 56 & 165 & 216 & 11555 &   75.6   & BMC3 &   \xmark \\
&   4/4/U   & 12 & 12 & 15 & 123 & 43 & 60 & 178 & 232 & 12999 &   144.4   & BMC3 &   \xmark \\
&   4/5/U   & 13 & 13 & 16 & 133 & 46 & 64 & 191 & 248 & 14523 &   883.7   & BMC2 &   \xmark \\
\cdashline{2-14}
&   5/1/U   & 10 & 10 & 13 & 103 & 37 & 52 & 152 & 200 & 10191 &   38.4   & BMC2 &   \xmark \\
&   5/2/U   & 11 & 11 & 14 & 113 & 40 & 56 & 165 & 216 & 11555 &   87.3   & BMC3 &   \xmark \\
&   5/3/U   & 12 & 12 & 15 & 123 & 43 & 60 & 178 & 232 & 12999 &   182.8   & BMC3 &   \xmark \\
&   5/4/U   & 13 & 13 & 16 & 133 & 46 & 64 & 191 & 248 & 14523 &   945.1   & BMC3 &   \xmark \\
&   5/5/U   & 14 & 14 & 17 & 143 & 49 & 68 & 204 & 264 & 16127 &   TO   & - &   ? \\
\cline{2-14}
&   1/1/M  & 6 & 9 & 11 & 63 & 30 & 42 & 106 & 146 & 6908 &   8.1   & BMC3 &   \xmark \\
&   1/2/M  & 7 & 10 & 12 & 73 & 33 & 46 & 119 & 162 & 8081 &   15.4   & BMC2 &   \xmark \\
&   1/3/M  & 8 & 11 & 13 & 83 & 36 & 50 & 132 & 178 & 9334 &   59.1   & BMC3 &   \xmark \\
&   1/4/M  & 9 & 12 & 14 & 93 & 39 & 54 & 145 & 194 & 10667 &   128.5   & BMC3 &   \xmark \\
&   1/5/M  & 10 & 13 & 15 & 103 & 42 & 58 & 158 & 210 & 12080 &   373.2   & BMC &   \xmark \\
\cdashline{2-14}
&   2/1/M  & 7 & 10 & 12 & 73 & 33 & 46 & 119 & 162 & 8081 &   11.7   & BMC3 &   \xmark \\
&   2/2/M  & 8 & 11 & 13 & 83 & 36 & 50 & 132 & 178 & 9334 &   45.1   & BMC3 &   \xmark \\
&   2/3/M  & 9 & 12 & 14 & 93 & 39 & 54 & 145 & 194 & 10667 &   149.5   & BMC2 &   \xmark \\
&   2/4/M  & 10 & 13 & 15 & 103 & 42 & 58 & 158 & 210 & 12080 &   457.5   & BMC &   \xmark \\
&   2/5/M  & 11 & 14 & 16 & 113 & 45 & 62 & 171 & 226 & 13573 &   TO   & - &   ? \\
\cdashline{2-14}
&   3/1/M  & 8 & 11 & 13 & 83 & 36 & 50 & 132 & 178 & 9334 &   47.3   & BMC3 &   \xmark \\
&   3/2/M  & 9 & 12 & 14 & 93 & 39 & 54 & 145 & 194 & 10667 &   129.9   & BMC2 &   \xmark \\
&   3/3/M  & 10 & 13 & 15 & 103 & 42 & 58 & 158 & 210 & 12080 &   347.0   & BMC2 &   \xmark \\
&   3/4/M  & 11 & 14 & 16 & 113 & 45 & 62 & 171 & 226 & 13573 &   1433.3   & BMC3 &   \xmark \\
&   3/5/M  & 12 & 15 & 17 & 123 & 48 & 66 & 184 & 242 & 15146 &   TO   & - &   ? \\
\cdashline{2-14}
&   4/1/M  & 9 & 12 & 14 & 93 & 39 & 54 & 145 & 194 & 10667 &   121.6   & BMC3 &   \xmark \\
&   4/2/M  & 10 & 13 & 15 & 103 & 42 & 58 & 158 & 210 & 12080 &   370.6   & BMC2 &   \xmark \\
&   4/3/M  & 11 & 14 & 16 & 113 & 45 & 62 & 171 & 226 & 13573 &   1449.6   & BMC2 &   \xmark \\
&   4/4/M  & 12 & 15 & 17 & 123 & 48 & 66 & 184 & 242 & 15146 &   TO   & - &   ? \\
&   4/5/M  & 13 & 16 & 18 & 133 & 51 & 70 & 197 & 258 & 16799 &   8.8   & BMC3 &   \xmark \\
\cdashline{2-14}
&   5/1/M  & 10 & 13 & 15 & 103 & 42 & 58 & 158 & 210 & 12080 &   462.7   & BMC2 &   \xmark \\
&   5/2/M  & 11 & 14 & 16 & 113 & 45 & 62 & 171 & 226 & 13573 &   1740.1   & BMC2 &   \xmark \\
&   5/3/M  & 12 & 15 & 17 & 123 & 48 & 66 & 184 & 242 & 15146 &   TO   & - &   ? \\
&   5/4/M  & 13 & 16 & 18 & 133 & 51 & 70 & 197 & 258 & 16799 &   TO   & - &   ? \\
&   5/5/M  & 14 & 17 & 19 & 143 & 54 & 74 & 210 & 274 & 18532 &   TO   & - &   ? \\
\cline{2-14}
  	&    1/1/C   & 6 &  9  &  11  &  70  &  30  &  42  &  113  &  151  &  7023   &   63.1   &    IC3 &   \cmark \\
  	&    1/2/C   & 7 &  10  &  12  &  80  &  33  &  46  &  126  &  167  &  8196   &   183.5   &    IC3 &   \cmark \\
  	&    1/3/C   & 8 &  11  &  13  &  90  &  36  &  50  &  139  &  183  &  9449   &   301.7   &    IC3 &   \cmark \\
  	&    1/4/C   & 9 &  12  &  14  &  100  &  39  &  54  &  152  &  199  &  10782   &   565.8   &    IC3 &   \cmark \\
  	&    1/5/C   & 10&  13  &  15  &  110  &  42  &  58  &  165  &  215  &  12195   &   TO   &    - &   ? \\
\cdashline{2-14}
  	&    2/1/C   & 7 &  10  &  12  &  80  &  33  &  46  &  126  &  167  &  8196   &   105.9   &    IC3 &   \cmark \\
  	&    2/2/C   & 8 &  11  &  13  &  90  &  36  &  50  &  139  &  183  &  9449   &   693.5   &    IC3 &   \cmark \\
  	&    2/3/C   & 9 &  12  &  14  &  100  &  39  &  54  &  152  &  199  &  10782   &   689.3   &    IC3 &   \cmark \\
  	&    2/4/C   & 10&  13  &  15  &  110  &  42  &  58  &  165  &  215  &  12195   &   TO   &    - &   ? \\
  	&    2/5/C   & 11&  14  &  16  &  120  &  45  &  62  &  178  &  231  &  13688   &   885.8   &    IC3 &   \cmark \\
\cdashline{2-14}
  	&    3/1/C   & 8 &  11  &  13  &  90  &  36  &  50  &  139  &  183  &  9449   &   451.5   &    IC3 &   \cmark \\
  	&    3/2/C   & 9 &  12  &  14  &  100  &  39  &  54  &  152  &  199  &  10782   &   1252.4   &    IC3 &   \cmark \\
  	&    3/3/C   & 10&  13  &  15  &  110  &  42  &  58  &  165  &  215  &  12195   &   1218.0   &    IC3 &   \cmark \\
  	&    3/4/C   & 11&  14  &  16  &  120  &  45  &  62  &  178  &  231  &  13688   &    TO   &    - &   ? \\
  	&    3/5/C   & 12&  15  &  17  &  130  &  48  &  66  &  191  &  247  &  15261   &    TO   &    - &   ? \\
\cdashline{2-14}
    &    4/1/C   & 9 &  12  &  14  &  100  &  39  &  54  &  152  &  199  &  10782   &   1717.1   &    IC3 &   \cmark \\
  	&    4/2/C   & 10&  13  &  15  &  110  &  42  &  58  &  165  &  215  &  12195   &   775.1   &    IC3 &   \cmark \\
  	&    4/3/C   & 11&  14  &  16  &  120  &  45  &  62  &  178  &  231  &  13688   &    TO   &    - &   ? \\
  	&    4/4/C   & 12&  15  &  17  &  130  &  48  &  66  &  191  &  247  &  15261   &    TO   &    - &   ? \\
  	&    4/5/C   & 13&  16  &  18  &  140  &  51  &  70  &  204  &  263  &  16914   &    TO   &    - &   ? \\
\cdashline{2-14}
  	&    5/1/C   & 10&  13  &  15  &  110  &  42  &  58  &  165  &  215  &  12195   &    TO   &    - &   ? \\
  	&    5/2/C   & 11&  14  &  16  &  120  &  45  &  62  &  178  &  231  &  13688   &    1641.3   &    IC3 &   \cmark \\
  	&    5/3/C   & 12&  15  &  17  &  130  &  48  &  66  &  191  &  247  &  15261   &    TO   &    - &   ? \\
  	&    5/4/C   & 13&  16  &  18  &  140  &  51  &  70  &  204  &  263  &  16914   &    TO   &    - &   ? \\
  	&    5/5/C   & 14&  17  &  19  &  150  &  54  &  74  &  217  &  279  &  18647   &    TO   &    - &   ? \\
\cline{1-14}
RU
  &    Arpanet196912T  & 4  & 14 & 10 & 117 & 31 & 39 & 154 & 188 & 7483  &   22.7   & IC3  &  \cmark \\ 
  &    Arpanet196912F  & 4  & 14 & 10 & 117 & 31 & 39 & 154 & 188 & 7483  &   2.0    & BMC3 &   \xmark \\
  &    NapnetT         & 6  & 23 & 17 & 199 & 48 & 64 & 254 & 292 & 15875 &   95.1  & IC3 &   \cmark \\ 
  &    NapnetF         & 6  & 23 & 17 & 199 & 48 & 64 & 254 & 292 & 15875 &   4.7   & BMC3 &   \xmark \\ 
  &    EpochT          & 6  & 23 & 17 & 199 & 48 & 64 & 254 & 292 & 15875 &   169.2   & IC3 &   \cmark \\ 
  &    EpochF          & 6  & 23 & 17 & 199 & 48 & 64 & 254 & 292 & 15875 &   18.6   & BMC3&   \xmark \\
  &    TelecomserbiaT  & 6  & 24 & 18 & 205 & 47 & 59 & 254 & 300 & 14779 & 492.3 & IC3 & \cmark \\ 
  &    TelecomserbiaF  & 6  & 24 & 18 & 205 & 47 & 59 & 254 & 300 & 14779 & 51.2 & BMC3 & \xmark \\ 
  &    Layer42T        & 6  & 24 & 18 & 209 & 49 & 65 & 264 & 304 & 16465 & 597.4  & IC3 &   \cmark \\ 
  &    Layer42F        & 6  & 24 & 18 & 209 & 49 & 65 & 264 & 304 & 16465 &   14.6   & BMC3 &   \xmark \\
  &    SanrenT         & 7  & 24 & 17 & 199 & 50 & 64 & 254 & 296 & 16177 & 1353.9  & IC3 &   \cmark \\ 
  &    SanrenF         & 7  & 24 & 17 & 199 & 50 & 64 & 254 & 296 & 16177 &   137.0   & BMC3 &   \xmark \\
  &    GetnetT         & 7  & 25 & 18 & 213 & 53 & 71 & 274 & 312 & 18601 & 1418.8  & IC3 &   \cmark \\ 
  &    GetnetF         & 7  & 25 & 18 & 213 & 53 & 71 & 274 & 312 & 18601 &   13.3   & BMC3 &   \xmark \\
  &    NetrailT        & 7  & 30 & 23 & 271 & 62 & 88 & 344 & 380 & 26101 & 145.3 & IC3 &   \cmark \\ 
  &    NetrailF        & 7  & 30 & 23 & 271 & 62 & 88 & 344 & 380 & 26101 &  58.3 & BMC3 &   \xmark \\
  &    Arpanet19706T  & 9  & 33 & 24 & 281 & 67 & 89 & 354 & 400 & 27619 & 507.8 & IC3 &   \cmark \\ 
  &    Arpanet19706F  & 9  & 33 & 24 & 281 & 67 & 89 & 354 & 400 & 27619 & 49.7 & BMC3 &   \xmark \\
  &    NsfcnetT        & 10 & 31 & 22 & 261 & 65 & 87 & 334 & 376 & 26181 &    304.8   & IC3 &  \cmark  \\ 
  &    NsfcnetF        & 10 & 31 & 22 & 261 & 65 & 87 & 334 & 376 & 26181 &   8.4   & BMC3 &   \xmark \\
\cline{2-14}
  &    HeanetT         & 7  & 33 & 26 & 305 & 67 & 97 & 384 & 420 & 30675 &    TO   & - &    ? \\ 
  &    HeanetF         & 7  & 33 & 26 & 305 & 67 & 97 & 384 & 420 & 30675 & 76.8 & BMC3 &   \xmark \\
  &    Nordu2005T      & 9  & 30 & 21 & 247 & 62 & 80 & 314 & 360 & 23269 & TO & - &   ? \\ 
  &    Nordu2005F      & 9  & 30 & 21 & 247 & 62 & 80 & 314 & 360 & 23269 & 48.5 & BMC3 &   \xmark \\
  &    ClaranetT       & 15 & 55 & 40 & 473 & 111 & 153 & 594 & 648 & 70127 &    TO   & - &    ? \\ 
  &    ClaranetF       & 15 & 55 & 40 & 473 & 111 & 153 & 594 & 648 & 70127 & 213.8 & BMC3 &   \xmark \\
  &    Garr199901T     & 16 & 56 & 40 & 473 & 113 & 153 & 594 & 652 & 70785 &    TO   & - &    ? \\ 
  &    Garr199901F     & 16 & 56 & 40 & 473 & 113 & 153 & 594 & 652 & 70785 & 795.4 & BMC3 &   \xmark \\
  &    FatmanT         & 17 & 62 & 45 & 535 & 126 & 176 & 674 & 728 & 89701 &    TO   & - &    ? \\ 
  &    FatmanF         & 17 & 62 & 45 & 535 & 126 & 176 & 674 & 728 & 89701 & 139.3 & BMC3 &   \xmark \\
  &    Nordu2010T      & 18 & 56 & 38 & 449 & 106 & 133 & 552 & 618 & 57950 &    TO   & - &    ? \\ 
  &    Nordu2010F      & 18 & 56 & 38 & 449 & 106 & 133 & 552 & 618 & 57950 & 136.1 & BMC2 &   \xmark \\
  &    PacificWaveT    & 18 & 68 & 50 & 589 & 135 & 187 & 734 & 796 & 101171 &    TO   & - &    ? \\ 
  &    PacificWaveF    & 18 & 68 & 50 & 589 & 135 & 187 & 734 & 796 & 101171 & 1707.4 & BMC2 &   \xmark \\
  &    TwarenT         & 20 & 65 & 45 & 531 & 130 & 170 & 664 & 736 & 87493 &    TO   & - &    ? \\ 
  &    TwarenF         & 20 & 65 & 45 & 531 & 130 & 170 & 664 & 736 & 87493 & 461.5 & BMC3 &   \xmark \\
  &    MarnetT         & 20 & 77 & 57 & 679 & 156 & 224 & 854 & 908 & 138103 &    TO   & - &    ? \\ 
  &    MarnetF         & 20 & 77 & 57 & 679 & 156 & 224 & 854 & 908 & 138103 & 746.1 & BMC3 &   \xmark \\
  &    JanetlenseT     & 20 & 91 & 71 & 847 & 184 & 280 & 1064 & 1104 & 203595 &    TO   & - &    ? \\ 
  &    JanetlenseF     & 20 & 91 & 71 & 847 & 184 & 280 & 1064 & 1104 & 203595 & 514.2 & BMC2 &   \xmark \\
  &    HarnetT         & 21 & 71 & 50 & 593 & 143 & 193 & 744 & 812 & 108415 &    TO   & - &    ? \\ 
  &    HarnetF         & 21 & 71 & 50 & 593 & 143 & 193 & 744 & 812 & 108415 & 919.0 & BMC3 &   \xmark \\
  &    Belnet2009T     & 21 & 71 & 50 & 597 & 145 & 199 & 754 & 816 & 113397 &    TO   & - &    ? \\ 
  &    Belnet2009F     & 21 & 71 & 50 & 597 & 145 & 199 & 754 & 816 & 113397 & 1163.3 & BMC2 &   \xmark \\
  &    Garr200404T     & 22 & 73 & 51 & 607 & 148 & 200 & 764 & 832 & 115567 &    TO   & - &    ? \\ 
  &    Garr200404F     & 22 & 73 & 51 & 607 & 148 & 200 & 764 & 832 & 115567 &  79.4 & BMC3 &   \xmark \\
  &    IstarT          & 23 & 73 & 50 & 593 & 147 & 193 & 744 & 820 & 110051 &    TO   & - &    ? \\ 
  &    IstarF          & 23 & 73 & 50 & 593 & 147 & 193 & 744 & 820 & 110051 & 276.6 & BMC3 &   \xmark \\
  &    Garr199905T     & 23 & 77 & 54 & 641 & 155 & 209 & 804 & 876 & 125707 &    TO   & - &    ? \\ 
  &    Garr199905F     & 23 & 77 & 54 & 641 & 155 & 209 & 804 & 876 & 125707 & 157.7 & BMC3 &   \xmark \\
  &    Garr199904T     & 23 & 77 & 54 & 641 & 155 & 209 & 804 & 876 & 125707 &    TO   & - &    ? \\ 
  &    Garr199904F     & 23 & 77 & 54 & 641 & 155 & 209 & 804 & 876 & 125707 & 460.2 & BMC2 &   \xmark \\
  &    UranT           & 24 & 56 & 38 & 449 & 106 & 133 & 552 & 618 & 57950 &    TO   & - &    ? \\ 
  &    UranF           & 24 & 56 & 38 & 449 & 106 & 133 & 552 & 618 & 57950 & 143.2 & BMC3 &   \xmark \\
  &    KentmanFeb2008T & 26 & 82 & 56 & 669 & 167 & 223 & 844 & 920 & 142291 &    TO   & - &    ? \\ 
  &    KentmanFeb2008F & 26 & 82 & 56 & 669 & 167 & 223 & 844 & 920 & 142291 &  111.2 & BMC3 &   \xmark \\
  &    Garr200212T     & 27 & 86 & 59 & 703 & 174 & 232 & 884 & 964 & 153509 &    TO   & - &    ? \\ 
  &    Garr200212F     & 27 & 86 & 59 & 703 & 174 & 232 & 884 & 964 & 153509 & 324.2 & BMC3 &   \xmark \\
  &    IinetT          & 31 & 104 & 73 & 871 & 210 & 288 & 1094 & 1176 & 227153 &    TO   & - &    ? \\ 
  &    IinetF          & 31 & 104 & 73 & 871 & 210 & 288 & 1094 & 1176 & 227153 & 1244.5 & BMC3 &   \xmark \\
  &    KentmanJan2011T & 38 & 117 & 79 & 943 & 236 & 312 & 1184 & 1288 & 269943 &    TO   & - &    ? \\ 
  &    KentmanJan2011F & 38 & 117 & 79 & 943 & 236 & 312 & 1184 & 1288 & 269943 &  112.6 & BMC3 &   \xmark \\
\cline{2-13}
 \end{longtable}

\end{document}